\documentclass{article}
\usepackage[left=3cm,right=3cm,top=3cm,bottom=2cm]{geometry}
\usepackage[colorlinks]{hyperref}
\hypersetup{linkcolor=blue, citecolor=blue}
\usepackage{amsmath, enumitem, amssymb, amsthm, bm, mathrsfs, braket, graphicx, tikz, comment, mathtools, cleveref}

\title{log-Coulomb gas with norm-density in $p$-fields}
\author{Joe Webster}
\date{\today}

\newcommand\R{\mathbb{R}}
\newcommand\Q{\mathbb{Q}}
\newcommand\N{\mathbb{N}}
\newcommand\Z{\mathbb{Z}}
\newcommand\C{\mathbb{C}}

\newcommand\E{\mathbb{E}}
\newcommand\F{\mathbb{F}}

\newcommand\spl{{\bm{\pitchfork}}}
\newcommand\ptn{{\pitchfork}}
\DeclareMathOperator\re{Re}

\DeclareMathOperator\charr{char}

\DeclareMathOperator\rank{rank}

\DeclareMathOperator\supp{supp}
\DeclareMathOperator\Sym{Sym}
\DeclareMathOperator\Aut{Aut}
\DeclareMathOperator\Orb{Orb}
\DeclareMathOperator\Stab{Stab}

\theoremstyle{definition}
\newtheorem{theorem}{Theorem}[section]
\newtheorem{proposition}[theorem]{Proposition}
\newtheorem{lemma}[theorem]{Lemma}
\newtheorem{definition}[theorem]{Definition}
\newtheorem{remark}[theorem]{Remark}
\newtheorem{corollary}[theorem]{Corollary}
\newtheorem{example}[theorem]{Example}

\numberwithin{equation}{subsection}

\begin{document}

\maketitle

\begin{abstract}
The main result of this paper is a formula for the integral $$\int_{K^N}\rho(x)\big(\max_{i<j}|x_i-x_j|\big)^a\big(\min_{i<j}|x_i-x_j|\big)^b\prod_{i<j}|x_i-x_j|^{s_{ij}}\,|dx|~,$$ where $K$ is a $p$-field (i.e., a nonarchimedean local field) with canonical absolute value $|\cdot|$, $N\geq 2$, $a,b\in\C$, the function $\rho:K^N\to\C$ has mild growth and decay conditions and factors through the norm $\|x\|=\max_i|x_i|$, and $|dx|$ is the usual Haar measure on $K^N$. The formula is a finite sum of functions described explicitly by combinatorial data, and the largest open domain of values $(s_{ij})_{i<j}\in\C^{\binom{N}{2}}$ on which the integral converges absolutely is given explicitly in terms of these data and the parameters $a$, $b$, $N$, and $K$. We then specialize the formula to $s_{ij}=\mathfrak{q}_i\mathfrak{q}_j\beta$, where $\mathfrak{q}_1,\mathfrak{q}_2,\dots,\mathfrak{q}_N>0$ represent the charges of an $N$-particle log-Coulomb gas in $K$ with background density $\rho$ and inverse temperature $\beta$. From this specialization we obtain a mixed-charge $p$-field analogue of Mehta's integral formula, as well as formulas and low-temperature limits for the joint moments of $\max_{i<j}|x_i-x_j|$ (the diameter of the gas) and $\min_{i<j}|x_i-x_j|$ (the minimum distance between its particles).
\end{abstract}

\let\thefootnote\relax 

\newpage
\tableofcontents

\newpage
\section{Introduction}

\subsection{log-Coulomb gas in local fields}\label{1_1}

A topological field $K$ is called a \emph{local field} if it is Hausdorff, non-discrete, and locally compact. As discussed in \cite{Weil}, every such field admits an additive Haar measure $\mu$ which is unique up to normalization. Given a measurable set $M\subset K$ with $0<\mu(M)<\infty$, it can be shown that the function $|\cdot|:K\to\R_{\geq 0}$ defined by $$|x|:=\begin{cases}\sqrt{\mu(xM)/\mu(M)}&\text{if $K\cong\C$},\\\mu(xM)/\mu(M)&\text{otherwise},\end{cases}$$ satisfies the axioms of an absolute value on $K$. In fact, $|\cdot|$ is independent of $M$ and the normalization of $\mu$, the metric topology generated by $|\cdot|$ coincides with the intrinsic topology on $K$, and $K$ is complete with respect to $|\cdot|$. Thus we call $|\cdot|$ the \emph{canonical absolute value} on $K$, denote the closed and open unit balls respectively by $$R:=\{x\in K:|x|\leq 1\}\qquad\text{and}\qquad P:=\{x\in K:|x|<1\}~,$$ and fix a normalization of $\mu$ once and for all by declaring $$\mu(R):=\begin{cases}\pi&\text{if }K\cong\C,\\2&\text{if }K\cong\R,\\1&\text{otherwise.}\end{cases}$$ Given a local field $K$, we henceforth reserve the symbols $|\cdot|$, $R$, $P$, and $\mu$ for the items defined above. Following conventions in \cite{DenefI}, when the integer $N\geq 1$ is understood we write generic elements of the $N$-fold product $K^N$ as $x=(x_1,x_2,\dots,x_N)$ and denote the polynomial ring $K[x_1,x_2,\dots,x_N]$ simply by $K[x]$. We also reserve $\|\cdot\|$ for the standard norm on $K^N$, which is defined by $$\|x\|:=\begin{cases}\sqrt{\sum_{i=1}^N|x_i|^2}&\text{if $K\cong\R$ or $K\cong\C$},\\\displaystyle{\max_{1\leq i\leq N}|x_i|}&\text{otherwise.}\end{cases}$$ This norm makes $K^N$ into a locally compact vector space on which $\mu^N$ is a Haar measure, and we shall write $|dx|$ when integrating against $\mu^N$. Following the setup for $K=\R$ given in \cite{Forrester}, we may define log-Coulomb gas in an arbitrary local field $K$.

\begin{definition}
Let $\mathfrak{q}_1,\mathfrak{q}_2,\dots,\mathfrak{q}_N>0$ be fixed charge magnitudes associated to particles with respective random locations $x_1,x_2,\dots,x_N\in K$. Let $\beta>0$ denote the inverse temperature of the system and choose a nonnegative measurable function $\rho$ on $K^N$ such that $$\mathcal{Z}_N(\beta):=\int_{K^N}\rho(x)\prod_{i<j}|x_i-x_j|^{\mathfrak{q}_i\mathfrak{q}_j\beta}\,|dx|$$ is positive and finite for all $\beta>0$. The system is called a \emph{log-Coulomb gas} if, given $\beta>0$, the vectors $x=(x_1,x_2,\dots,x_N)\in K^N$ have probability density $\frac{1}{\mathcal{Z}_N(\beta)}\rho(x)\prod_{i<j}|x_i-x_j|^{\mathfrak{q}_i\mathfrak{q}_j\beta}\,|dx|$. In this case the vectors $x\in K^N$ are called \emph{microstates} of the system, $\rho$ is called the \emph{background density}, $\mathcal{Z}_N$ is called the \emph{canonical partition function}, and the number of distinct values in $\{\mathfrak{q}_1,\mathfrak{q}_2,\dots,\mathfrak{q}_N\}$ is called the number of \emph{components} of the gas.
\end{definition}

The function $\rho$ should be selected to have fast decay (say, sub-exponential) as $\|x\|\to\infty$ if $\mathcal{Z}_N(\beta)$ is to be finite, i.e., $\rho$ should impose a potential well that keeps the charges from scattering to infinity. We will further assume that $\rho$ is a \emph{norm-density}, meaning it factors through the standard norm $\|\cdot\|:K^N\to\R_{\geq 0}$, and henceforth regard $\rho$ as a function on $\|K^N\|$ instead of $K^N$. On the other hand, the factor $\prod_{i<j}|x_i-x_j|^{\mathfrak{q}_i\mathfrak{q}_j\beta}$ increases with each particle pair distance $|x_i-x_j|$, so mutual repulsion between particles is probabilistically favored. This repulsion is favored more if the gas is cold (i.e., $\beta\gg0$) and less if the gas is hot (i.e., $\beta\approx 0$). Thus microstates $x\in K^N$ satisfying $\min_{i<j}|x_i-x_j|\gg 0$ have high probability if the gas' total energy has little fluctuation (i.e., the gas is cold), while microstates distribute more uniformly throughout the potential well if the energy is allowed larger fluctuations (i.e., the gas is hot). The precise variations of the microstate probability densities with $\beta$ are governed by $\mathcal{Z}_N$, and hence finding an explicit formula for $\mathcal{Z}_N(\beta)$ is a central problem in the study of log-Coulomb gases. 

In the mid-1960's Mehta and Dyson showed that the joint probability density functions of the eigenvalues $x_1,x_2,\dots,x_N\in\R$ for $N\times N$ Gaussian orthogonal, unitary, and real-quaternion matrix ensembles are respectively $$\frac{1}{\mathcal{Z}_N(1)}\rho(\|x\|)\prod_{i<j}|x_i-x_j|,\quad\frac{1}{\mathcal{Z}_N(2)}\rho(\|x\|)\prod_{i<j}|x_i-x_j|^2,\quad\text{and}\quad\frac{1}{\mathcal{Z}_N(4)}\rho(\|x\|)\prod_{i<j}|x_i-x_j|^4~,$$ where $\rho(t)=e^{-\frac{t^2}{2}}$ for all $t\in\|\R^N\|=\R_{\geq 0}$. That is, the eigenvalues form a real one-component log-Coulomb gas in $K=\R$ with charges $\mathfrak{q}_1=\mathfrak{q}_2=\dots=\mathfrak{q}_N=1$, Gaussian background-density, and inverse temperature $1$, $2$, or $4$. Explicit computations of $\mathcal{Z}_N(1)$, $\mathcal{Z}_N(2)$, and $\mathcal{Z}_N(4)$ led Mehta and Dyson to conjecture the following formula, which is known today as \emph{Mehta's integral formula}:

\begin{theorem}[\cite{FW}]\label{MIF}
If $\beta$ is any complex number with $\re(\beta)>-\frac{2}{N}$, then $$\mathcal{Z}_N(\beta)=\int_{\R^N}e^{-\frac{\|x\|^2}{2}}\prod_{i<j}|x_i-x_j|^\beta\,|dx|=(2\pi)^{N/2}\prod_{j=2}^N\frac{\Gamma(1+\frac{j\beta}{2})}{\Gamma(1+\frac{\beta}{2})}~.$$
\end{theorem}

By a clever application of Selberg's integral formula, Bombieri found the first proof of \Cref{MIF} a decade after its conjecture. His proof, several others, and the related random matrix theory can be found in \cite{Forrester}. \Cref{MIF} does not generalize easily to multi-component ensembles, however. Multi-component analogues were established in \cite{Chris12} for a large class of integer-valued $\beta$ and the $\{\mathfrak{q}_1,\mathfrak{q}_2,\dots,\mathfrak{q}_N\}=\{1,2\}$ case was thoroughly explored in \cite{Chris13}, but a general multi-component analogue of \Cref{MIF} remains unknown for $K=\R$.

In this paper we will find explicit combinatorial formulas for multi-component (i.e., mixed charge) canonical partition functions for any $K\not\cong\R,\C$, and in this setting we will compute the joint moments of $\max_{i<j}|x_i-x_j|$ (the diameter of the gas) and $\min_{i<j}|x_i-x_j|$ (the minimum distance between charges) and their low temperature limits. All of these computations follow from our main theorem, which establishes a formula for the following integral:

\begin{definition}\label{ZNKabs}
For a local field $K$, an integer $N\geq 2$, a measurable function $\rho:\|K^N\|\to\C$, complex numbers $a,b\in\C$, and suitable $\bm{s}=(s_{ij})_{1\leq i<j\leq N}\in\C^{\binom{N}{2}}$, define $$Z_N^\rho(K,a,b,\bm{s}):=\int_{K^N}\rho(\|x\|)\big(\max_{i<j}|x_i-x_j|\big)^a\big(\min_{i<j}|x_i-x_j|\big)^b\prod_{i<j}|x_i-x_j|^{s_{ij}}\,|dx|~.$$
\end{definition}

Given $s_{ij}=\mathfrak{q}_i\mathfrak{q}_j\beta$ for all $i<j$ with suitable $\beta$ and norm-density $\rho$, the canonical partition function satisfies $\mathcal{Z}_N(\beta)=Z_N^\rho(K,0,0,\bm{s})\in(0,\infty)$. Then for such $\beta$ and $\rho$, the expectation of the random variable $\big(\max_{i<j}|x_i-x_j|\big)^a\big(\min_{i<j}|x_i-x_j|\big)^b$ against $\frac{1}{\mathcal{Z}_N(\beta)}\rho(\|x\|)\prod_{i<j}|x_i-x_j|^{\mathfrak{q}_i\mathfrak{q}_j\beta}$ is given by
\begin{equation}\label{Jointeq}
\E\left[\big(\max_{i<j}|x_i-x_j|\big)^a\big(\min_{i<j}|x_i-x_j|\big)^b\right]=\frac{Z_N^\rho(K,a,b,\bm{s})}{Z_N^\rho(K,0,0,\bm{s})}~.
\end{equation}
In particular, taking $a,b\in\Z_{\geq 0}$ in \eqref{Jointeq} yields joint moments of $\max_{i<j}|x_i-x_j|$ and $\min_{i<j}|x_i-x_j|$. To our knowledge, these canonical partition functions and joint moments have not been computed for $K\not\cong\R,\C$ before. However, the study of Coulomb gases (which include log-Coulomb gases) in such $K$, or more generally $K^d$ with $d\geq 1$, has become increasingly active in the last decade. Following the classical $\R^d$ analogue in \cite{Serfaty}, the article \cite{ZT} gives Coulomb gases in $K^d$ a natural motivation by realizing their potentials as fundamental solutions to the pseudodifferential analogue of Poisson's electrostatic equation. The same article also realizes the Haar measure $\mu^d$ (restricted to $R^d$) as the equilibrium measure for the confining potential $V(x)=-\log\mathbf{1}_{[0,1]}(\|x\|)$ in the non-log-Coulomb case. In the log-Coulomb case, the more recent article \cite{ZZC} uses graph-theoretic machinery to establish and analyze formulas for $\mathcal{Z}_N(\beta)=Z_N^\rho(K,0,0,\bm{s})$ with $\rho=\mathbf{1}_{[0,1]}$ and $s_{ij}=\mathfrak{q}_i\mathfrak{q}_j\beta$, with arbitrary $\mathfrak{q}_1,\mathfrak{q}_2,\dots,\mathfrak{q}_N\in\R$. Our present results are closely related to \cite{ZZC} but were discovered by ``unwinding" a recurrence for $\mathcal{Z}_N(\beta)$ established in \cite{ChrisA}. Using the recurrence and the assumption $\mathfrak{q}_1=\mathfrak{q}_2=\dots=\mathfrak{q}_N=1$, the latter article reveals strikingly simple algebraic properties and high/low temperature limits for the corresponding \emph{grand canonical partition function} $Z(t,\beta):=\sum_{N=0}^\infty Z_N(\beta)\frac{t^N}{N!}$. The grand canonical partition function will not be discussed further in this paper, but we will allow $N$ to be arbitrarily large and independent of $K$, $\rho$, $a$, and $b$ in our main formulas for $\bm{s}\mapsto Z_N^\rho(K,a,b,\bm{s})$. This will ensure that $Z(t,\beta)$ and thermodynamic/scaling limits may be addressed with our present results in subsequent papers. To put our results in the appropriate context, we will now discuss local zeta functions and local fields $K\not\cong\R,\C$.

\subsection{Local zeta functions and $p$-fields}\label{1_2}
If $K\cong\R$ or $K\cong\C$, one defines smooth functions $\Phi:K^N\to\C$ as usual by identifying $K$ with $\R$ or $\C$, and then defines \emph{Schwartz-Bruhat} functions to be those $\Phi$ satisfying a ``rapid decay" condition, namely $\sup_{x\in K^N}|(D\Phi)(x)|<\infty$ for all $D=x_1^{m_1}(\partial/\partial x_1)^{n_1}\cdots x_N^{m_N}(\partial/\partial x_N)^{n_N}$ with $m_i,n_i\in\Z_{\geq 0}$. If instead $K\not\cong\R,\C$, then one calls $\Phi:K^N\to\C$ a \emph{Schwartz-Bruhat} function if it is locally constant (an analogue of ``smooth") with compact support (an analogue of ``rapid decay"). For general local fields $K$, the $\C$-vector space of Schwartz-Bruhat functions plays a fundamental role in Fourier analysis, the theory of distributions on $K^N$, and the theory of local zeta functions.

\begin{definition}\label{UVZF}
Given a local field $K$, a Schwartz-Bruhat function $\Phi:K^N\to\C$, and a tuple of polynomials $\bm{f}=(f_1,f_2,\dots,f_k)$ with $f_j\in K[x]$ for all $j$, the associated \emph{multivariate local zeta function} is defined on $\mathcal{H}^k:=\{\bm{s}\in\C^k:\re(s_j)>0\text{ for all }j\}$ by $$Z_\Phi(\bm{s},\bm{f}):=\int_{K^N}\Phi(x)\prod_{j=1}^k|f_j(x)|^{s_j}\,|dx|~.$$
\end{definition}

It is easily seen that $Z_\Phi(\cdot,\bm{f})$ is holomorphic on $\mathcal{H}^k$, though it is generally difficult to compute a formula for $Z_\Phi(\bm{s},\bm{f})$ and describe its meromorphic continuation. The classification of local fields breaks this problem into two main cases:
\begin{enumerate}
\item[(1)] $K$ is \emph{archimedean}, meaning the image of the canonical ring homomorphism $\Z\to K$ is unbounded with respect to $|\cdot|$. In this case $K\cong\R$ or $K\cong\C$, and $|\cdot|$ and $\mu$ are respectively identified with the usual absolute value and the Lebesgue measure on $\R$ or $\C$.
\item[(2)] $K$ is \emph{nonarchimedean}, meaning the image of $\Z\to K$ is contained in $R$. In this case $R$ is a local PID in which $P$ is the maximal ideal, and the \emph{residue field} $\kappa:=R/P$ is isomorphic to the finite field $\F_q$ where $q$ is a power of a prime number $p$. Thus $K$ is called a \emph{$p$-field}, for which there is a dichotomy:
\begin{enumerate}
\item If $\charr(K)=0$, then $K$ is isomorphic to a finite extension of $\Q_p$ and $K$ is called a \emph{$p$-adic field}. In particular, if $K\cong\Q_p$ then $R\cong\Z_p$, $P\cong p\Z_p$, $\kappa\cong\Z_p/p\Z_p\cong\F_p$, and hence $q=p$. 
\item If $\charr(K)=p$, then $K\cong\F_q((t))$, $R\cong\F_q[[t]]$, $P\cong t\F_q[[t]]$, and $K$ is called a \emph{function field}.
\end{enumerate}
\end{enumerate}

\noindent This classification and other features of general local fields $K$, such as Pontrjagin-self duality and their role in the ``ad\`elic formulation" of algebraic number theory, are elegantly laid out by Weil in \cite{Weil}.\\

In this paper we work exclusively with $p$-fields (i.e., $K\not\cong\R,\C$), and when $K$ is understood, we will always take $q$ to mean the cardinality of its residue field. Many results in the theory of local zeta functions for $p$-fields are inspired by the celebrated \emph{Igusa's Theorem}, of which the following proposition is an important consequence.
\begin{proposition}[\cite{Igusa75}]\label{IguThm}
Let $K$ be a $p$-adic field. If $\Phi:K^N\to\C$ is a Shwartz-Bruhat function and $f\in K[x]$ is a non-constant polynomial, then there is a rational function $r\in\C(T)$ such that the local zeta function defined by $$Z_\Phi(s,f)=\int_{K^N}\Phi(x)|f(x)|^s\,|dx|$$ satisfies $Z_\Phi(s,f)=r(q^{-s})$ for $\re(s)>0$. In particular, a meromorphic continuation of $Z_\Phi(s,f)$ is given by $r(q^{-s})$.
\end{proposition} 
The general theorem is established in \cite{Igusa74} and \cite{Igusa75}, and the proof therein relies on the existence of a certain type of resolution of singularities for $f$. Existence of such a resolution is guaranteed by \cite{Hironaka} if $\charr(K)=0$, but otherwise depends more subtly on $K$ and $f$. Thus Igusa's Theorem requires $\charr(K)=0$ (i.e., $K$ must be $p$-adic) in order to hold for general $f\in K[x]$. Loeser used a similar resolution technique to give a multivariate generalization of Igusa's Theorem in \cite{Loeser}, which implies the following analogue of \Cref{IguThm}. 
\begin{proposition}[\cite{Loeser}]\label{LoeThm}
Let $K$ be a $p$-adic field. If $\Phi:K^N\to\C$ is a Schwartz-Bruhat function and $\bm{f}=(f_1,f_2,\dots,f_k)$ with $f_j\in K[x]$ not all constant, then there is a $k$-variate rational function $r\in\C(T_1,T_2,\dots,T_k)$ such that the local zeta function defined by $$Z_\Phi(\bm{s},\bm{f}):=\int_{K^N}\Phi(x)\prod_{j=1}^k|f_j(x)|^{s_j}\,|dx|$$ satisfies $Z_\Phi(\bm{s},\bm{f})=r(q^{-s_1},q^{-s_2},\dots,q^{-s_k})$ for all $\bm{s}\in\mathcal{H}^k$.
\end{proposition} 

If $\supp(\Phi)$ is no longer assumed to be compact, $Z_\Phi(\cdot,\bm{f})$ may still admit a meromorphic continuation of a similar rational form. Such an example was recently investigated in \cite{GGZ} with applications to $p$-adic string theory. Therein it is shown that for $N\geq 4$ the \emph{$p$-adic open string $N$-point zeta function}, defined by $$Z^{(N)}(\bm{s}):=\int_{\Q_p^{N-3}}\prod_{i=2}^{N-2}|x_i|^{s_{1i}}|1-x_i|^{s_{i(N-1)}}\prod_{2\leq i<j\leq N-2}|x_i-x_j|^{s_{ij}}\,|dx|~,$$ coincides with a rational function in $p^{-s_{ij}}$ for all $1\leq i<j\leq N-1$ on a nonempty open domain in $\C^{\binom{N-1}{2}}$, despite the unbounded support of the integrand. Unlike Igusa's original method, a formula for $Z^{(N)}(\bm{s})$ was found by decomposing $\Q_p^{N-3}$ into finitely many sets, integrating over each one, and summing the results. This method does not require $\charr(K)=0$ and generalizes to all $p$-fields, while also providing a description of the domain and poles of $Z^{(N)}$ in terms of the decomposition of $\Q_p^{N-3}$. We will use a similar method to prove our main formulas for $Z_N^\rho(K,a,b,\bm{s})$, without placing any restrictions on $\charr(K)$ or $q$. For fixed $\rho$, $a$, and $b$, we will also describe regions of values $\bm{s}\in\C^{\binom{N}{2}}$ for which the integral in \Cref{ZNKabs} converges absolutely. In particular, we will show that $Z_N^\rho(K,a,b,\bm{s})$ can be written as a product of two main factors: a series that depends on $\rho$, $K$, $a$, $b$, and $\bm{s}$ in a simple way, and an explicit rational function in $q^{-b}$ and $q^{-s_{ij}}$ that does not depend on $a$ or $\rho$.

\section{Statement of results}
The main result of this paper is a pair of formulas for $Z_N^\rho(K,a,b,\bm{s})$, where $K$ is an arbitrary $p$-field. We are primarily interested in $Z_N^\rho(K,a,b,\bm{s})$ as a function of $\bm{s}$, and we would like our formulas to hold for arbitrary $N$, $K$, $\rho$, $a$, and $b$. These five parameters are not entirely independent, however, as the domain of $\rho$ given in \Cref{ZNKabs} depends on $N$ and $K$. Though it is possible to give similar results for arbitrary $\rho:\|K^N\|\to\C$, the required notation, cases, and proofs become prohibitively cumbersome. We will avoid this problem by making the following mild assumptions about $\rho$. It is well-known that for every $p$-field $K$ and every positive integer $N$ we have $\|K^N\|\subset\mathcal{N}$, where $$\mathcal{N}:=\{0\}\cup\bigcup_{n=1}^\infty\left\{n,\frac{1}{n}\right\}~,$$ so we will henceforth assume $\rho$ is defined on all of $\mathcal{N}$. This assumption ensures that $\rho$ is independent of $K$ and $N$ while also maintaining that, for any choice of $K$ and $N$, the function $x\mapsto\rho(\|x\|)$ is measurable on $K^N$. To keep most of our upcoming discussion independent of $\rho$, we will further assume
\begin{equation}\label{rhogro}
\limsup_{n\to\infty}\frac{\log|\rho(\frac{1}{n})|_{\C}}{\log(n)}\leq 1\qquad\text{and}\qquad\limsup_{n\to\infty}\frac{\log|\rho(n)|_{\C}}{\log(n)}=-\infty
\end{equation}
where $|\cdot|_{\C}$ denotes the canonical absolute value on $\C$ and $\log:[0,\infty]\to[-\infty,\infty]$ is the extended natural logarithm (i.e., $\log(0):=-\infty$ and $\log(\infty):=\infty$). That is, for any choice of $K$ and $N$, the function $x\mapsto\rho(\|x\|)$ has modest growth as $\|x\|\to 0$ and fast decay as $\|x\|\to\infty$. Examples of $\rho:\mathcal{N}\to\C$ satisfying \eqref{rhogro} include $\rho(t)=e^{-t}$, $\rho(t)=e^{-t^2/2}$, $\rho(t)=\bm{1}_{[0,1]}(t)$, and $\rho(t)=\log(t)\bm{1}_{[0,1]}(t)$.

\subsection{The main theorem}\label{2_1}
There are two main factors comprising $Z_N^\rho(K,a,b,\bm{s})$, and they can be defined in their own right. Thus, until the statement of the main theorem we will assume that $N$ and $q$ are arbitrary integers satisfying $N\geq 2$ and $q\geq 2$. The first of the two main factors is the \emph{root function}, defined on a convex domain called the \emph{root polytope} as follows:

\begin{definition}\label{rootdef}
For $N\geq 2$ and $a,b\in\C$ we define the \emph{root polytope} $\mathcal{RP}_N(a,b)$ by $$\mathcal{RP}_N(a,b):=\left\{\bm{s}\in\C^{\binom{N}{2}}:\re\bigg(N-1+a+b+\sum_{i<j}s_{ij}\bigg)>0\right\}~.$$ For such $N$, $a$, $b$, an integer $q\geq 2$, and a function $\rho:\mathcal{N}\to\C$ satisfying \eqref{rhogro}, the associated \emph{root function} $\mathcal{RP}_N(a,b)\to\C$ is defined by $$\bm{s}\mapsto H_q^\rho\bigg(N+a+b+\sum_{i<j}s_{ij}\bigg)\qquad\text{where}\qquad H_q^\rho(z):=\frac{1-q^{-z}}{1-q^{-(z-1)}}\cdot\sum_{m\in\Z}\rho(q^m)q^{mz}~.$$
\end{definition}

The second factor is more complicated and requires some combinatorial language. Recall that a \emph{partition} of the set $[N]:=\{1,2,\dots,N\}$ is a set $\ptn$ of nonempty pairwise disjoint subsets $\lambda\subset[N]$ satisfying $\bigcup_{\lambda\in\ptn}\lambda=[N]$. If $\ptn_1$ and $\ptn_2$ are partitions of $[N]$, we write $\ptn_2\leq\ptn_1$ and call $\ptn_2$ a \emph{refinement} of $\ptn_1$ if each part $\lambda_2\in\ptn_2$ is contained in some part $\lambda_1\in\ptn_1$. We write $\ptn_2<\ptn_1$ and call $\ptn_2$ a \emph{proper refinement} of $\ptn_1$ if both $\ptn_2\leq\ptn_1$ and $\ptn_2\neq\ptn_1$. The relation $\leq$ makes the collection of all partitions of $[N]$ into a partially ordered lattice with height $N$, unique maximal element $\overline{\ptn}:=\left\{[N]\right\}$, and unique minimal element $\underline{\ptn}:=\left\{\{1\},\{2\},\dots,\{N\}\right\}$. The \emph{rank} of a partition $\ptn$ of $[N]$ is the integer $$\rank(\ptn):=N-\#\ptn=\sum_{\lambda\in\ptn}(\#\lambda-1)~.$$
\begin{definition}\label{exponents}
Suppose $N\geq 2$. As needed, sums over empty index sets are defined to be 0.
\begin{itemize}
\item[(a)] For each nonempty subset $\lambda\subset[N]$, define the \emph{part exponent} $e_\lambda:\C^{\binom{N}{2}}\to\C$ by $$e_\lambda(\bm{s}):=\sum_{\substack{i<j\\i,j\in\lambda}}\left(s_{ij}+\frac{2}{\#\lambda}\right)=(\#\lambda-1)+\sum_{\substack{i<j\\i,j\in\lambda}}s_{ij}~.$$
\item[(b)] For each partition $\ptn$ of $[N]$, define the \emph{partition exponent} $E_\ptn:\C^{\binom{N}{2}}\to\C$ by $$E_\ptn(\bm{s}):=\sum_{\lambda\in\ptn}e_\lambda(\bm{s})=\rank(\ptn)+\sum_{\lambda\in\ptn}\sum_{\substack{i<j\\i,j\in\lambda}}s_{ij}~.$$
\end{itemize}
\end{definition}
\begin{definition}[Splitting chains]\label{splchdef}
A finite tuple $\spl=(\ptn_0,\ptn_1,\dots,\ptn_L)$ of partitions of $[N]$ satisfying $$\overline{\ptn}=\ptn_0>\ptn_1>\ptn_2>\dots>\ptn_L=\underline{\ptn}$$ shall be called a \emph{splitting chain} of order $N$. We write $\mathcal{S}_N$ for the set of all splitting chains of order $N$, and we attach the following terminology and notation to each $\spl\in\mathcal{S}_N$ with $N\geq 2$:
\begin{itemize}
\item[(a)] The positive integer $L(\spl):=L$ is the \emph{length} of $\spl$ and the partitions $\ptn_0,\ptn_1,\dots,\ptn_{L(\spl)-1}$ are the \emph{levels} of $\spl$. Call each non-singleton part $\lambda\in\ptn_0\cup\ptn_1\cup\dots\cup\ptn_{L(\spl)-1}$ a \emph{branch} of $\spl$ and write $\mathcal{B}(\spl)$ for the set of all branches of $\spl$, i.e., $$\mathcal{B}(\spl):=(\ptn_0\cup\ptn_1\cup\dots\cup\ptn_{L(\spl)-1})\setminus\underline{\ptn}~.$$
\item[(b)] Since $\spl$ must terminate at $\ptn_{L(\spl)}=\underline{\ptn}$, each branch appears in a final level $\ptn_\ell$ before it refines into two or more parts in $\ptn_{\ell+1}$. Thus for each $\lambda\in\mathcal{B}(\spl)$ we define the \emph{depth} $\ell_\spl(\lambda)\in\{0,1,\dots,L(\spl)-1\}$ and \emph{degree} $\deg_\spl(\lambda)\in\{2,3,\dots, N\}$ respectively by $$\ell_\spl(\lambda):=\max\{\ell:\lambda\in\ptn_\ell\}\qquad\text{and}\qquad\deg_\spl(\lambda):=\#\{\lambda'\in\ptn_{\ell_\spl(\lambda)+1}:\lambda'\subset\lambda\}~.$$
\item[(c)] Finally, we use the falling factorial notation $(z)_n=z\cdot(z-1)\cdot(z-2)\cdot\ldots\cdot(z-n+1)$ for integers $n\geq 1$ and define the \emph{multiplicity polynomial} $M_\spl(t)\in\Z[t]$ by $$M_\spl(t):=\prod_{\lambda\in\mathcal{B}(\spl)}(t-1)_{\deg_\spl(\lambda)-1}~.$$
\end{itemize}
\end{definition}

It is a key observation that $\mathcal{S}_N$ is finite for each $N\geq 2$. This is easily seen from the definition, since every $\spl\in\mathcal{S}_N$ must satisfy $0<L(\spl)<N$ and there are at most finitely many $\spl\in\mathcal{S}_N$ of a given length. One should also note that the multiplicity polynomial for a splitting chain $\spl\in\mathcal{S}_N$ factors as $$M_\spl(t)=\prod_{d=1}^{N-1}(t-d)^{p_\spl(d)}\qquad\text{where}\qquad p_\spl(d)=\#\{\lambda\in\mathcal{B}(\spl):\deg_\spl(\lambda)>d\}~,$$ because the falling factorial $(t-1)_{\deg_\spl(\lambda)-1}$ contributes a factor of $(t-d)$ if and only if $\deg_\spl(\lambda)>d$, and there are precisely $p_\spl(d)$ such falling factorials in the definition of $M_\spl(t)$. Thus, given an integer $q\geq 2$, we have $M_\spl(q)>0$ if $\deg_\spl(\lambda)\leq q$ for all $\lambda\in\mathcal{B}(\spl)$, and $M_\spl(q)=0$ otherwise. Multiplicity polynomials and the exponents in \Cref{exponents} together form the \emph{branch/level polytopes} and \emph{branch/level functions} defined below. As we shall soon see, the sum of level functions over all $\spl\in\mathcal{S}_N$ is the second main factor in our formula for $Z_N^\rho(K,a,b,\bm{s})$.

\begin{definition}\label{branchlevel}
Suppose $N\geq 2$ and $q\geq 2$ are integers and suppose $\spl\in\mathcal{S}_N$. As needed, products and intersections over empty index sets are respectively defined to be $1$ and $\C^{\binom{N}{2}}$.
\begin{enumerate}
\item[(a)] The \emph{branch polytope} $\mathcal{BP}_\spl$ and \emph{branch function} $I_{\spl,q}:\mathcal{BP}_\spl\to\C$ are respectively defined by
\begin{align*}
\mathcal{BP}_\spl&:=\bigcap_{\lambda\in\mathcal{B}(\spl)\setminus\overline{\ptn}}\left\{\bm{s}\in\C^{\binom{N}{2}}:\re(e_\lambda(\bm{s}))>0\right\}\qquad\text{and}\\
I_{\spl,q}(\bm{s})&:=\frac{M_\spl(q)}{q^{N-1}}\cdot\prod_{\lambda\in\mathcal{B}(\spl)\setminus\overline{\ptn}}\frac{1}{q^{e_\lambda(\bm{s})}-1}~.
\end{align*}
\item[(b)] Given $b\in\C$, the \emph{level polytope} $\mathcal{LP}_\spl(b)$ and \emph{level function} $J_{\spl,q}(b,\cdot):\mathcal{LP}_\spl(b)\to\C$ are respectively defined by
\begin{align*}
\mathcal{LP}_\spl(b)&:=\bigcap_{\ell=1}^{L(\spl)-1}\left\{\bm{s}\in\C^{\binom{N}{2}}:\re(b+E_{\ptn_\ell}(\bm{s}))>0\right\}\qquad\text{and}\\
J_{\spl,q}(b,\bm{s})&:=\frac{M_\spl(q)}{q^{N-1}}\cdot\prod_{\ell=1}^{L(\spl)-1}\frac{1}{q^{b+E_{\ptn_\ell}(\bm{s})}-1}~.
\end{align*}
\end{enumerate}
\end{definition}

Note that $I_{\spl,q}(\bm{s})$ and $J_{\spl,q}(b,\bm{s})$ are rational functions in the variables $q$, $q^{-b}$, and $q^{-s_{ij}}$ for all $i<j$, with $\Q$ coefficients determined by $\spl$ alone. Given $\spl$ and an integer $q\geq 2$, there are two possibilities:
\begin{itemize}
\item[(i)] If $\deg_\spl(\lambda)\leq q$ for all $\lambda\in\mathcal{B}(\spl)$, then $M_\spl(q)>0$, and hence $I_{\spl,q}(\bm{s})$ and $J_{\spl,q}(b,\bm{s})$ are never zero.
\item[(ii)] If $\deg_\spl(\lambda)>q$ for some $\lambda\in\mathcal{B}(\spl)$, then $M_\spl(q)=0$, and hence $I_{\spl,q}$ and $J_{\spl,q}(b,\cdot)$ are identically zero on $\mathcal{BP}_\spl$ and $\mathcal{LP}_\spl(b)$ respectively.
\end{itemize}
In any case, $I_{\spl,q}$ and $J_{\spl,q}(b,\cdot)$ are holomorphic on their respective polytopes $\mathcal{BP}_\spl$ and $\mathcal{LP}_\spl(b)$, which are both open and convex. These polytopes are related by the following lemma, which is the last ingredient we need to state the main theorem.

\begin{lemma}\label{reduction} We say that a splitting chain $\spl$ is \emph{reduced} if for each $\lambda\in\mathcal{B}(\spl)$ there is a unique level $\ptn_\ell$ containing $\lambda$ (namely, the level $\ptn_{\ell_\spl(\lambda)}$). We write $\mathcal{R}_N:=\{\spl\in\mathcal{S}_N:\spl\text{ is reduced}\}$ and define an equivalence relation $\simeq$ on $\mathcal{S}_N$ by writing $\spl\simeq\spl'$ if and only if $\mathcal{B}(\spl)=\mathcal{B}(\spl')$.
\begin{enumerate}
\item[(a)] If $\spl\simeq\spl'$, then the branch degrees, part exponents, multiplicity polynomials, and branch polytopes for $\spl$ and $\spl'$ respectively coincide.
\item[(b)] For each $\spl\in\mathcal{S}_N$ there is a unique $\spl^*\in\mathcal{R}_N$ such that $\spl\simeq\spl^*$. We call this $\spl^*$ the \emph{reduction} of $\spl$ and regard $\mathcal{R}_N$ as a complete set of representatives for $\mathcal{S}_N$ modulo $\simeq$.
\item[(c)] For each $\spl^*\in\mathcal{R}_N$ we have $$\bigcap_{\substack{\spl\in\mathcal{S}_N\\\spl\simeq\spl^*}}\mathcal{LP}_\spl(0)=\mathcal{BP}_{\spl^*}~,$$ and therefore $$\quad\bigcap_{\spl\in\mathcal{S}_N}\mathcal{LP}_\spl(0)=\bigcap_{\spl^*\in\mathcal{R}_N}\mathcal{BP}_{\spl^*}~.$$
\end{enumerate}
\end{lemma}

\begin{theorem}[Main Theorem]\label{main}
Fix $N\geq 2$ and $a,b\in\C$ and define the convex open polytope $$\Omega_N(a,b):=\mathcal{RP}_N(a,b)\cap\bigcap_{\spl\in\mathcal{S}_N}\mathcal{LP}_\spl(b)~.$$
\begin{enumerate}
\item[(a)] If $K$ is any $p$-field and $\rho:\mathcal{N}\to\C$ satisfies \eqref{rhogro} and is not identically zero, then the integral $$Z_N^\rho(K,a,b,\bm{s})=\int_{K^N}\rho(\|x\|)\big(\max_{i<j}|x_i-x_j|\big)^a\big(\min_{i<j}|x_i-x_j|\big)^b\prod_{i<j}|x_i-x_j|^{s_{ij}}\,|dx|$$ converges absolutely for all $\bm{s}\in\Omega_N(a,b)$, and $\Omega_N(a,b)$ is the largest open subset of $\C^{\binom{N}{2}}$ with this property.
\item[(b)] If $K$ has residue field cardinality $q$, then on each compact subset of $\Omega_N(a,b)$ the integral is given by the uniformly convergent sum $$Z_N^\rho(K,a,b,\bm{s})=H_q^\rho\bigg(N+a+b+\sum_{i<j}s_{ij}\bigg)\cdot\sum_{\spl\in\mathcal{S}_N}J_{\spl,q}(b,\bm{s})~.$$
\item[(c)] Given $b=0$, $K$ and $q$ as above, and $\spl^*\in\mathcal{R}_N$, we have $$\sum_{\substack{\spl\in\mathcal{S}_N\\\spl\simeq\spl^*}}J_{\spl,q}(0,\bm{s})=I_{\spl^*\!\!,q}(\bm{s})\quad\text{for all $\bm{s}\in\mathcal{BP}_{\spl^*}$}~.$$ Thus on each compact subset of $\Omega_N(a,0)$ the integral is given by the uniformly convergent sum $$Z_N^\rho(K,a,0,\bm{s})=H_q^\rho\bigg(N+a+\sum_{i<j}s_{ij}\bigg)\cdot\sum_{\spl^*\in\mathcal{R}_N}I_{\spl^*\!\!,q}(\bm{s})~.$$ 
\end{enumerate}
\end{theorem}

\Cref{reduction} and \Cref{main} will be proved in Section 3. Note that part (a) of \Cref{main} is independent of $K$ and $\rho$, parts (b) and (c) depend on $K$ only via $q$, and all parts of the theorem depend on $a$ and $b$ in relatively simple ways. The parameter $N$ has a far more complicated role, so we only discuss the $N=2$ and $N=3$ examples at the moment. In order to streamline notation, we will write each partition $\ptn=\{\lambda_1,\lambda_2,\dots,\lambda_k\}$ as a string of parts, i.e., $\ptn=\lambda_1\lambda_2\dots\lambda_k$. 

\begin{example}\label{Ex_N2n3}
Fix $a,b\in\C$ and a not-identically-zero function $\rho:\mathcal{N}\to\C$ satisfying \eqref{rhogro}.
\begin{itemize}
\item If $N=2$, then $\binom{N}{2}=1$, so each $\bm{s}\in\C^{\binom{N}{2}}$ is simply a number $s\in\C$. The only splitting chain in $\mathcal{S}_2$ is $\spl=(\{1,2\},\{1\}\{2\})$, and by \Cref{branchlevel} it has $$\mathcal{LP}_\spl(b)=\C\qquad\text{and}\qquad J_{\spl,q}(b,\bm{s})=I_{\spl,q}(\bm{s})=\frac{q-1}{q}\qquad\text{for all }q\geq 2~.$$ Thus for all $p$-fields $K$ with residue cardinality $q$ and all $s$ in the region $$\Omega_2(a,b)=\mathcal{RP}_2(a,b)\cap\C=\mathcal{RP}_2(a,b)=\{s\in\C:\re(1+a+b+s)>0\}~,$$ the integral $Z_2^\rho(K,a,b,\bm{s})$ converges absolutely to the value $$Z_2^\rho(K,a,b,\bm{s})=\frac{q-1}{q}\cdot\frac{1-q^{-(2+a+b+s)}}{1-q^{-(1+a+b+s)}}\cdot\sum_{m\in\Z}\rho(q^m)q^{m(2+a+b+s)}~.$$

\item If $N=3$, then we have $\bm{s}=(s_{12},s_{13},s_{23})\in\C^3$ with root polytope $$\mathcal{RP}_3(a,b)=\{\bm{s}\in\C^3:\re(2+a+b+s_{12}+s_{13}+s_{23})>0\}~,$$ and \Cref{branchlevel} provides the following table for all four splitting chains in $\mathcal{S}_3$:
\begin{center}
\def\arraystretch{1.5}
\begin{tabular}{c||c|c}
$\spl=(\ptn_0,\ptn_1,\dots,\ptn_{L(\spl)})\in\mathcal{S}_3$ & $J_{\spl,q}(b,\bm{s})$ & $\mathcal{LP}_\spl(b)$\\
\hline\hline
\begin{tabular}{l}
$\ptn_0=\{1,2,3\}$\\
$\ptn_1=\{1\}\{2\}\{3\}$\\
\end{tabular}
& $\dfrac{(q-1)(q-2)}{q^2}$ & $\C^3$\\
\hline
\begin{tabular}{l}
$\ptn_0=\{1,2,3\}$\\
$\ptn_1=\{1,2\}\{3\}$\\
$\ptn_2=\{1\}\{2\}\{3\}$
\end{tabular}
& $\dfrac{(q-1)^2}{q^2}\cdot\dfrac{1}{q^{1+b+s_{12}}-1}$ & $\{\bm{s}\in\C^3:\re(1+b+s_{12})>0\}$\\
\hline
\begin{tabular}{l}
$\ptn_0=\{1,2,3\}$\\
$\ptn_1=\{1,3\}\{2\}$\\
$\ptn_2=\{1\}\{2\}\{3\}$
\end{tabular}
& $\dfrac{(q-1)^2}{q^2}\cdot\dfrac{1}{q^{1+b+s_{13}}-1}$ & $\{\bm{s}\in\C^3:\re(1+b+s_{13})>0\}$\\
\hline
\begin{tabular}{l}
$\ptn_0=\{1,2,3\}$\\
$\ptn_1=\{1\}\{2,3\}$\\
$\ptn_2=\{1\}\{2\}\{3\}$
\end{tabular}
& $\dfrac{(q-1)^2}{q^2}\cdot\dfrac{1}{q^{1+b+s_{23}}-1}$ & $\{\bm{s}\in\C^3:\re(1+b+s_{23})>0\}$
\end{tabular}
\end{center}

Thus for all $p$-fields $K$ with residue cardinality $q$ and for all $\bm{s}$ in the region 
 $$\Omega_3(a,b)=\{\bm{s}\in\C^3:\re(2+a+b+s_{12}+s_{13}+s_{23})>0\}\cap\bigcap_{1\leq i<j\leq 3}\{\bm{s}\in\C^3:\re(1+b+s_{ij})>0\}~,$$ the integral $Z_3^\rho(K,a,b,\bm{s})$ converges absolutely to the value
 \begin{align*}
 Z_3^\rho(K,a,b,\bm{s})&=\frac{1-q^{-(3+a+b+s_{12}+s_{13}+s_{23})}}{1-q^{-(2+a+b+s_{12}+s_{13}+s_{23})}}\cdot\sum_{m\in\Z}\rho(q^m)q^{m(3+a+b+s_{12}+s_{13}+s_{23})}\\
&\cdot\frac{1}{q^2}\left((q-1)(q-2)+(q-1)^2\left[\frac{1}{q^{1+b+s_{12}}-1}+\frac{1}{q^{1+b+s_{13}}-1}+\frac{1}{q^{1+b+s_{23}}-1}\right]\right)~.
\end{align*}
\end{itemize}
\end{example}

\begin{remark}\label{counting}
Note that every splitting chain of order 2 or 3 is reduced and $J_\spl(0,\bm{s})=I_\spl(\bm{s})$ for all $\spl\in\mathcal{S}_2$ and all $\spl\in\mathcal{S}_3$. Therefore part (c) of \Cref{main} is redundant when $N=2$ or $N=3$, for in these cases it coincides with part (b) applied to $b=0$. If $N\geq 4$ we have $\mathcal{R}_N\subsetneq\mathcal{S}_N$, because there is at least one non-reduced splitting chain $\spl=(\ptn_0,\ptn_1,\ptn_2,\ptn_3)\in\mathcal{S}_N$ such as the one given by
\begin{align*}
\ptn_0&=\{1,2,3,4,\dots,N\}~,\\
\ptn_1&=\{1,2\}\{3,4,\dots,N\}~,\\
\ptn_2&=\{1,2\}\{3\}\{4\}\dots\{N\}~,\\
\ptn_3&=\{1\}\{2\}\{3\}\{4\}\dots\{N\}~.
\end{align*}
Finding closed forms for the cardinalities of $\mathcal{S}_N$ and $\mathcal{R}_N$ for general $N$ is nontrivial, but they can be bounded below as follows. Given $\spl\in\mathcal{R}_N$ and $i\in[N]$, we may construct a particular $\spl'\in\mathcal{R}_{N+1}$: For each $\ell\in\{0,1,2,\dots,L(\spl)\}$, let $\ptn_\ell'$ be the partition of $[N+1]$ obtained from $\ptn_\ell$ by replacing the unique part $\lambda\in\ptn_\ell$ containing $i$ by the larger part $\lambda\cup\{N+1\}$. If we then set $\ptn_{L(\spl)+1}':=\underline{\ptn}$, it is easily verified that $\spl'=(\ptn_0',\ptn_1',\dots,\ptn_{L(\spl)+1}')$ is a reduced splitting chain of order $N+1$. Thus $(\spl,i)\mapsto\spl'$ defines a function $\mathcal{R}_N\times[N]\to\mathcal{R}_{N+1}$, which is injective because it has a left inverse: The integer $i$ can be recovered from $\spl'$ because it is the only element of $[N]$ satisfying $\{i,N+1\}\in\spl_{L(\spl)}'$, and then $\spl$ can be recovered from $\spl'$ by simply removing $\ptn_{L(\spl)+1}$ and all copies of $N+1$ from $\spl'$. Thus we have $\#\mathcal{R}_N\cdot N\leq\#\mathcal{R}_{N+1}$ for all $N\geq 2$, and we saw before that $\#\mathcal{R}_2=\#\mathcal{S}_2=1$, $\#\mathcal{R}_3=\#\mathcal{S}_3=4$, and $\mathcal{R}_N\subsetneq\mathcal{S}_N$ for all $N\geq 4$. Induction on $N$ then gives the following bounds: $$(N-1)!\leq\#\mathcal{R}_N\leq\#\mathcal{S}_N\qquad\text{for all }N\geq 2.$$ The left inequality is strict for $N\geq 3$ and both are strict for $N\geq 4$.
\end{remark}

The preceding remark implies that the branch function sum $\sum_{\spl^*\in\mathcal{R}_N}I_{\spl^*\!\!,q}(\bm{s})$ has strictly fewer terms than the level function sum $\sum_{\spl\in\mathcal{S}_N}J_{\spl,q}(b,\bm{s})$ when $N\geq 4$, and hence part (c) of \Cref{main} becomes a simplification of part (b) applied to $b=0$. Though $N=4$ is the least $N$ for which this simplification is noticeable, the sums of branch functions and level functions respectively have $\#\mathcal{R}_4=26$ terms and $\#\mathcal{S}_4=32$ terms in this case, so the rather large computation of $Z_4^\rho(K,a,b,\bm{s})$ will be postponed until the appendix. For now we consider only three elements from $\mathcal{S}_4$ to discuss how part (c) of \Cref{main} simplifies the $b=0$ case of part (b). 

\begin{example}\label{Ex_N4}
Consider the three splitting chains $\spl^*,\spl',\spl''\in\mathcal{S}_4$ defined by $$\parbox{4cm}{
\begin{align*}
\ptn_0^*&=\{1,2,3,4\}~,\\
\ptn_1^*&=\{1,2\}\{3,4\}~,\\
\ptn_2^*&=\{1\}\{2\}\{3\}\{4\}~,
\end{align*}}
\qquad\parbox{4cm}{
\begin{align*}
\ptn_0'&=\{1,2,3,4\}~,\\
\ptn_1'&=\{1,2\}\{3,4\}~,\\
\ptn_2'&=\{1,2\}\{3\}\{4\}~,\\
\ptn_3'&=\{1\}\{2\}\{3\}\{4\}~,
\end{align*}}
\qquad\text{and}
\qquad\parbox{4cm}{
\begin{align*}\ptn_0''&=\{1,2,3,4\}~,\\
\ptn_1''&=\{1,2\}\{3,4\}~,\\
\ptn_2''&=\{1\}\{2\}\{3,4\}~,\\
\ptn_3''&=\{1\}\{2\}\{3\}\{4\}~.
\end{align*}}$$
Recalling \Cref{reduction}, $\spl^*$ is the common reduction of all three, and it is easily verified that no other $\spl\in\mathcal{S}_4\setminus\{\spl^*,\spl',\spl''\}$ satisfies $\spl\simeq\spl^*$. By \Cref{branchlevel}, the splitting chains $\spl^*$, $\spl'$, and $\spl''$ contribute the following level functions to the sum $\sum_{\spl\in\mathcal{S}_4}J_{\spl,q}(b,\bm{s})$ in part (b) of \Cref{main}:
\begin{align*}
J_{\spl^*\!\!,q}(b,\bm{s})&=\frac{(q-1)^3}{q^3}\cdot\frac{1}{q^{2+b+s_{12}+s_{34}}-1}~,\\
J_{\spl'\!\!,q}(b,\bm{s})&=\frac{(q-1)^3}{q^3}\cdot\frac{1}{q^{2+b+s_{12}+s_{34}}-1}\cdot\frac{1}{q^{1+b+s_{12}}-1}~,\\
J_{\spl''\!\!,q}(b,\bm{s})&=\frac{(q-1)^3}{q^3}\cdot\frac{1}{q^{2+b+s_{12}+s_{34}}-1}\cdot\frac{1}{q^{1+b+s_{34}}-1}~.
\end{align*} 
In particular, their total contribution to the sum can be written as
\begin{equation}\label{prism}
\sum_{\substack{\spl\in\mathcal{S}_4\\\spl\simeq\spl^*}}J_{\spl,q}(b,\bm{s})=\frac{(q-1)^3}{q^3}\cdot\frac{1}{q^{1+b+s_{12}}-1}\cdot\frac{1}{q^{1+b+s_{34}}-1}\cdot\frac{q^{2+2b+s_{12}+s_{34}}-1}{q^{2+b+s_{12}+s_{34}}-1}~.
\end{equation}
Equation \eqref{prism} hints at an interesting analytic feature of the parameter $b$. Indeed, if $q\geq 2$ and $b\in\C$ are fixed, then each of the summands $J_{\spl^*\!\!,q}(b,\bm{s})$, $J_{\spl'\!\!,q}(b,\bm{s})$, and $J_{\spl''\!\!,q}(b,\bm{s})$ is meromorphic in $\bm{s}=(s_{12},s_{13},s_{14},s_{23},s_{24},s_{34})\in\C^6$, and each of their sets of poles contains the infinite set $$C(b)=\left\{\bm{s}\in\C^6:2+b+s_{12}+s_{34}\in\frac{2\pi i\Z}{\log(q)}~,~1+b+s_{12}\notin\frac{2\pi i\Z}{\log(q)}~,\text{ and }1+b+s_{34}\notin\frac{2\pi i}{\log(q)}\right\}~.$$ If $b$ is not an integer multiple of $2\pi i/\log(q)$, the poles for the sum in \eqref{prism} also include $C(b)$. However, if $b$ \emph{is} an integer multiple of $2\pi i/\log(q)$, then $(q^{2+2b+s_{12}+s_{34}}-1)/(q^{2+b+s_{12}+s_{34}}-1)=1$ and \emph{none} of the $\bm{s}\in C(b)$ are poles for the sum in \eqref{prism}. In particular, $C(0)$ is a common set of poles for all of the level functions $J_{\spl^*\!\!,q}(0,\bm{s})$, $J_{\spl'\!\!,q}(0,\bm{s})$, and $J_{\spl''\!\!,q}(0,\bm{s})$, but all such poles ``cancel" when the level functions are added together: $$\sum_{\substack{\spl\in\mathcal{S}_4\\\spl\simeq\spl^*}}J_{\spl,q}(0,\bm{s})=\frac{(q-1)^3}{q^3}\cdot\frac{1}{q^{1+s_{12}}-1}\cdot\frac{1}{q^{1+s_{34}}-1}=I_{\spl^*\!\!,q}(\bm{s})~.$$ Thus, by collapsing the sum $\sum_{\spl\in\mathcal{S}_4}J_{\spl,q}(0,\bm{s})$ from part (b) of \Cref{main} into its ``reduced" form $\sum_{\spl\in\mathcal{R}_4}I_{\spl,q}(\bm{s})$, part (c) shows that many level function poles ``cancel" in the $b=0$ case.
\end{example}
 
\begin{remark}\label{Rk_poles}
For simple choices of $\rho$, the root function sums to a closed form. In this case \Cref{main} provides meromorphic continuations of both $\bm{s}\mapsto Z_N^\rho(K,a,b,\bm{s})$ and $\bm{s}\mapsto Z_N^\rho(K,a,0,\bm{s})$, and their candidate poles may be easily described. For example, suppose $K$ has residue field cardinality $q$ and $\rho=\bm{1}_{[0,1]}$. It is easily verified from \Cref{rootdef} and \Cref{main} that $Z_N^\rho(K,a,b,\bm{s})$ coincides with the sum
\begin{equation}\label{Lsum}
\frac{q^{a+b+\sum_{i<j}s_{ij}}}{q^{N-1+a+b+\sum_{i<j}s_{ij}}-1}\cdot\sum_{\spl\in\mathcal{S}_N}M_\spl(q)\cdot\prod_{\ell=1}^{L(\spl)-1}\frac{1}{q^{b+E_{\ptn_\ell}(\bm{s})}-1}
\end{equation}
on the convex open region $\Omega_N(a,b)$. Since each summand is meromorphic in $\C^{\binom{N}{2}}$ with set of poles $$\mathscr{L}_{\spl,q}:=\left\{\bm{s}\in\C^{\binom{N}{2}}:N-1+a+b+\sum_{i<j}s_{ij}\in\frac{2\pi i\Z}{\log(q)}\right\}\cup\bigcup_{\ell=1}^{L(\spl)-1}\left\{\bm{s}\in\C^{\binom{N}{2}}:b+E_{\ptn_\ell}(\bm{s})\in\frac{2\pi i\Z}{\log(q)}\right\}~,$$ then \eqref{Lsum} defines the meromorphic continuation of $Z_N^\rho(K,a,b,\bm{s})$ to $\C^{\binom{N}{2}}$, and its poles are contained in the union $\bigcup_{\spl\in\mathcal{S}_N}\mathscr{L}_{\spl,q}$. Similarly, \Cref{rootdef} and part (c) of \Cref{main} show that $Z_N^\rho(K,a,0,\bm{s})$ coincides with the sum 
\begin{equation}\label{Bsum}
\frac{q^{a+\sum_{i<j}s_{ij}}}{q^{N-1+a+\sum_{i<j}s_{ij}}-1}\cdot\sum_{\spl^*\in\mathcal{R}_N}M_{\spl^*}(q)\cdot\prod_{\lambda\in\mathcal{B}(\spl^*)\setminus\overline{\ptn}}\frac{1}{q^{e_\lambda(\bm{s})}-1}
\end{equation}
on a convex open region, and each summand is meromorphic in $\C^{\binom{N}{2}}$ with set of poles $$\mathscr{B}_{\spl^*\!\!,q}:=\left\{\bm{s}\in\C^{\binom{N}{2}}:N-1+a+\sum_{i<j}s_{ij}\in\frac{2\pi i\Z}{\log(q)}\right\}\cup\bigcup_{\lambda\in\mathcal{B}(\spl^*)\setminus\overline{\ptn}}\left\{\bm{s}\in\C^{\binom{N}{2}}:e_\lambda(\bm{s})\in\frac{2\pi i\Z}{\log(q)}\right\}~.$$ Therefore \eqref{Bsum} defines the meromorphic continuation of $Z_N^\rho(K,a,0,\bm{s})$ to $\C^{\binom{N}{2}}$, and its poles are contained in the union $\bigcup_{\spl^*\in\mathcal{R}_N}\mathscr{B}_{\spl^*\!\!,q}$. Though the poles of each summand in \eqref{Lsum} and \eqref{Bsum} are easily described, we saw in \Cref{Ex_N4} that pole cancellation is possible when summands are brought together. As is true for general local zeta functions, determining precisely which of the poles in the candidate sets $\bigcup_{\spl\in\mathcal{S}_N}\mathscr{L}_{\spl,q}$ and $\bigcup_{\spl^*\in\mathcal{R}_N}\mathscr{B}_{\spl^*\!\!,q}$ cancel is a highly nontrivial task.
\end{remark}

\subsection{Applications to log-Coulomb gas}\label{2_2}

The desired formulas for the mixed-charge $p$-field analogue of $\mathcal{Z}_N(\beta)$ and the expected value in \eqref{Jointeq} are easily obtained by evaluating the formulas in \Cref{main} at special values of $\bm{s}$. To this end, we define several new items related to those in \Cref{rootdef,branchlevel}.

\begin{definition}\label{absdef}
Suppose $a,b\in\C$ and $\mathfrak{q}_1,\mathfrak{q}_2,\dots,\mathfrak{q}_N>0$ where $N\geq 2$, and let $\bm{c}:=(\mathfrak{q}_i\mathfrak{q}_j)_{i<j}$.
\begin{itemize}
\item[(a)] Define the \emph{root abscissa} $\mathcal{RP}_N^{\bm{c}}(a,b)$ by $$\mathcal{RP}_N^{\bm{c}}(a,b):=-\frac{N-1+\re(a+b)}{\sum_{i<j}\mathfrak{q}_i\mathfrak{q}_j}~.$$
\item[(b)] For each $\spl\in\mathcal{S}_N$, define the \emph{branch abscissa} $\mathcal{BP}_\spl^{\bm{c}}$ by $$\mathcal{BP}_\spl^{\bm{c}}:=-\inf_{\lambda\in\mathcal{B}(\spl)\setminus\overline{\ptn}}\left\{\frac{\#\lambda-1}{\varepsilon_\lambda(\bm{c})}\right\}\qquad\text{where}\qquad\varepsilon_\lambda(\bm{c}):=\sum_{\substack{i<j\\i,j\in\lambda}}\mathfrak{q}_i\mathfrak{q}_j~.$$
\item[(c)] For each $\spl\in\mathcal{S}_N$, define the \emph{level abscissa} $\mathcal{LP}_\spl^{\bm{c}}$ by $$\mathcal{LP}_\spl^{\bm{c}}(b):=-\inf_{1\leq\ell\leq L(\spl)-1}\left\{\frac{\rank(\ptn_\ell)+\re(b)}{\mathcal{E}_{\ptn_\ell}(\bm{c})}\right\}\qquad\text{where}\qquad\mathcal{E}_{\ptn_\ell}(\bm{c}):=\sum_{\lambda\in\ptn_\ell}\varepsilon_\lambda(\bm{c})~.$$
\end{itemize}
\end{definition}

If $\beta\in\C$ and $\bm{c}$ is defined as above, \Cref{rootdef,splchdef,branchlevel,absdef} together imply
\begin{align*}
\beta\bm{c}\in\mathcal{RP}_N(a,b)\qquad&\iff\qquad\re(\beta)>\mathcal{RP}_N^{\bm{c}}(a,b)~,\\
\beta\bm{c}\in\mathcal{BP}_\spl\qquad&\iff\qquad\re(\beta)>\mathcal{BP}_\spl^{\bm{c}}~,\\
\beta\bm{c}\in\mathcal{LP}_\spl(b)\qquad&\iff\qquad\re(\beta)>\mathcal{LP}_\spl^{\bm{c}}(b)~,
\end{align*}
and hence the convergence criteria for $\bm{s}$ in \Cref{main} become criteria for $\beta$ when $\bm{s}=\beta\bm{c}$. The following corollary comes straight from this observation and \Cref{main}:

\begin{corollary}\label{univariate}
Fix $N\geq 2$, $a,b\in\C$, and $\bm{c}=(\mathfrak{q}_i\mathfrak{q}_j)_{i<j}$ where $\mathfrak{q}_1,\mathfrak{q}_2,\dots,\mathfrak{q}_N>0$. Given a $p$-field $K$, a function $\rho:\mathcal{N}\to\C$ satisfying \eqref{rhogro}, and $\beta\in\C$, recall $$Z_N^\rho(K,a,b,\beta\bm{c})=\int_{K^N}\rho(\|x\|)\big(\max_{i<j}|x_i-x_j|\big)^a\big(\min_{i<j}|x_i-x_j|\big)^b\prod_{i<j}|x_i-x_j|^{\mathfrak{q}_i\mathfrak{q}_j\beta}\,|dx|~.$$
\begin{itemize}
\item[(a)] If $K$ has residue field cardinality $q$, the integral above converges absolutely to $$Z_N^\rho(K,a,b,\beta\bm{c})=H_q^\rho\bigg(N+a+b+\sum_{i<j}\mathfrak{q}_i\mathfrak{q}_j\beta\bigg)\cdot\sum_{\spl\in\mathcal{S}_N}J_{\spl,q}(b,\beta\bm{c})$$ when $$\re(\beta)>\sup\left\{\mathcal{RP}_N^{\bm{c}}(a,b),~\sup_{\spl\in\mathcal{S}_N}\mathcal{LP}_\spl^{\bm{c}}(b)\right\}~.$$
\item[(b)] For the same $K$, if $b=0$, the integral above converges absolutely to $$Z_N^\rho(K,a,0,\beta\bm{c})=H_q^\rho\bigg(N+a+\sum_{i<j}\mathfrak{q}_i\mathfrak{q}_j\beta\bigg)\cdot\sum_{\spl^*\in\mathcal{R}_N}I_{\spl^*\!\!,q}(\beta\bm{c})$$ when $$\re(\beta)>\sup\left\{\mathcal{RP}_N^{\bm{c}}(a,0),~\sup_{\spl^*\in\mathcal{R}_N}\mathcal{BP}_{\spl^*}^{\bm{c}}\right\}~.$$
\end{itemize}
\end{corollary}

Before concluding this section with formulas for the analogue of Mehta's integral and the expectation in \eqref{Jointeq}, we remark on the one-component case, namely $\mathfrak{q}_1=\mathfrak{q}_2=\dots=\mathfrak{q}_N=1$. In this case $\bm{c}=\bm{1}$ is simply the $\binom{N}{2}$-tuple of $1$'s, and for each $\spl\in\mathcal{S}_N$ it is easily verified that $$e_\lambda(\beta\bm{1})=\#\lambda-1+\varepsilon_\lambda(\bm{1})\beta=\binom{\#\lambda}{2}\left(\beta+\frac{2}{\#\lambda}\right)$$ for all $\lambda\in\mathcal{B}(\spl)$ and $$E_{\ptn_\ell}(\beta\bm{1})=\sum_{\lambda\in\ptn_\ell}e_\lambda(\beta\bm{1})=\sum_{\lambda\in\ptn_\ell}\binom{\#\lambda}{2}\left(\beta+\frac{2}{\#\lambda}\right)$$ for all $\ell\in\{0,1,\dots,L(\spl)-1\}$. The exponents above have no dependence on the particular labels $1,2,\dots,N$, so we shall take a moment to discuss a relationship between $\mathcal{S}_N$ and the symmetric group action on the label set $\{1,2,\dots,N\}$.

\begin{definition}\label{sym_action}
Denote the symmetric group on $[N]=\{1,2,\dots,N\}$ by $\Sym([N])$. Given $\sigma\in\Sym([N])$ and a nonempty subset $\lambda=\{i_1,i_2,\dots,i_k\}\subset[N]$, we write $\sigma(\lambda):=\{\sigma(i_1),\sigma(i_2),\dots,\sigma(i_k)\}$, for a partition $\ptn=\{\lambda_1,\lambda_2,\dots,\lambda_n\}$ of $[N]$ we write $\sigma(\ptn):=\{\sigma(\lambda_1),\sigma(\lambda_2),\dots,\sigma(\lambda_n)\}$, and finally, for each $\spl=(\ptn_0,\ptn_1,\dots,\ptn_{L(\spl)})\in\mathcal{S}_N$ we write $\sigma(\spl):=(\sigma(\ptn_0),\sigma(\ptn_1),\dots,\sigma(\ptn_{L(\spl)}))$.
\end{definition} 

If $\Aut(\mathcal{S}_N)$ denotes the group of bijections $\mathcal{S}_N\to\mathcal{S}_N$, the homomorphism $\Sym([N])\to\Aut(\mathcal{S}_N)$ given by $\sigma\mapsto(\spl\mapsto\sigma(\spl))$ is an action of $\Sym([N])$ on $\mathcal{S}_N$. The following properties of this action are clear from \Cref{splchdef,branchlevel}: If $\spl\in\mathcal{S}_N$ and $\sigma\in\Sym([N])$, then
\begin{itemize}
\item $L(\sigma(\spl))=L(\spl)$, and $\sigma(\spl)=\spl$ if and only if $\sigma(\ptn_\ell)=\ptn_\ell$ for all $\ell\in\{0,1,\dots,L(\spl)\}$,
\item $\sigma(\lambda)\in\mathcal{B}(\sigma(\spl))$ if and only if $\lambda\in\mathcal{B}(\spl)$,
\item for each $\lambda\in\mathcal{B}(\spl)$ we have $\#\sigma(\lambda)=\#\lambda$, $\ell_{\sigma(\spl)}(\sigma(\lambda))=\ell_\spl(\lambda)$, and $\deg_{\sigma(\spl)}(\sigma(\lambda))=\deg_\spl(\lambda)$, so
\item $M_{\sigma(\spl)}(t)=M_\spl(t)$, $e_{\sigma(\lambda)}(\beta\bm{1})=e_\lambda(\beta\bm{1})$ for all $\lambda\in\mathcal{B}(\spl)$, and hence $E_{\sigma(\ptn_\ell)}(\beta\bm{1})=E_{\ptn_\ell}(\beta\bm{1})$ for all $\ell\in\{0,1,\dots,L(\spl)-1\}$.
\end{itemize}
In particular, the action of $\Sym([N])$ on $\mathcal{S}_N$ restricts to a well-defined action on $\mathcal{R}_N$.

\begin{definition}\label{orb_stab_wt}
For each $\spl\in\mathcal{S}_N$, define the \emph{orbit}, \emph{stabilizer}, and \emph{weight} of $\spl$ respectively by $$\Orb(\spl):=\{\sigma(\spl):\sigma\in\Sym([N])\}~,\qquad\Stab(\spl):=\{\sigma\in\Sym([N]):\sigma(\spl)=\spl\}~,$$ and $$W(\spl):=\#\Orb(\spl)=\frac{N!}{\#\Stab(\spl)}~.$$
\end{definition}

\Cref{branchlevel,orb_stab_wt} and the properties of the action immediately imply the following:

\begin{lemma}\label{liketerms}
Suppose $q\geq 2$, $b\in\C$, and $\spl\in\mathcal{S}_N$. 
\begin{itemize}
\item[(a)] For each $\beta$ in the domain of $\beta\mapsto I_{\spl,q}(\beta\bm{1})$ we have $$\sum_{\spl'\in\Orb(\spl)}I_{\spl'\!\!,q}(\beta\bm{1})=W(\spl)I_{\spl,q}(\beta\bm{1})=\frac{W(\spl)M_\spl(q)}{q^{N-1}}\cdot\prod_{\lambda\in\mathcal{B}(\spl)\setminus\overline{\ptn}}\frac{1}{q^{\binom{\#\lambda}{2}\left(\beta+\frac{2}{\#\lambda}\right)}-1}~.$$
\item[(b)] For each $\beta$ in the domain of $\beta\mapsto J_{\spl,q}(b,\beta\bm{1})$ we have $$\sum_{\spl'\in\Orb(\spl)}J_{\spl'\!\!,q}(b,\beta\bm{1})=W(\spl)J_{\spl,q}(b,\beta\bm{1})=\frac{W(\spl)M_\spl(q)}{q^{N-1}}\cdot\prod_{\ell=1}^{L(\spl)-1}\frac{1}{q^{b+\sum_{\lambda\in\ptn_\ell}\binom{\#\lambda}{2}\left(\beta+\frac{2}{\#\lambda}\right)}-1}~.$$
\end{itemize}
\end{lemma}

\begin{remark}\label{-2/N}
If $\mathcal{C}_N\subset\mathcal{S}_N$ is a complete set of orbit representatives for the action of $\Sym([N])$ on $\mathcal{S}_N$, then $\mathcal{C}_N\cap\mathcal{R}_N$ is a complete set of orbit representatives for the restricted action on $\mathcal{R}_N$. Then by part (a) of \Cref{liketerms}, the sum over $\spl\in\mathcal{S}_N$ appearing in the main formula for $Z_N^\rho(K,a,b,\beta\bm{1})$ can be grouped into a weighted sum over $\mathcal{C}_N$. Similarly, part (b) of \Cref{liketerms} implies that the sum over $\spl^*\in\mathcal{R}_N$ in the formula for $Z_N^\rho(K,a,0,\beta\bm{1})$ can be grouped into a weighted sum over $\mathcal{C}_N\cap\mathcal{R}_N$. From the viewpoint of log-Coulomb gas, the appearance of these weighted sums has an intuitive explanation: The condition $\mathfrak{q}_1=\mathfrak{q}_2=\dots=\mathfrak{q}_N=1$ makes the particles of the gas identical and imposes symmetries on the set of microstates $x\in K^N$. Each $\spl\in\mathcal{C}_N$ represents a distinct symmetry class of microstates, the factor $\frac{W(\spl)M_\spl(q)}{q^{N-1}}$ can be regarded as its weight, and the two products of rational functions of $q^{-\beta}$ appearing in \Cref{liketerms} are its respective contributions to the functions $\beta\mapsto Z_N^\rho(K,a,0,\beta\bm{1})$ and $\beta\mapsto Z_N^\rho(K,a,b,\beta\bm{1})$. In particular, each symmetry class contributes a weighted term to the canonical partition function $\beta\mapsto\mathcal{Z}_N(\beta)=Z_N^\rho(K,0,0,\beta\bm{1})$. It is also worth noting that the condition on $\re(\beta)$ in part (b) of \Cref{univariate} simplifies further when $a=b=0$ and $\bm{c}=\bm{1}$. Indeed, for general $\bm{c}=(\mathfrak{q}_i\mathfrak{q}_j)_{i<j}$ we have
\begin{equation}\label{absup}
\sup\left\{\mathcal{RP}_N^{\bm{c}}(0,0),~\sup_{\spl^*\in\mathcal{R}_N}\mathcal{BP}_{\spl^*}^{\bm{c}}\right\}=-\inf_{\spl^*\in\mathcal{R}_N}\left\{\inf_{\lambda\in\mathcal{B}(\spl^*)}\left\{\frac{\#\lambda-1}{\sum_{\substack{i<j\\i,j\in\lambda}}\mathfrak{q}_i\mathfrak{q}_j}\right\}\right\}~,
\end{equation}
so if $\spl^*\in\mathcal{R}_N$ and $\bm{c}=\bm{1}$ we have $$\frac{\#\lambda-1}{\sum_{\substack{i<j\\i,j\in\lambda}}\mathfrak{q}_i\mathfrak{q}_j}=\frac{\#\lambda-1}{\binom{\#\lambda}{2}}=\frac{2}{\#\lambda}\qquad\text{for all }\lambda\in\mathcal{B}(\spl^*)~.$$ Thus if $\bm{c}=\bm{1}$, the inner infima in \eqref{absup} are all $\frac{2}{N}$, so the quantity in \eqref{absup} is simply $-\frac{2}{N}$.
\end{remark}

Combining \Cref{liketerms} and \Cref{-2/N} yields the $p$-field analogue of Mehta's integral formula:

\begin{theorem}[Mehta's integral formula for $p$-fields]\label{pfieldMehta}
Suppose $K$ is a $p$-field with residue field cardinality $q$, suppose $\rho:\mathcal{N}\to\C$ satisfies \eqref{rhogro}, and let $\bm{c}=(\mathfrak{q}_i\mathfrak{q}_j)_{i<j}$ where $\mathfrak{q}_1,\mathfrak{q}_2,\dots,\mathfrak{q}_N>0$.
\begin{itemize}
\item[(a)] If $\beta$ is any complex number satisfying $$\re(\beta)>-\inf_{\spl^*\in\mathcal{R}_N}\left\{\inf_{\lambda\in\mathcal{B}(\spl^*)}\left\{\frac{\#\lambda-1}{\sum_{\substack{i<j\\i,j\in\lambda}}\mathfrak{q}_i\mathfrak{q}_j}\right\}\right\}~,$$ then $$\mathcal{Z}_N(\beta)=\int_{K^N}\rho(\|x\|)\prod_{i<j}|x_i-x_j|^{\mathfrak{q}_i\mathfrak{q}_j\beta}\,|dx|=H_q^\rho\left(N+\sum_{i<j}\mathfrak{q}_i\mathfrak{q}_j\beta\right)\cdot\sum_{\spl^*\in\mathcal{R}_N}I_{\spl^*\!\!,q}(\beta\bm{c})~.$$
\item[(b)] In particular, if $\mathfrak{q}_1=\mathfrak{q}_2=\dots=\mathfrak{q}_N=1$ and $\re(\beta)>-\frac{2}{N}$, then
\begin{align*}
\mathcal{Z}_N(\beta)=\frac{1-q^{-\binom{N}{2}\left(\beta+\frac{2}{N-1}\right)}}{1-q^{-\binom{N}{2}\left(\beta+\frac{2}{N}\right)}}&\cdot\sum_{m\in\Z}\rho(q^m)q^{m\binom{N}{2}\left(\beta+\frac{2}{N-1}\right)}\\
&\cdot\sum_{\spl^*\in\mathcal{C}_N\cap\mathcal{R}_N}\frac{W(\spl^*)M_{\spl^*\!}(q)}{q^{N-1}}\prod_{\lambda\in\mathcal{B}(\spl^*)\setminus\overline{\ptn}}\frac{1}{q^{\binom{\#\lambda}{2}\left(\beta+\frac{2}{\#\lambda}\right)}-1}~,
\end{align*}
where $\mathcal{C}_N\subset\mathcal{S}_N$ is a full set of orbit representatives for the action of $\Sym([N])$ on $\mathcal{S}_N$.
\end{itemize}
\end{theorem}

If the $\rho$ above is also nonzero and nonnegative, then $\mathcal{Z}_N(\beta)\in(0,\infty)$ for all $\beta>0$ and the function $x\mapsto\frac{1}{\mathcal{Z}_N(\beta)}\rho(\|x\|)\prod_{i<j}|x_i-x_j|^{\mathfrak{q}_i\mathfrak{q}_j\beta}$ is a well-defined probability density on the microstates $x\in K^N$. Moreover, none of the abscissae in \Cref{absdef} are positive if both $\re(b)\geq-1$ and $\re(a+b)\geq1-N$, in which case the conditions on $\re(\beta)$ in \Cref{univariate} are met by all $\beta>0$. This observation and \eqref{Jointeq} provide the following corollary:

\begin{corollary}\label{EV}
Suppose $K$ is a $p$-field with residue field cardinality $q$, suppose $\rho:\mathcal{N}\to\R_{\geq 0}$ is a nonzero norm-density satisfying \eqref{rhogro}, and let $\bm{c}=(\mathfrak{q}_i\mathfrak{q}_j)_{i<j}$ where $\mathfrak{q}_1,\mathfrak{q}_2,\dots,\mathfrak{q}_N>0$.
\begin{itemize}
\item[(a)] If $\re(b)\geq-1$ and $\re(a+b)\geq1-N$, then for any inverse temperature $\beta>0$ we have
\begin{align*}
\E\left[\big(\max_{i<j}|x_i-x_j|\big)^a\big(\min_{i<j}|x_i-x_j|\big)^b\right]&=\dfrac{H_q^\rho\bigg(N+a+b+\sum_{i<j}\mathfrak{q}_i\mathfrak{q}_j\beta\bigg)\cdot\sum_{\spl^*\in\mathcal{S}_N}J_{\spl,q}(b,\beta\bm{c})}{H_q^\rho\bigg(N+\sum_{i<j}\mathfrak{q}_i\mathfrak{q}_j\beta\bigg)\cdot\sum_{\spl^*\in\mathcal{S}_N}J_{\spl,q}(0,\beta\bm{c})}\\
&=\dfrac{H_q^\rho\bigg(N+a+b+\sum_{i<j}\mathfrak{q}_i\mathfrak{q}_j\beta\bigg)\cdot\sum_{\spl^*\in\mathcal{S}_N}J_{\spl,q}(b,\beta\bm{c})}{H_q^\rho\bigg(N+\sum_{i<j}\mathfrak{q}_i\mathfrak{q}_j\beta\bigg)\cdot\sum_{\spl^*\in\mathcal{R}_N}I_{\spl^*\!\!,q}(\beta\bm{c})}~.
\end{align*}
\item[(b)] In particular, if $b=0$ and $\re(a)\geq1-N$, then for any inverse temperature $\beta>0$ we have $$\E\left[\big(\max_{i<j}|x_i-x_j|\big)^a\right]=\dfrac{H_q^\rho\bigg(N+a+\sum_{i<j}\mathfrak{q}_i\mathfrak{q}_j\beta\bigg)}{H_q^\rho\bigg(N+\sum_{i<j}\mathfrak{q}_i\mathfrak{q}_j\beta\bigg)}~.$$
\end{itemize}
\end{corollary}

As mentioned at the end of \Cref{1_1}, applying part (a) of \Cref{EV} to $a,b\in\Z_{\geq 0}$ gives the joint moments of the random variables $\max_{i<j}|x_i-x_j|$ and $\min_{i<j}|x_i-x_j|$. In particular, the average value in part (b) of \Cref{EV} can be computed without the use of branch or level functions, and thus admits a simple closed form for suitably chosen $\rho$. The next example demonstrates this and addresses the low-temperature limit (i.e., $\beta\to\infty$) in the $b=0$ case.

\begin{example}
Recall that $\|K^N\setminus\{0\}\|=q^\Z$ if the residue field of $K$ has cardinality $q$, and let $\rho$ be the norm-density defined by $\rho(t)=\bm{1}_{[0,q^M]}(t)$ where $M\in\Z$. Since $\rho(\|x\|)=1$ if and only if all $x_i$ are in the disk $\{y\in K:|y|\leq q^M\}$ and otherwise $\rho(\|x\|)=0$, $\rho$ guarantees that the charges are almost surely confined to this disk, and by \Cref{rootdef} we have $$H_q^\rho(z)=\frac{1-q^{-z}}{1-q^{-(z-1)}}\cdot\sum_{m=-M}^\infty(q^{-z})^m=\frac{q^{Mz}}{1-q^{-(z-1)}}\qquad\text{for }\re(z)>1~.$$ Then for $\re(a)\geq 1-N$ part (b) of \Cref{EV} gives the explicit formula $$\E\left[\big(\max_{i<j}|x_i-x_j|\big)^a\right]=\dfrac{\frac{q^{M(N+a+\sum_{i<j}\mathfrak{q}_i\mathfrak{q}_j\beta)}}{1-q^{-(N-1+a+\sum_{i<j}\mathfrak{q}_i\mathfrak{q}_j\beta)}}}{\frac{q^{M(N+\sum_{i<j}\mathfrak{q}_i\mathfrak{q}_j\beta)}}{1-q^{-(N-1+\sum_{i<j}\mathfrak{q}_i\mathfrak{q}_j\beta)}}}=q^{Ma}\cdot\dfrac{q^{N-1+\sum_{i<j}\mathfrak{q}_i\mathfrak{q}_j\beta}-1}{q^{N-1+\sum_{i<j}\mathfrak{q}_i\mathfrak{q}_j\beta}-q^{-a}}~,$$ from which the following asymptotic estimate is clear: $$\E\left[\big(\max_{i<j}|x_i-x_j|\big)^a\right]\sim q^{Ma}\qquad\text{as $N\to\infty$ or $\beta\to\infty$}~.$$ (By taking $N\to\infty$, we are assuming here that a charge $\mathfrak{q}_i>0$ has been specified for every $i\in\N$.) Since $\max_{i<j}|x_i-x_j|\leq q^M$ almost surely, this estimate implies that a gas comprised of many particles and/or held at a low temperature has a relatively high probability of attaining microstates $x\in K^N$ with $\max_{i<j}|x_i-x_j|=q^M$. Roughly speaking, this says the gas is very likely to spread out as widely as possible if it is cold and/or if it has many particles.  
\end{example}

\begin{remark}
The previous example hints at a more general feature of low-temperature limits: Suppose $\rho$ is a compactly supported nonzero norm-density satisfying \eqref{rhogro}. There is a greatest $M\in\Z$ for which $\rho(q^M)\neq0$, so given $\delta>1$ the scaled sum $\frac{H_q^\rho(z)}{q^{Mz}}=\frac{1-q^{-z}}{1-q^{-(z-1)}}\cdot\sum_{m=-M}^\infty\rho(q^{-m})q^{-(m+M)z}$ converges uniformly for $\re(z)\geq\delta$ by \eqref{rhogro}. Therefore we may take $z\to\infty$ term-by-term to obtain $\lim_{z\to\infty}\frac{H_q^\rho(z)}{q^{Mz}}=\rho(q^M)$, and the ratio of root functions in part (a) of \Cref{EV} satisfies
\begin{equation}\label{root_ratio_limit}
\lim_{\beta\to\infty}\dfrac{H_q^\rho\bigg(N+a+b+\sum_{i<j}\mathfrak{q}_i\mathfrak{q}_j\beta\bigg)}{H_q^\rho\bigg(N+\sum_{i<j}\mathfrak{q}_i\mathfrak{q}_j\beta\bigg)}=\lim_{\beta\to\infty}\dfrac{q^{M(a+b)}\cdot\frac{H_q^\rho\bigg(N+a+b+\sum_{i<j}\mathfrak{q}_i\mathfrak{q}_j\beta\bigg)}{q^{M(N+a+b+\sum_{i<j}\mathfrak{q}_i\mathfrak{q}_j\beta)}}}{\frac{H_q^\rho\bigg(N+\sum_{i<j}\mathfrak{q}_i\mathfrak{q}_j\beta\bigg)}{q^{M(N+\sum_{i<j}\mathfrak{q}_i\mathfrak{q}_j\beta)}}}=q^{M(a+b)}~.
\end{equation}
The ratio of sums in part (a) of \Cref{EV} also converges. More precisely, for each $\spl\in\mathcal{S}_N$, define $$Q_\spl(\bm{c}):=\sum_{\ell=1}^{L(\spl)-1}\mathcal{E}_{\ptn_\ell}(\bm{c})=\sum_{\ell=1}^{L(\spl)-1}\sum_{\lambda\in\ptn_\ell}\sum_{\substack{i,j\in\lambda\\i<j}}\mathfrak{q}_i\mathfrak{q}_j~.$$ For any $q\geq 2$, the set $\mathcal{S}_{N,q}=\{\spl\in\mathcal{S}_N:M_\spl(q)>0\}$ contains the splitting chain $$\spl=([N],~[N-1]\{N\},~[N-2]\{N-1\}\{N\},~\dots,\{1\}\{2\}\dots\{N\})~,$$ so $\mathcal{S}_{N,q}\neq\varnothing$ and hence a non-negative minimum $Q_{N,q}^{\min}(\bm{c}):=\min\{Q_\spl(\bm{c}):\spl\in\mathcal{S}_{N,q}\}$ exists. Taking $\beta\to\infty$ gives
\begin{align*}
J_{\spl,q}(b,\beta\bm{c})&=\frac{M_\spl(q)}{q^{N-1}}\cdot\prod_{\ell=1}^{L(\spl)-1}\frac{1}{q^{b+E_{\ptn_\ell}(\beta\bm{c})}-1}\\
&\sim\frac{M_\spl(q)}{q^{N-1}}\cdot q^{-\sum_{\ell=1}^{L(\spl)-1}(b+E_{\ptn_\ell}(\beta\bm{c}))}=\frac{M_\spl(q)}{q^{N-1+\sum_{\ell=1}^{L(\spl)-1}(b+\rank(\ptn_\ell))}}\cdot q^{-\beta Q_\spl(\bm{c})}~,
\end{align*}
and therefore $$\sum_{\spl\in\mathcal{S}_N}J_{\spl,q}(b,\beta\bm{c})\sim\sum_{\substack{\spl\in\mathcal{S}_{N,q}\\Q_\spl(\bm{c})=Q_N^{\min}(\bm{c})}}\frac{M_\spl(q)}{q^{N-1+\sum_{\ell=1}^{L(\spl)-1}(b+\rank(\ptn_\ell))}}\cdot q^{-\beta Q_{N,q}^{\min}(\bm{c})}~.$$ The factors $q^{-(N-1)}$ and $q^{-\beta Q_{N,q}^{\min}(\bm{c})}$ are independent of $b$ and common to all terms in the right-hand sum, so we may abbreviate the above summation by $\sum'$ and obtain
\begin{equation}\label{level_ratio_limit}
\lim_{\beta\to\infty}\dfrac{\sum_{\spl\in\mathcal{S}_N}J_{\spl,q}(b,\beta\bm{c})}{\sum_{\spl\in\mathcal{S}_N}J_{\spl,q}(0,\beta\bm{c})}=\dfrac{\sum'M_\spl(q)q^{-\sum_{\ell=1}^{L(\spl)-1}(\rank(\ptn_\ell)+b)}}{\sum'M_\spl(q)q^{-\sum_{\ell=1}^{L(\spl)-1}\rank(\ptn_\ell)}}~.
\end{equation}
Combining this with part (a) of \Cref{EV} and \eqref{root_ratio_limit} gives the low-temperature limit of any joint moment:
\begin{equation}\label{lowtemp_ab}
\lim_{\beta\to\infty}\E\left[\big(\max_{i<j}|x_i-x_j|\big)^a\big(\min_{i<j}|x_i-x_j|\big)^b\right]=q^{M(a+b)}\cdot\dfrac{\sum'M_\spl(q)q^{-\sum_{\ell=1}^{L(\spl)-1}(\rank(\ptn_\ell)+b)}}{\sum'M_\spl(q)q^{-\sum_{\ell=1}^{L(\spl)-1}\rank(\ptn_\ell)}}~.
\end{equation}
\end{remark}

Explicit computation of \eqref{lowtemp_ab} is generally impractical as it depends on $N$, $q$, and $\bm{c}$ in very complicated ways. Still, it is interesting that the ratio of sums in \eqref{lowtemp_ab} is a weighted average of the finite set of values $$\{q^{-b(L(\spl)-1)}:\spl\in\mathcal{S}_N\text{ with }M_\spl(q)>0\text{ and }Q_\spl(\bm{c})=Q_{N,q}^{\min}(\bm{c})\}~,$$ with each weight $M_\spl(q)q^{-\sum_{\ell=1}^{L(\spl)-1}\rank(\ptn_\ell)}$ independent of $a$, $b$, and $\rho$. Moreover, if $q\geq N$, then the splitting chain $\spl=([N],~\{1\}\{2\}\dots\{N\})\in\mathcal{S}_N$ has $M_\spl(q)=(q-1)_{N-1}>0$ and $Q_\spl(\bm{c})=Q_{N,q}^{\min}(\bm{c})=0$, and in fact it is the only one satisfying $Q_\spl(\bm{c})=Q_{N,q}^{\min}(\bm{c})$. Thus for $q\geq N$ we have $$\lim_{\beta\to\infty}\sum_{\spl\in\mathcal{S}_N}J_{\spl,q}(b,\beta\bm{c})=(q-1)_{N-1}>0$$ and we obtain a final corollary:

\begin{corollary}
Suppose $K$ is a $p$-field with residue field cardinality $q\geq N$ and suppose $\re(b)\geq-1$ and $\re(a+b)\geq 1-N$. Then if $\rho$ is a compactly supported nonzero norm-density satisfying \eqref{rhogro} and $M$ is the largest integer satisfying $\rho(q^M)\neq 0$, we have $$\lim_{\beta\to\infty}\E\left[\big(\max_{i<j}|x_i-x_j|\big)^a\big(\min_{i<j}|x_i-x_j|\big)^b\right]=q^{M(a+b)}~.$$
\end{corollary} 

\section{The proof of the main theorem}
In this section we let $K$ be a fixed $p$-field with $\mu$, $|\cdot|$, $\|\cdot\|$, $R$, and $P$ as defined in \Cref{1_1}. The steps in the proof are divided as follows. In \Cref{3_1} we briefly recall well-known properties of $K$ (see \cite{Weil}, for example), and in \Cref{3_2} we use these properties to explain how elements of $R^N$ may be visualized as trees. In \Cref{3_3} we find the key connection between these trees and splitting chains, then use it to show that certain integrals can be expressed in terms of level functions in \Cref{spl_integral}. \Cref{3_3} culminates with \Cref{pair_conclusion}, which establishes parts (a) and (b) of \Cref{main} in the special case $\rho=\bm{1}_{[0,1]}$. \Cref{3_4} provides the foundation for part (c) of \Cref{main} with a ``branch-centric" analogue of \Cref{pair_conclusion}, namely \Cref{bspl_conclusion}, establishing \Cref{reduction} in the process. Finally, in \Cref{3_5} we prove that \Cref{pair_conclusion} and \Cref{bspl_conclusion} extend to general $\rho$ via \Cref{rho_thm}, and conclude the proof of \Cref{main}. 
\subsection{Basic properties of $p$-fields}\label{3_1}
\begin{proposition}\label{pfield_props}\
\begin{itemize}
\item[(a)] (The strong triangle inequality and equality.) Every pair of elements $x,y\in K$ satisfies the inequality $|x+y|\leq\max\{|x|,|y|\}$. It becomes equality if $|x|\neq|y|$.
\item[(b)] The closed ball $R$ is a local PID, the open ball $P$ is its unique maximal ideal, and the unit group is $R^\times=R\setminus P=\{x\in K:|x|=1\}$.
\item[(c)] The \emph{residue field} $\kappa:=R/P$ is isomorphic to $\F_q$ for some prime power $q\geq 2$.
\item[(d)] The canonical absolute value $|\cdot|$ restricts to a surjective homomorphism $K^\times\to q^\Z$ and satisfies $|x|=\mu(xR)$ for every $x\in K$.
\item[(e)] The fraction field of $R$ is $K$, in which the fractional ideals of $R$ are precisely the balls $$P^m=\{x\in K:|x|\leq q^{-m}\}~,\qquad m\in\Z~.$$ Moreover, every ball in $K$ is open, compact, of the form $y+P^m=\{x\in K:|x-y|\leq q^{-m}\}$ for some $m\in\Z$ and $y\in K$, and with measure $\mu(y+P^m)=q^{-m}$.
\end{itemize}
\end{proposition}

\noindent The strong triangle inequality and equality distinguish $K$ from its archimedean counterparts in striking ways. To name a few, any two open balls in $K$ are either nested or disjoint, $K$ is totally disconnected, and $|1+1+\dots+1|\leq 1$ for any finite sum of $1$'s (this is why $K$ and $|\cdot|$ are called \emph{nonarchimedean}). Of particular contrast and importance is the countability of the set $|K|=q^\Z\cup\{0\}$. This fact implies $\|K^N\|=q^\Z\cup\{0\}\subset\mathcal{N}$, motivates the next definition, and implies the following corollary.

\begin{definition}
The \emph{canonical valuation} is the surjective function $v:K\to\Z\cup\{\infty\}$ defined by $$v(x):=\begin{cases}-\log_q|x|&\text{if $x\neq 0$},\\\infty&\text{if $x=0$.}\end{cases}$$ A \emph{uniformizer} for $v$ is any element $\pi\in K$ satisfying $v(\pi)=1$ or equivalently $\pi\in P\setminus P^2$.
\end{definition}

\begin{corollary}\label{val_unif}\
\begin{itemize}
\item[(a)] The canonical valuation restricts to a surjective homomorphism $K^\times\to\Z$ and satisfies the inequality $v(x+y)\geq\min\{v(x),v(y)\}$ for all $x,y\in K$. It becomes equality if $v(x)\neq v(y)$.
\item[(b)] Suppose $\pi\in K$ is a uniformizer. Then $P^m=\pi^mR=\{x\in K:v(x)\geq m\}$ for all $m\in\Z$. In particular, $|x|=q^{-m}\iff v(x)=m$, and in this case $x=\pi^mu$ for a unique $u\in R^\times$.
\end{itemize}
\end{corollary}

Note that $|\cdot|$, $v$, $R$, $P$, $q$, and the family of additive Haar measures on $K$ are all canonical in the sense that they are completely determined by $K$. In fact, the only choice we have insisted on so far is our particular Haar measure $\mu$, for it satisfies the convenient identity $\mu(xR)=|x|$ and hence takes values in $q^\Z\cup\{0\}$. We will now make two more choices in order to apply the following proposition consistently in upcoming proofs. Namely, fix a uniformizer $\pi\in K$ and a set of representatives $D\subset R$ for $\kappa=R/P$ such that $0\in D$.

\begin{proposition}\label{sum_rep}
For each $x\in R$ there is a unique sequence $(d(0),d(1),d(2),\dots)$ in $D$ such that $$x=\sum_{n=0}^\infty\pi^nd(n)~,$$ and this series is absolutely convergent with respect to $|\cdot|$. In this case $v(x)=\inf\{n:d(n)\neq 0\}$, and if $(d'(0),d'(1),d'(2),\dots)$ is the corresponding sequence for $y\in R$ then $v(x-y)=\inf\{n:d(n)\neq d'(n)\}$. Moreover, given $m\in\N$, the collection of partial sums $\{\sum_{n=0}^{m-1}\pi^nd(n):d(n)\in D\}$ is a full set of representatives for the quotient $R/P^m=R/\pi^mR$.
\end{proposition}

\begin{remark}\label{val_m}
In light of \Cref{sum_rep}, if $x,y\in R$ have series representations $x=\sum_{n=0}^\infty\pi^nd(n)$ and $y=\sum_{n=0}^\infty\pi^nd'(n)$, we may henceforth use the following equivalent statements interchangeably:
\begin{itemize}
\item $|x-y|\leq q^{-m}$,
\item $v(x-y)\geq m$,
\item $\inf\{n:d(n)\neq d'(n)\}\geq m$,
\item $x\equiv y\mod\pi^m$.
\end{itemize}
\end{remark}

\subsection{The tree part of a series representation}\label{3_2}
With $\pi$, $D$, and \Cref{sum_rep} in hand, we can now present a method for decomposing and visualizing elements $x\in R^N\setminus V_0$, where $N\geq 2$ is fixed and $V_0:=\{x\in R^N:x_i=x_j\text{ for some $i<j$}\}$. Given $x=(x_1,x_2,\dots,x_N)\in R^N$, \Cref{sum_rep} provides a unique sequence $(d_i(0),d_i(1),d_i(2),\dots)$ in $D$ satisfying $x_i=\sum_{n=0}^\infty\pi^nd_i(n)$ for each entry $x_i$. This gives a unique series representation for $x$, namely $$x=\sum_{n=0}^\infty\pi^nd(n)\quad\text{ where }\quad d(n)=(d_1(n),d_2(n),\dots,d_N(n))\in D^N~,$$ and this series converges absolutely in $R^N$. Moreover, given $m\in\N$, $\{\sum_{n=0}^{m-1}\pi^nd(n):d(n)\in D^N\}$ is a complete set of representatives for the quotient $R^N/\pi^mR^N$, so we will abuse notation and write $$R^N/\pi^mR^N=\left\{\sum_{n=0}^{m-1}\pi^nd(n):d(n)\in D^N\right\}~.$$ Given $x=\sum_{n=0}^\infty\pi^nd(n)\in R^N$ and $m\in\N$, it is clear that the unique elements $y\in R^N/\pi^mR^N$ and $z\in\pi^mR^N$ satisfying $x=y+z$ are respectively $y=\sum_{n=0}^{m-1}\pi^nd(n)$ and $z=\sum_{n=m}^\infty\pi^nd(n)$. The following definition makes use of this and the following observation: $$x\in R^N\setminus V_0\quad\iff\quad x\in R^N\text{ and }\sup_{i<j}v(x_i-x_j)<\infty~.$$

\begin{definition}
We call an element $y\in R^N\setminus V_0$ a \emph{tree} of length $m\in\N$ if $$y\in R^N/\pi^mR^N\qquad\text{and}\qquad m=\max_{i<j}v(y_i-y_j)+1~.$$
\end{definition}

Given $x=\sum_{n=0}^\infty\pi^nd(n)\in R^N\setminus V_0$ with $m=\max_{i<j}v(x_i-x_j)+1$, note that $y=\sum_{n=0}^{m-1}\pi^nd(n)$ is the unique partial sum of $x$ that forms a tree, so $y$ will accordingly be called the \emph{tree part} of $x$. The reason for the name ``tree" is clarified by the next example, which will be revisited during the proofs of the main theorems.

\begin{example}\label{Extree_1}
Suppose $N=9$ and $K=\Q_5$ with uniformizer $\pi=5$ and digit set $D=\{0,1,2,3,4\}$. The tree $y=\sum_{n=0}^75^nd(n)$ corresponding to the digit vectors $d(0),d(1),\dots,d(7)$ in \Cref{fig_1} can be visualized as a rooted tree. The root represents the value 0, and the nodes traversed by the path from the root down to the leaf $y_i$ represent the consecutive partial sums of $y_i=\sum_{n=0}^75^nd_i(n)$. It should be noted that for general trees $y\in R^N\setminus V_0$, the corresponding diagram need not have $y_i$ in index order at the bottom. The tree in this example was only chosen this way only to make the diagram in \Cref{fig_1} easily discernible from the digits.
\end{example}

\begin{figure}[h]
\includegraphics[scale=0.86]{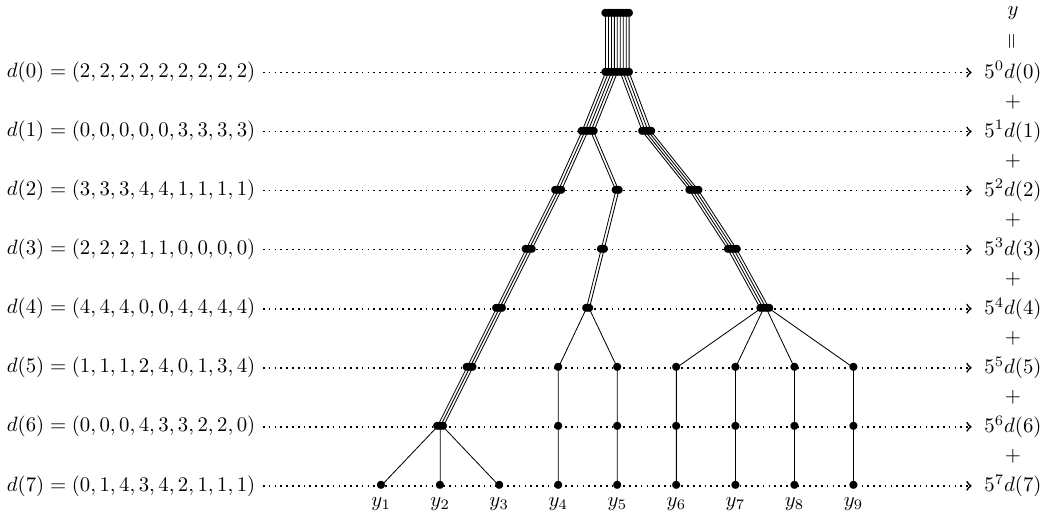}
\caption{The diagram for a tree $y\in\Z_5^9$ of length $8$.}\label{fig_1}
\end{figure}

\subsection{Integration with level pairs}\label{3_3}
Maintaining the notation above, we now establish the key connection between splitting chains and elements of $R^N\setminus V_0$.

\begin{definition}\label{spl_pair_def}
If $\spl\in\mathcal{S}_N$ and $\bm{n}=(\eta_0,\eta_1,\dots,\eta_{L(\spl)-1})\in\N^{L(\spl)}$, we call the pair $(\spl,\bm{n})$ a \emph{level pair}.
\end{definition}

Given $x\in R^N\setminus V_0$, we may associate a unique level pair to $x$ as follows. Let $y$ be the tree part of $x$ and suppose it has length $m$. Then $m=\max_{i<j}\{v(y_i-y_j)\}+1$, so there is a unique $L\in\N$ and unique integers $m_0,m_1,\dots,m_{L+1}$ satisfying $-1=:m_0<m_1<\dots<m_{L+1}:=m_L+1=m$ and $$\{v(y_i-y_j):1\leq i<j\leq N\}=\{m_1,m_2,m_3,\dots,m_L\}~.$$ Then for each $\ell\in\{0,1,2,\dots,L\}$ we define an equivalence relation $\sim_\ell$ on $[N]$ via $$i\sim_\ell j\qquad\iff\qquad y_i\equiv y_j\mod\pi^{m_{\ell+1}}$$ and let $\ptn_\ell$ be the partition of $[N]$ comprised of $\sim_\ell$-equivalence classes. Since $\min_{i<j}\{v(y_i-y_j)\}=m_1$, \Cref{val_m} implies $y_i\equiv y_j\mod\pi^{m_1}$ for all $i<j$ and hence $\ptn_0=\{[N]\}=\overline{\ptn}$. On the other hand, since $\max_{i<j}\{v(y_i-y_j)\}=m_L<m_{L+1}$, the same remark implies $y_i\not\equiv y_j\mod\pi^{m_{L+1}}$ for all $i<j$ and hence $\ptn_L=\{\{1\},\{2\},\dots,\{N\}\}=\underline{\ptn}$. For each $\ell\in\{0,1,\dots,L-1\}$ note that every pair $i<j$ satisfying $i\sim_{\ell+1}j$ also satisfies $i\sim_\ell j$, and hence $\ptn_{\ell+1}\leq\ptn_\ell$. In particular, since $v(y_i-y_j)=m_{\ell+1}$ for at least one pair $i<j$, then this pair satisfies $i\sim_\ell j$ and $i\not\sim_{\ell+1}j$, so in fact we have $\ptn_{\ell+1}<\ptn_\ell$. Then $\overline{\ptn}=\ptn_0>\ptn_1>\ptn_2>\dots>\ptn_L=\underline{\ptn}$, meaning $\spl=(\ptn_0,\ptn_1,\ptn_2,\dots,\ptn_L)$ is a splitting chain of order $N$ and length $L(\spl)=L$. Finally, define $\bm{n}=(\eta_0,\eta_1,\dots,\eta_{L-1})\in\N^L$ via $\eta_\ell:=m_{\ell+1}-m_\ell$. Thus $(\spl,\bm{n})$ is a level pair determined completely by $x$, so we call it the \emph{level pair associated to} $x$.

\begin{figure}[h]
\includegraphics[scale=0.8]{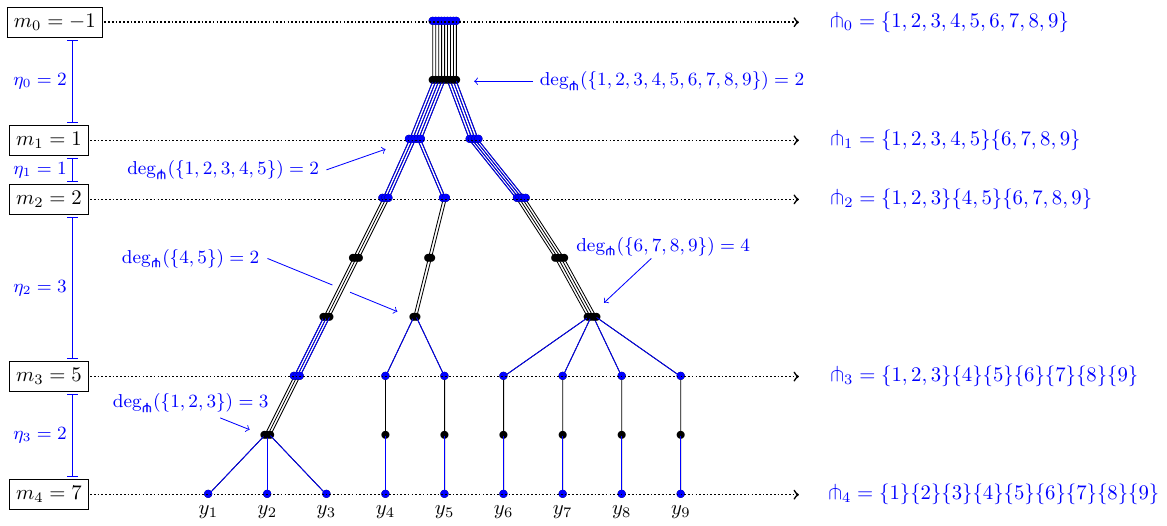}

\caption{The level pair $(\spl,\bm{n})$ associated to the tree in \Cref{Extree_1} is comprised of the splitting chain $\spl=(\ptn_0,\ptn_1,\ptn_2,\ptn_3,\ptn_4)\in\mathcal{S}_9$ described at right and the tuple $\bm{n}=(2,1,3,2)$. As mentioned above, the integers $m_0,m_1,m_2,m_3,m_4$ satisfy $m_0=-1$ and $m_{\ell+1}=-1+\eta_0+\dots+\eta_\ell$ for $0\leq\ell\leq3$.}\label{fig_2}
\end{figure}

The level pair associated to $x$ should be regarded as a compact summary of key features of the diagram for the tree part of $x$. More precisely, for each $\ell\in\{0,1,\dots,L(\spl)-1\}$ we have $y_i-y_j\in\pi^{m_{\ell+1}}R$ (where $m_{\ell+1}=-1+\eta_0+\eta_1+\dots+\eta_\ell$) if and only if $i$ and $j$ are contained in the same $\lambda\in\ptn_\ell$. The proper refinement $\ptn_\ell>\ptn_{\ell+1}$ reflects the fact that at least one $\lambda\in\ptn_\ell$ breaks into $\deg_\spl(\lambda)>1$ parts in $\ptn_{\ell+1}$, because at least one pair $i,j\in\lambda$ satisfies $y_i\not\equiv y_j\mod\pi^{m_{\ell+1}+1}$, and hence the paths for $y_i$ and $y_j$ in the diagram split at level $m_{\ell+1}$ (see \Cref{fig_2}). The integers $m_1,m_2,\dots,m_{L(\spl)}$ mark the levels where these splittings happen, and the integers $\eta_0,\eta_1,\dots,\eta_{L(\spl)-1}$ appearing in the tuple $\bm{n}$ are the spacings between those $m_\ell$.

\begin{definition}\label{pairset_def}
For each level pair $(\spl,\bm{n})$ define $$\mathcal{T}(\spl,\bm{n}):=\{x\in R^N\setminus V_0:(\spl,\bm{n})\text{ is the level pair associated to $x$}\}~.$$
\end{definition}

There are three key properties of the sets $\mathcal{T}(\spl,\bm{n})$ that will be used in our proof. The first is the following decomposition of $R^N$, which is immediate from \Cref{pairset_def} because each $x\in R^N\setminus V_0$ has exactly one associated level pair $(\spl,\bm{n})$: 
\begin{equation}\label{decomp}
R^N=V_0\sqcup\bigsqcup_{\spl\in\mathcal{S}_N}\bigsqcup_{\bm{n}\in\N^{L(\spl)}}\mathcal{T}(\spl,\bm{n})~.
\end{equation}
In particular, note that the union is countable because $\mathcal{S}_N$ is finite and $\N^{L(\spl)}$ is countable for each $\spl\in\mathcal{S}_N$, and note that some $\mathcal{T}(\spl,\bm{n})$ may be empty. The second key property of $\mathcal{T}(\spl,\bm{n})$ is the following lemma:

\begin{lemma}\label{pairset_measure}
Each $\mathcal{T}(\spl,\bm{n})$ is compact and open with measure $$\mu^N(\mathcal{T}(\spl,\bm{n}))=M_\spl(q)\cdot\prod_{\ell=0}^{L(\spl)-1}q^{-\rank(\ptn_\ell)\eta_\ell}~.$$ In particular, $\mathcal{T}(\spl,\bm{n})=\varnothing$ if and only if $M_\spl(q)=0$.
\end{lemma}

\begin{proof} Fix a level pair $(\spl,\bm{n})$. Using the tuple $\bm{n}=(\eta_0,\eta_1,\dots,\eta_{L(\spl)-1})\in\N^{L(\spl)}$, we define the familiar integers $m_0,m_1,\dots,m_{L(\spl)+1}$ by $m_0:=-1$, $$m_{\ell'+1}:=-1+\sum_{\ell=0}^{\ell'}\eta_\ell\qquad\text{for $0\leq\ell'\leq L(\spl)-1$}\qquad\text{and}\qquad m_{L(\spl)+1}:=m_{L(\spl)}+1=\sum_{\ell=0}^{L(\spl)-1}\eta_\ell~,$$ and note that $\eta_\ell=m_{\ell+1}-m_\ell$ for all $\ell\in\{0,1,\dots,L(\spl)-1\}$. By the discussion following \Cref{spl_pair_def}, note that $x\in\mathcal{T}(\spl,\bm{n})$ if and only if $x\in y+\pi^{m_{L(\spl)+1}}R^N$, where $y$ is a tree with the following properties:
\begin{itemize}
\item[(i)] $y$ is a finite sum of the form $y=\sum_{n=0}^{m_{L(\spl)}}\pi^nd(n)$,
\item[(ii)] $\{v(y_i-y_j):1\leq i<j\leq N\}=\{m_1,m_2,\dots,m_{L(\spl)}\}$, and
\item[(iii)] for $\lambda\in\ptn_\ell$, $i,j\in\lambda$ if and only if $y_i\equiv y_j\mod\pi^{m_{\ell+1}}$.
\end{itemize}
Since $y+\pi^{m_{L(\spl)+1}}R^N$ is open and compact with measure $$\mu^N(y+\pi^{m_{L(\spl)+1}}R^N)=\mu^N(\pi^{m_{L(\spl)+1}}R^N)=q^{-Nm_{L(\spl)+1}}=\prod_{\ell=0}^{L(\spl)-1}q^{-N\eta_\ell}~,$$ it remains to find the number of trees $y$ satisfying (i)-(iii) and multiply the measure above by this number. According to (i), every such $y$ corresponds to a unique finite sequence of digit tuples $d(0),d(1),\dots,d(m_{L(\spl)})\in D^N$, so we will count all valid $y$ by counting sequences. The terms in such a sequence may be chosen independently, so we will start by counting valid $d(n)\in D^N$ for each $n\in\{0,1,\dots,m_{L(\spl)}\}$ in two cases, maintaining conditions (i)-(iii) as we go:

\begin{itemize}
\item[(I)] Suppose $m_\ell<n<m_{\ell+1}$ for some $\ell\in\{0,1,\dots,L(\spl)-1\}$. For each $\lambda\in\ptn_\ell$ we must have $y_i\equiv y_j\mod\pi^{m_{\ell+1}}$ for all $i,j\in\lambda$. By \Cref{val_m}, we must therefore choose $d(n)\in D^N$ in such a way that for every $\lambda\in\ptn_\ell$, we have $\inf\{n:d_i(n)\neq d_j(n)\}=v(y_i-y_j)\geq m_{\ell+1}$ for all $i,j\in\lambda$. As $n<m_{\ell+1}$, this means we must ensure $d_i(n)=d_j(n)$ for all $i,j\in\lambda$. Thus, for each $\lambda\in\ptn_\ell$ we must choose one value $d_\lambda\in D$ and set $d_i(n)=d_\lambda$ for all $i\in\lambda$. This must be done for $\#\ptn_\ell$ parts $\lambda$ with $\#D=q$ choices per part, so we have $q^{\#\ptn_\ell}$ valid choices of $d(n)$.
\item[(II)] Suppose $n=m_{\ell+1}$ for some $\ell\in\{0,1\dots,L(\spl)-1\}$. Recall that $\ptn_{\ell+1}$ is a proper refinement of $\ptn_\ell$, note that $\ptn_\ell$ decomposes into the two disjoint sets $$\ptn_\ell':=\{\lambda\in\ptn_\ell:\lambda\in\ptn_{\ell+1}\}\qquad\text{and}\qquad\ptn_\ell'':=\{\lambda\in\ptn_\ell:\lambda\text{ is a union of at least two }\lambda'\in\ptn_{\ell+1}\}~,$$ and note that the latter is actually $\ptn_\ell''=\{\lambda\in\mathcal{B}(\spl):\ell_\spl(\lambda)=\ell\}$ by part (b) of \Cref{splchdef}. We use the decomposition $\ptn_\ell=\ptn_\ell'\sqcup\ptn_\ell''$ to break the problem of counting valid digit tuples $d(m_{\ell+1})\in D^N$ into two corresponding subcases:
\begin{itemize}
\item[$\bullet$] If $\lambda\in\ptn_\ell'$, then $\lambda\in\ptn_{\ell+1}$, and this means any $i,j\in\lambda$ must satisfy $y_i\equiv y_j\mod\pi^{m_{\ell+2}}$. Then by \Cref{val_m} we must have $\inf\{n:d_i(n)\neq d_j(n)\}=v(y_i-y_j)\geq m_{\ell+2}$, so we need only choose one value $d_\lambda\in D$ and set $d_i(m_{\ell+1})=d_\lambda$ for all $i\in\lambda$ just as in (I). Thus for each $\lambda\in\ptn_\ell'$ we have $q=\#D$ valid ways to choose the partial digit tuple $(d_i(m_{\ell+1}))_{i\in\lambda}$.
\item[$\bullet$] If $\lambda\in\ptn_\ell''$, then the number of parts $\lambda'\in\ptn_{\ell+1}$ comprising $\lambda$ is precisely $\deg_\spl(\lambda)$. Given one such $\lambda'\subset\lambda$, every pair $i,j\in\lambda'$ must satisfy $y_i\equiv y_j\mod\pi^{m_{\ell+2}}$, or equivalently $\inf\{n:d_i(n)\neq d_j(n)\}=v(y_i-y_j)\geq m_{\ell+2}$. Thus by \Cref{val_m} again, for every pair $i,j\in\lambda'$ we must have $d_i(m_{\ell+1})=d_j(m_{\ell+1})$. On the other hand, if $\lambda',\lambda''\in\ptn_{\ell+1}$ are distinct parts contained in $\lambda$ and we have $i\in\lambda'$ and $j\in\lambda''$, then both $y_i\equiv y_j\mod\pi^{m_{\ell+1}}$ and $y_i\not\equiv y_j\mod\pi^{m_{\ell+2}}$ must be satisfied. By \Cref{val_m} and the necessary condition $v(y_i-y_j)\in\{m_1,m_2,\dots,m_{L(\spl)}\}$, we must ensure $\inf\{n:d_i(n)\neq d_j(n)\}=v(y_i-y_j)=m_{\ell+1}$ and hence $d_i(m_{\ell+1})\neq d_j(m_{\ell+1})$. Therefore we must choose an ordered set of $\deg_\spl(\lambda)$ distinct values $d_{\lambda'}\in D$ (one for each part $\lambda'\in\ptn_{\ell+1}$ contained in $\lambda$, and ordered because these $\lambda'$ are distinct), and for each $\lambda'\subset\lambda$ we must set $d_i(m_{\ell+1})=d_{\lambda'}$ for all $i\in\lambda'$. Thus, for each $\lambda\in\ptn_\ell''$ the number of valid ways to choose the partial digit tuple $(d_i(m_{\ell+1}))_{i\in\lambda}$ is the number of ways of choosing these $d_{\lambda'}$, namely $$\binom{\#D}{\deg_\spl(\lambda)}\cdot(\deg_\spl(\lambda))!=(q)_{\deg_\spl(\lambda)}=q\cdot(q-1)_{\deg_\spl(\lambda)-1}~.$$
\end{itemize}
The two subcases now combine to conclude case (II) as follows. Since the partition $\ptn_\ell=\ptn_\ell'\sqcup\ptn_\ell''$ gives rise to a partition of the entries of the tuple $d(m_{\ell+1})=(d_1(m_{\ell+1}),\dots,d_N(m_{\ell+1}))$, the number of valid ways to choose this tuple is simply the product 
\begin{align*}
\prod_{\lambda\in\ptn_\ell}\#\{\text{valid ways to choose }(d_i(m_{\ell+1}))_{i\in\lambda}\}&=\prod_{\lambda\in\ptn_\ell'}q\cdot\prod_{\lambda\in\ptn_\ell''}(q\cdot(q-1)_{\deg_\spl(\lambda)-1})\\
&=q^{\#\ptn_\ell'}\cdot q^{\#\ptn_\ell''}\cdot\prod_{\lambda\in\ptn_\ell''}(q-1)_{\deg_\spl(\lambda)-1}\\
&=q^{\#\ptn_\ell}\cdot\prod_{\substack{\lambda\in\mathcal{B}(\ptn)\\\ell_\spl(\lambda)=\ell}}(q-1)_{\deg_\spl(\lambda)-1}~.
\end{align*}
\end{itemize}
Finally, we combine cases (I) and (II): For each $0\leq\ell\leq L(\spl)-1$, case (I) provides $q^{\#\ptn_\ell(m_{\ell+1}-m_\ell-1)}=q^{\#\ptn_\ell(\eta_\ell-1)}$ valid choices for the partial sequence of tuples $d(m_\ell+1),d(m_\ell+2),\dots,d(m_{\ell+1}-1)$, and the final product from case (II) is the number of valid ways to choose $d(m_{\ell+1})$ and hence extend the sequence to one of the form $d(m_\ell+1),d(m_\ell+2),\dots,d(m_{\ell+1})$. Thus, concatenating these sequences for $\ell\in\{0,1,\dots,L(\spl)-1\}$, we take the product over such $\ell$ and conclude that there are $$\prod_{\ell=0}^{L(\spl)-1}\left(q^{\#\ptn_\ell(\eta_\ell-1)}\cdot q^{\#\ptn_\ell}\cdot\prod_{\substack{\lambda\in\mathcal{B}(\ptn)\\\ell_\spl(\lambda)=\ell}}(q-1)_{\deg_\spl(\lambda)-1}\right)=M_\spl(q)\cdot\prod_{\ell=0}^{L(\spl)-1}q^{\#\ptn_\ell \eta_\ell}$$ ways to choose a sequence of digit tuples $d(0),d(1),\dots,d(m_{L(\spl)})$ such that $y=\sum_{n=0}^{m_{L(\spl)}}\pi^nd(n)$ satisfies (i)-(iii). Thus $\mathcal{T}(\spl,\bm{n})$ is a disjoint union of $M_\spl(q)\cdot\prod_{\ell=0}^{L(\spl)-1}q^{\#\ptn_\ell\eta_\ell}$ sets of the form $y+\pi^{m_{L(\spl)+1}}R^N$, so clearly $\mathcal{T}(\spl,\bm{n})=\varnothing$ if and only if $M_\spl(q)=0$, and $\mathcal{T}(\spl,\bm{n})$ is open and compact with measure $$\mu^N(\mathcal{T}(\spl,\bm{n}))=M_\spl(q)\cdot\prod_{\ell=0}^{L(\spl)-1}q^{\#\ptn_\ell \eta_\ell}\cdot\prod_{\ell=0}^{L(\spl)-1}q^{-N\eta_\ell}=M_\spl(q)\cdot\prod_{\ell=0}^{L(\spl)-1}q^{-\rank(\ptn_\ell)\eta_\ell}~.$$
\end{proof}

The final key property of the sets $\mathcal{T}(\spl,\bm{n})$ is that all factors of the integrand in \Cref{ZNKabs} except possibly $\rho$ are constant on each one. More precisely, we have the following lemma:

\begin{lemma}\label{pairset_value}
If $a,b\in\C$, $\bm{s}\in\C^{\binom{N}{2}}$ and $x\in\mathcal{T}(\spl,\bm{n})$, then $$\big(\max_{i<j}|x_i-x_j|\big)^a\big(\min_{i<j}|x_i-x_j|\big)^b\prod_{i<j}|x_i-x_j|^{s_{ij}}=q^{-(a+b+\sum_{i<j}s_{ij})(\eta_0-1)}\cdot\prod_{\ell=1}^{L(\spl)-1}q^{-(b+E_{\ptn_\ell}(\bm{s})-\rank(\ptn_\ell))\eta_\ell}~.$$
\end{lemma}

\begin{proof}
Just as in the proof of \Cref{pairset_measure}, we use the given tuple $\bm{n}=(\eta_0,\eta_1,\dots,\eta_{L(\spl)-1})$ to define integers $m_0,m_1,\dots,m_{L(\spl)+1}$ via $m_0:=-1$, $$m_{\ell'+1}:=-1+\sum_{\ell=0}^{\ell'}\eta_\ell\qquad\text{for }\ell'\in\{0,1,\dots,L(\spl)-1\}$$ and $m_{L(\spl)+1}:=m_{L(\spl)}+1$, and note that $\eta_\ell=m_{\ell+1}-m_\ell$ for all $\ell\in\{0,1,\dots,L(\spl)-1\}$. Now if $y$ is the tree part of $x$, we have $m_{L(\spl)}=\max_{i<j}\{v(y_i-y_j)\}$ and $x=y+z$ with $z\in\pi^{m_{L(\spl)+1}}R^N$, so $\min_{i<j}\{v(z_i-z_j)\}>m_{L(\spl)}$ and hence $v(y_i-y_j)=v(x_i-x_j)$ for all $i<j$ by part (a) \Cref{val_unif}. Therefore $$\big(\max_{i<j}|x_i-x_j|\big)^a\big(\min_{i<j}|x_i-x_j|\big)^b\prod_{i<j}|x_i-x_j|^{s_{ij}}=\big(\max_{i<j}|y_i-y_j|\big)^a\big(\min_{i<j}|y_i-y_j|\big)^b\prod_{i<j}|y_i-y_j|^{s_{ij}}~,$$ where
\begin{itemize}
\item[(i)] $y$ is a finite sum of the form $y=\sum_{n=0}^{m_{L(\spl)}}\pi^nd(n)$,
\item[(ii)] $\{v(y_i-y_j):1\leq i<j\leq N\}=\{m_1,m_2,\dots,m_{L(\spl)}\}$, and
\item[(iii)] for $\lambda\in\ptn_\ell$, $i,j\in\lambda$ if and only if $y_i\equiv y_j\mod\pi^{m_{\ell+1}}$
\end{itemize}
as in the proof of \Cref{pairset_measure}. Now
\begin{align*}
\left(\max_{i<j}|y_i-y_j|\right)^a&=q^{-a\cdot\min_{i<j}v(y_i-y_j)}=q^{-am_1}=q^{-a(\eta_0-1)}~,\\
\left(\min_{i<j}|y_i-y_j|\right)^b&=q^{-b\cdot\max_{i<j}v(y_i-y_j)}=q^{-bm_{L(\spl)}}=q^{-b(\eta_0-1)}\cdot\prod_{\ell=1}^{L(\spl)-1}q^{-b\eta_\ell}~,
\end{align*}
and
\begin{align*}
\sum_{i<j}s_{ij}v(y_i-y_j)&=\sum_{\ell=1}^{L(\spl)}\sum_{\substack{i<j\\v(y_i-y_j)=m_\ell}}s_{ij}m_\ell\\
&=\sum_{\ell=1}^{L(\spl)}\sum_{\substack{i<j\\v(y_i-y_j)=m_\ell}}s_{ij}(-1+\eta_0+\eta_1+\dots+\eta_{\ell-1})\\
&=\sum_{\substack{i<j\\v(y_i-y_j)=m_1}}s_{ij}(-1+\eta_0)\\
&~~~+\sum_{\substack{i<j\\v(y_i-y_j)=m_2}}s_{ij}(-1+\eta_0+\eta_1)\\
&~~~~~~~~\vdots\\
&~~~~~~~~~+\sum_{\substack{i<j\\v(y_i-y_j)=m_{L(\spl)}}}s_{ij}(-1+\eta_0+\eta_1+\dots+\eta_{L(\spl)-1})~,
\end{align*}
so exchanging the order of summation in the above sum of sums gives $$\sum_{i<j}s_{ij}v(y_i-y_j)=\left[\sum_{\substack{i<j\\v(y_i-y_j)\geq m_1}}s_{ij}\right](\eta_0-1)+\sum_{\ell=1}^{L(\spl)-1}\left[\sum_{\substack{i<j\\v(y_i-y_j)\geq m_{\ell+1}}}s_{ij}\right]\eta_\ell~.$$ Since $v(y_i-y_j)\geq m_1$ for all $i<j$, the first term in brackets is simply $\sum_{i<j}s_{ij}$. For the other terms in brackets, recall $$v(y_i-y_j)\geq m_{\ell+1}\qquad\iff\qquad y_i\equiv y_j\mod\pi^{m_{\ell+1}}\qquad\iff\qquad i,j\in\lambda\text{ for some $\lambda\in\ptn_\ell$}$$ by \Cref{val_m} and property (iii) of $y$. Therefore $$\sum_{\substack{i<j\\v(y_i-y_j)\geq m_{\ell+1}}}s_{ij}=\sum_{\lambda\in\ptn_\ell}\sum_{\substack{i<j\\i,j\in\lambda}}s_{ij}=E_{\ptn_\ell}(\bm{s})-\rank(\ptn_\ell)$$ by part (c) of \Cref{splchdef}, and hence $$\sum_{i<j}s_{ij}v(y_i-y_j)=\left[\sum_{i<j}s_{ij}\right](\eta_0-1)+\sum_{\ell=1}^{L(\spl)-1}\left[E_{\ptn_\ell}(\bm{s})-\rank(\ptn_\ell)\right]\eta_\ell$$ implies $$\prod_{i<j}|y_i-y_j|^{s_{ij}}=q^{-(\sum_{i<j}s_{ij})(\eta_0-1)}\cdot\prod_{\ell=1}^{L(\spl)-1}q^{-(E_{\ptn_\ell}(\bm{s})-\rank(\ptn_\ell))\eta_\ell}~.$$ Combining this with the max and min factors gives the desired result:
\begin{align*}
\big(\max_{i<j}|x_i-x_j|\big)^a\big(\min_{i<j}|x_i-x_j|\big)^b&\prod_{i<j}|x_i-x_j|^{s_{ij}}=\big(\max_{i<j}|y_i-y_j|\big)^a\big(\min_{i<j}|y_i-y_j|\big)^b\prod_{i<j}|y_i-y_j|^{s_{ij}}\\
&=q^{-(a+b+\sum_{i<j}s_{ij})(\eta_0-1)}\cdot\prod_{\ell=1}^{L(\spl)-1}q^{-(b+E_{\ptn_\ell}(\bm{s})-\rank(\ptn_\ell))\eta_\ell}~.
\end{align*}
\end{proof}

Though \Cref{pairset_measure,pairset_value} are useful on their own, their combination is especially important. Indeed, \Cref{pairset_measure} provides an explicit formula for the measure of $\mathcal{T}(\spl,\bm{n})$, on which the constant value taken by $x\mapsto\big(\max_{i<j}|x_i-x_j|\big)^a\big(\min_{i<j}|x_i-x_j|\big)^b\prod_{i<j}|x_i-x_j|^{s_{ij}}$ is given in \Cref{pairset_value}. Thus the integral of this function over a given set $\mathcal{T}(\spl,\bm{n})$ is simply the product of the function value and the value of $\mu^N(\mathcal{T}(\spl,\bm{n}))$:

\begin{corollary}\label{pairset_integral}
If $a,b\in\C$, then for every $\bm{s}\in\C^{\binom{N}{2}}$ we have
\begin{align*}\int_{\mathcal{T}(\spl,\bm{n})}\big(\max_{i<j}|x_i-x_j|\big)^a&\big(\min_{i<j}|x_i-x_j|\big)^b\prod_{i<j}|x_i-x_j|^{s_{ij}}\,|dx|\\
&=q^{-(N-1+a+b+\sum_{i<j}s_{ij})(\eta_0-1)}\cdot\frac{M_\spl(q)}{q^{N-1}}\cdot\prod_{\ell=1}^{L(\spl)-1}q^{-(b+E_{\ptn_\ell}(\bm{s}))\eta_\ell}~.
\end{align*}
Note that this quantity is entire in each of the variables $a,b$, and $s_{ij}$, and all mixed partial derivatives in those variables commute with each other and the integral sign.
\end{corollary}

\begin{remark}\label{level_reason}
Note that \Cref{pairset_integral} actually generalizes \Cref{pairset_measure}, as it can be recovered by setting $s_{ij}=a=b=0$ in integral formula above. Moreover, the exponential factors in the formula are completely determined by the level pair $(\spl,\bm{n})$, which encodes the common features of the tree diagrams for $x\in\mathcal{T}(\spl,\bm{n})$ (recall \Cref{fig_2}). In particular, we may regard $\ptn_0=\{[N]\}$ and $\eta_0$ as ``root data" that determine the factor $$q^{-(a+b+E_{\ptn_0}(\bm{s}))(\eta_0-1)}=q^{-(N-1+a+b+\sum_{i<j}s_{ij})(\eta_0-1)}~,$$ and note that
\begin{equation}\label{root_criterion}
|q^{-(N-1+a+b+\sum_{i<j}s_{ij})}|_{\C}<1\qquad\iff\qquad\bm{s}\in\mathcal{RP}_N(a,b)~.
\end{equation}
This is precisely the reason we named $\mathcal{RP}_N(a,b)$  the ``root polytope". Similarly, for each $\ell\in\{1,2,\dots,L(\spl)-1\}$, recall that $\ptn_\ell$ describes how the $N$ paths representing $(x_1,x_2,\dots,x_N)=x\in\mathcal{T}(\spl,\bm{n})$ branch in a particular level in the tree diagram, and that $\eta_\ell$ measures the vertical distance between the tree diagram levels corresponding to $\ptn_\ell$ and $\ptn_{\ell+1}$. Thus we regard $\ptn_\ell$ and $\eta_\ell$ as the $\ell$th ``level data", which determine the exponential factor $q^{-(b+E_{\ptn_\ell}(\bm{s}))\eta_\ell}$. Accordingly, we named $\mathcal{LP}_\spl(b)$ the ``level polytope" in \Cref{branchlevel} because
\begin{equation}\label{level_criterion}
|q^{-(b+E_{\ptn_\ell}(\bm{s}))}|_{\C}<1\quad\text{for all}\quad\ell\in\{1,2,\dots,L(\spl)-1\}\qquad\iff\qquad\bm{s}\in\mathcal{LP}_\spl(b)~.
\end{equation}
In the following proposition, we will finally see how the exponential factors corresponding to the root and how level polytopes combine to form the root and level functions. It should be regarded as the main result of \Cref{3_3}.
\end{remark}

\begin{proposition}\label{spl_integral}
Suppose $a,b\in\C$ and define $R_\spl^N:=\bigsqcup_{\bm{n}\in\N^{L(\spl)}}\mathcal{T}(\spl,\bm{n})$ for each $\spl\in\mathcal{S}_N$. If $M_\spl(q)>0$, then the integral $$\int_{R_\spl^N}\big(\max_{i<j}|x_i-x_j|\big)^a\big(\min_{i<j}|x_i-x_j|\big)^b\prod_{i<j}|x_i-x_j|^{s_{ij}}\,|dx|$$ converges absolutely if and only if $\bm{s}\in\mathcal{RP}_N(a,b)\cap\mathcal{LP}_\spl(b)$, and for such $\bm{s}$ it converges to $$\frac{1}{1-q^{-(N-1+a+b+\sum_{i<j}s_{ij})}}\cdot J_\spl(b,\bm{s})~.$$ Otherwise $M_\spl(q)=0$, in which case $R_\spl^N=\varnothing$ and the integral is simply zero.
\end{proposition}

\begin{proof}
The $M_\spl(q)=0$ case is immediate from \Cref{pairset_measure}, so suppose $M_\spl(q)>0$ and $\bm{s}\in\C^{\binom{N}{2}}$. Then \Cref{pairset_integral} and Fubini's Theorem for sums of nonnegative terms imply
\begin{align*}
\int_{R_\spl^N}&\left|\big(\max_{i<j}|x_i-x_j|\big)^a\big(\min_{i<j}|x_i-x_j|\big)^b\prod_{i<j}|x_i-x_j|^{s_{ij}}\right|_{\C}|dx|\\
&=\sum_{\bm{n}\in\N^{L(\spl)}}\int_{\mathcal{T}(\spl,\bm{n})}\big(\max_{i<j}|x_i-x_j|\big)^{\re(a)}\big(\min_{i<j}|x_i-x_j|\big)^{\re(b)}\prod_{i<j}|x_i-x_j|^{\re(s_{ij})}\,|dx|\\
&=\sum_{\bm{n}\in\N^{L(\spl)}}q^{-\re(N-1+a+b+\sum_{i<j}s_{ij})(\eta_0-1)}\cdot\frac{M_\spl(q)}{q^{N-1}}\prod_{\ell=1}^{L(\spl)-1}q^{-\re(b+E_{\ptn_\ell}(\bm{s}))\eta_\ell}\\
&=\sum_{\eta_0=1}^\infty|q^{-(N-1+a+b+\sum_{i<j}s_{ij})}|_{\C}^{(\eta_0-1)}\cdot\frac{M_\spl(q)}{q^{N-1}}\cdot\prod_{\ell=1}^{L(\spl)-1}\sum_{\eta_\ell=1}^\infty|q^{-(b+E_{\ptn_\ell}(\bm{s}))}|_{\C}^{\eta_\ell}~.
\end{align*}
Therefore the integral on the first line converges if and only if all of the geometric series in the product on the last line converge. But this is the case if and only if $\bm{s}\in\mathcal{RP}_N(a,b)\cap\mathcal{LP}_\spl(b)$ by \eqref{root_criterion} and \eqref{level_criterion}, so we have established the first claim. Moreover, if $\bm{s}\in\mathcal{RP}_N(a,b)\cap\mathcal{LP}_\spl(b)$ then the function $$x\mapsto\bm{1}_{R_\spl^N}(x)\left|\big(\max_{i<j}|x_i-x_j|\big)^a\big(\min_{i<j}|x_i-x_j|\big)^b\prod_{i<j}|x_i-x_j|^{s_{ij}}\right|_{\C}$$ is in $L^1(K^N,\mu^N)$ and dominates every partial sum of the function $$x\mapsto\sum_{\bm{n}\in\N^{L(\spl)}}\bm{1}_{\mathcal{T}(\spl,\bm{n})}(x)\big(\max_{i<j}|x_i-x_j|\big)^a\big(\min_{i<j}|x_i-x_j|\big)^b\prod_{i<j}|x_i-x_j|^{s_{ij}}~,$$ so the Dominated Convergence Theorem, \Cref{pairset_integral}, and Fubini's Theorem for absolutely convergent sums together imply 
\begin{align*}
\int_{R_\spl^N}&\big(\max_{i<j}|x_i-x_j|\big)^a\big(\min_{i<j}|x_i-x_j|\big)^b\prod_{i<j}|x_i-x_j|^{s_{ij}}|dx|\\
&=\sum_{\bm{n}\in\N^{L(\spl)}}\int_{\mathcal{T}(\spl,\bm{n})}\big(\max_{i<j}|x_i-x_j|\big)^a\big(\min_{i<j}|x_i-x_j|\big)^b\prod_{i<j}|x_i-x_j|^{s_{ij}}\,|dx|\\
&=\sum_{\bm{n}\in\N^{L(\spl)}}q^{-(N-1+a+b+\sum_{i<j}s_{ij})(\eta_0-1)}\cdot\frac{M_\spl(q)}{q^{N-1}}\cdot\prod_{\ell=1}^{L(\spl)-1}q^{-(b+E_{\ptn_\ell}(\bm{s}))\eta_\ell}\\
&=\sum_{\eta_0=1}^\infty q^{-(N-1+a+b+\sum_{i<j}s_{ij})(\eta_0-1)}\cdot\frac{M_\spl(q)}{q^{N-1}}\cdot\prod_{\ell=1}^{L(\spl)-1}\sum_{\eta_\ell=1}^\infty q^{-(b+E_{\ptn_\ell}(\bm{s}))\eta_\ell}\\
&=\frac{1}{1-q^{-(N-1+a+b+\sum_{i<j}s_{ij})}}\cdot J_{\spl,q}(b,\bm{s})~.
\end{align*}
\end{proof}

\Cref{spl_integral} is the key ingredient in the following proposition, which is the foundation of parts (a) and (b) of \Cref{main}.

\begin{proposition}\label{pair_conclusion}
Suppose the residue field of $K$ has cardinality $q$ and suppose $a,b\in\C$. Then the integral $$\int_{R^N}\big(\max_{i<j}|x_i-x_j|\big)^a\big(\min_{i<j}|x_i-x_j|\big)^b\prod_{i<j}|x_i-x_j|^{s_{ij}}\,|dx|$$ converges absolutely if and only if $\bm{s}$ belongs to $\Omega_N(a,b)$, and for such $\bm{s}$ it converges to $$\frac{1}{1-q^{-(N-1+a+b+\sum_{i<j}s_{ij})}}\cdot\sum_{\spl\in\mathcal{S}_N}J_{\spl,q}(b,\bm{s})~.$$
\end{proposition}

\begin{proof}
First, note that the decomposition in \eqref{decomp} can be rewritten as
\begin{equation}\label{newdecomp}
R^N=V_0\sqcup\bigsqcup_{\spl\in\mathcal{S}_N}R_\spl^N~,
\end{equation}
and that for each integer $m\geq 1$ we have $$V_0=\bigcup_{1\leq i<j\leq N}\{x\in R^N:x_i=x_j\}\subset\bigcup_{1\leq i<j\leq N}\bigsqcup_{\substack{y\in R^N\!\!/\pi^m\!R^N\\y_i=y_j}}(y+\pi^mR^N)~.$$ For each pair $\{i,j\}$ satisfying $1\leq i<j\leq N$, we have $\#\{y\in R^N/\pi^mR^N:y_i=y_j\}=q^{(N-1)m}$ and $\mu^N(y+\pi^mR^N)=q^{-Nm}$ by the basic properties of $K$ laid out in \Cref{3_1} and the beginning of \Cref{3_2}. Thus $V_0$ is contained in a union of $\binom{N}{2}$ sets of $\mu^N$-measure $q^{(N-1)m}\cdot q^{-Nm}=q^{-m}$, and since $m\geq 1$ can be arbitrarily large, it follows that $\mu^N(V_0)=0$. This fact and \eqref{newdecomp} together imply
\begin{align*}
\int_{R^N}\big(\max_{i<j}|x_i-x_j|\big)^a&\big(\min_{i<j}|x_i-x_j|\big)^b\prod_{i<j}|x_i-x_j|^{s_{ij}}\,|dx|=\\
&=\sum_{\spl\in\mathcal{S}_N}\int_{R_\spl^N}\big(\max_{i<j}|x_i-x_j|\big)^a\big(\min_{i<j}|x_i-x_j|\big)^b\prod_{i<j}|x_i-x_j|^{s_{ij}}\,|dx|~.
\end{align*}
According to \Cref{spl_integral}, the integral over $R_\spl^N$ converges absolutely if and only if $M_\spl(q)=0$ (in which case $R_\spl^N=\varnothing$) or $\bm{s}\in\mathcal{RP}_N(a,b)\cap\mathcal{LP}_\spl(b)$. Therefore the integral over $R^N$ converges if and only if $\bm{s}$ is in the polytope $$\mathcal{RP}_N(a,b)\cap\bigcap_{\substack{\spl\in\mathcal{S}_N\\M_\spl(q)>0}}\mathcal{LP}_\spl(b)~.$$ Recalling the definition of $\Omega_N(a,b)$ in part (a) of \Cref{main}, it remains to show that the condition ``$M_\spl(q)>0$" in the intersection above is extraneous. If $N=2$, the only splitting chain in $\mathcal{S}_2$ is $\spl=([N],~\{1\}\{2\})$, which has $M_\spl(q)=q-1>0$ because $q\geq 2$. Thus if $N=2$ the condition ``$M_\spl(q)>0$" is automatic and the proof is complete. Now suppose $N>2$. It suffices to show that
\begin{equation}\label{crux}
\bigcap_{\substack{\spl\in\mathcal{S}_N\\M_\spl(q)>0}}\mathcal{LP}_\spl(b)\subset\bigcap_{\spl\in\mathcal{S}_N}\mathcal{LP}_\spl(b)
\end{equation}
because the reverse containment is obvious. To this end, let $\ptn^\circ$ be an arbitrary partition of $[N]$ other than $\overline{\ptn}=\{[N]\}$ or $\underline{\ptn}=\{\{1\},\{2\},\dots,\{N\}\}$. We will construct a splitting chain $\spl^\circ\in\mathcal{S}_N$ that has $\ptn^\circ$ as a level, satisfies $M_{\spl^\circ}(q)>0$ for any $q\geq 2$, and has length $L(\spl^\circ)\geq 2$ as follows. Put $k=\#\ptn^\circ-1$ and define $\ptn_k:=\ptn^\circ$. Then $k\geq 1$ and we may write $\ptn_k=\{\lambda_1,\lambda_2,\dots,\lambda_{k+1}\}$ where $\#\lambda_1\geq 2$ and $\#\lambda_1\geq\#\lambda_2\geq\dots\geq\#\lambda_{k+1}$. Now for each $\ell\in\{0,1,2,\dots,k-1\}$, define $$\ptn_\ell:=\{\lambda_1,\lambda_2,\dots,\lambda_\ell,(\lambda_{\ell+1}\cup\lambda_{\ell+2}\cup\dots\cup\lambda_{k+1})\}$$ and note that $\overline{\ptn}=\ptn_0>\ptn_1>\dots>\ptn_k$ where each refinement is given by splitting a single part into two parts. For $\ell\geq k+1$, recursively define $\ptn_\ell$ to be any refinement of $\ptn_{\ell-1}$ such that each non-singleton part $\lambda\in\ptn_{\ell-1}$ splits into $\lambda'=\lambda\setminus\{i\}\in\ptn_\ell$ and $\{i\}\in\ptn_\ell$ for some $i\in\lambda$. The largest part $\lambda_1\in\ptn_k$ will fully refine into singletons after $\#\lambda_1-1$ steps in the recursion, by which time all other parts will have also refined into singletons. Therefore the recursion must stop at $\ell=k+\#\lambda_1-1$ with $\ptn_k>\ptn_{k+1}>\dots>\ptn_{k+\#\lambda_1-1}=\underline{\ptn}$, where each refinement is given by refining non-singleton parts into exactly two parts. Thus we have constructed a splitting chain $\spl^\circ=(\ptn_0,\ptn_1,\dots,\ptn_{k+\#\lambda_1-1})\in\mathcal{S}_N$ that has the given partition $\ptn^\circ$ as its $k$th level, has length $$L(\spl^\circ)=k+\#\lambda_1-1\geq k+1=\#\ptn\geq 2~,$$ and has $\deg_{\spl^\circ}(\lambda)=2\leq q$ for all branches $\lambda\in\mathcal{B}(\spl^\circ)$. The last property implies $M_{\spl^\circ}(q)>0$, so we have
\begin{align*}
\bigcap_{\substack{\spl\in\mathcal{S}_N\\M_\spl(q)>0}}\mathcal{LP}_\spl(b)\subset\mathcal{LP}_{\spl^\circ}(b)&=\bigcap_{\ell=1}^{L(\spl^\circ)-1}\left\{\bm{s}\in\C^{\binom{N}{2}}:\re(b+E_{\ptn_\ell}(\bm{s}))>0\right\}\\
&\subset\left\{\bm{s}\in\C^{\binom{N}{2}}:\re(b+E_{\ptn^\circ}(\bm{s}))>0\right\}~.
\end{align*}
This argument works for every partition $\ptn^\circ$ with $\underline{\ptn}<\ptn^\circ<\overline{\ptn}$, so it follows that $$\bigcap_{\substack{\spl\in\mathcal{S}_N\\M_\spl(q)>0}}\mathcal{LP}_\spl(b)\subset\bigcap_{\substack{\text{partitions }\ptn\\\underline{\ptn}<\ptn<\overline{\ptn}}}\left\{\bm{s}\in\C^{\binom{N}{2}}:\re(b+E_\ptn(\bm{s}))>0\right\}~.$$ On the other hand, for every splitting chain $\spl\in\mathcal{S}_N$, each level $\ptn_\ell$ with $1\leq\ell\leq L(\spl)-1$ is a partition of $[N]$ satisfying $\underline{\ptn}<\ptn_\ell<\overline{\ptn}$, so $$\bigcup_{\spl\in\mathcal{S}_N}\{\ptn_1,\ptn_2,\dots,\ptn_{L(\spl)-1}\}\subset\{\text{partitions }\ptn:\underline{\ptn}<\ptn<\overline{\ptn}\}~.$$ This implies
\begin{align*}
\bigcap_{\substack{\spl\in\mathcal{S}_N\\M_\spl(q)>0}}\mathcal{LP}_\spl(b)&\subset\bigcap_{\substack{\text{partitions }\ptn\\\underline{\ptn}<\ptn<\overline{\ptn}}}\left\{\bm{s}\in\C^{\binom{N}{2}}:\re(b+E_\ptn(\bm{s}))>0\right\}\\
&\subset\bigcap_{\spl\in\mathcal{S}_N}\bigcap_{\ell=1}^{L(\spl)-1}\left\{\bm{s}\in\C^{\binom{N}{2}}:\re(b+E_{\ptn_\ell}(\bm{s}))>0\right\}=\bigcap_{\spl\in\mathcal{S}_N}\mathcal{LP}_\spl(b)~,
\end{align*}
so \eqref{crux} holds and the proof is complete.\\
\end{proof}

\subsection{Integration with branch pairs}\label{3_4}

Branch pairs are an analogue of level pairs that relate branch functions to level functions, and this relationship is the key idea behind part (c) of \Cref{main}. Before defining branch pairs, we will restate and prove parts (a) and (b) of \Cref{reduction}.

\begin{lemma}[\Cref{reduction}]
We say that a splitting chain $\spl$ is \emph{reduced} if for each $\lambda\in\mathcal{B}(\spl)$ there is a unique level $\ptn_\ell$ containing $\lambda$ (namely, the level $\ptn_{\ell_\spl(\lambda)}$). We write $\mathcal{R}_N:=\{\spl\in\mathcal{S}_N:\spl\text{ is reduced}\}$ and define an equivalence relation $\simeq$ on $\mathcal{S}_N$ by writing $\spl\simeq\spl'$ if and only if $\mathcal{B}(\spl)=\mathcal{B}(\spl')$.
\begin{enumerate}
\item[(a)] If $\spl\simeq\spl'$, then the branch degrees, part exponents, multiplicity polynomials, and branch polytopes for $\spl$ and $\spl'$ respectively coincide.
\item[(b)] For each $\spl\in\mathcal{S}_N$ there is a unique $\spl^*\in\mathcal{R}_N$ such that $\spl\simeq\spl^*$. We call this $\spl^*$ the \emph{reduction} of $\spl$ and regard $\mathcal{R}_N$ as a complete set of representatives for $\mathcal{S}_N$ modulo $\simeq$.
\item[(c)] For each $\spl^*\in\mathcal{R}_N$ we have $$\bigcap_{\substack{\spl\in\mathcal{S}_N\\\spl\simeq\spl^*}}\mathcal{LP}_\spl(0)=\mathcal{BP}_{\spl^*}~,$$ and therefore $$\quad\bigcap_{\spl\in\mathcal{S}_N}\mathcal{LP}_\spl(0)=\bigcap_{\spl^*\in\mathcal{R}_N}\mathcal{BP}_{\spl^*}~.$$
\end{enumerate}
\end{lemma}

\begin{proof}\
\begin{enumerate}
\item[(a)] Suppose $\spl,\spl'\in\mathcal{S}_N$ and $\spl\simeq\spl'$. Then $\mathcal{B}(\spl)=\mathcal{B}(\spl')$ and our only task is to prove that $\deg_\spl(\lambda)=\deg_{\spl'}(\lambda)$ for all $\lambda\in\mathcal{B}(\spl)$, for then the rest of (a) will follow immediately from part (c) of \Cref{splchdef} part (b) of \Cref{branchlevel}. To this end, suppose $\lambda\in\mathcal{B}(\spl)$ and recall $$\deg_\spl(\lambda)=\#\{\lambda'\in\ptn_{\ell_\spl(\lambda)+1}:\lambda'\subset\lambda\}\qquad\text{and}\qquad\deg_{\spl'}(\lambda)=\#\{\lambda'\in\ptn_{\ell_{\spl'}(\lambda)+1}':\lambda'\subset\lambda\}~.$$ Note that any branch $\lambda'\in\mathcal{B}(\spl)$ contained in both $\ptn_{\ell_\spl(\lambda)+1}$ and $\lambda$ must not appear in any of the levels $\ptn_0,\ptn_1,\dots,\ptn_{\ell_\spl(\lambda)}$ because $\ptn_{\ell_\spl(\lambda)+1}$ properly refines all of them and by definition, $\ell_\spl(\lambda)=\max\{\ell\in\{0,1,\dots,L(\spl)-1\}:\lambda\in\ptn_\ell\}$. Moreover, no branch $\lambda''\subsetneq\lambda'$ can appear in $\ptn_{\ell_\spl(\lambda)+1}$ because $\lambda'\in\ptn_{\ell_\spl(\lambda)+1}$. Therefore $\{\lambda'\in\ptn_{\ell_\spl(\lambda)+1}:\lambda'\subset\lambda\}$ is comprised of precisely the largest branches in $\mathcal{B}(\spl)$ that are properly contained in $\lambda$, along with any remaining singletons $\{i\}\subset\lambda$. Thus $\{\lambda'\in\ptn_{\ell_\spl(\lambda)+1}:\lambda'\subset\lambda\}$ is completely determined by $\mathcal{B}(\spl)$ and $\lambda$. But $\mathcal{B}(\spl)=\mathcal{B}(\spl')$, so $\{\lambda'\in\ptn_{\ell_\spl(\lambda)+1}:\lambda'\subset\lambda\}=\{\lambda'\in\ptn_{\ell_{\spl'}(\lambda)+1}':\lambda'\subset\lambda\}$ and we conclude that $\deg_\spl(\lambda)=\deg_{\spl'}(\lambda)$.

\item[(b)] Suppose $\spl\in\mathcal{S}_N$ and note that $\mathcal{B}(\spl)$ is partially ordered by $\subset$ with unique largest element $[N]$. We will construct $\spl^*\in\mathcal{R}_N$ satisfying $\mathcal{B}(\spl^*)=\mathcal{B}(\spl)$. Begin by letting $\ptn_0^*:=\{[N]\}$, and continue recursively for $\ell\geq 0$ as follows: Define a partition $\ptn_{\ell+1}^*$ of $[N]$ by taking the largest branches remaining in $\mathcal{B}(\spl^*)\setminus(\ptn_0^*\cup\ptn_1^*\cup\dots\cup\ptn_\ell^*)$ and any leftover singletons in $[N]$. At the first $\ell\geq 0$ for which $\mathcal{B}(\spl)\setminus(\ptn_0^*\cup\ptn_1^*\cup\dots\cup\ptn_\ell^*)=\varnothing$, end the recursion, let $L^*:=\ell+1$, and finally let $\ptn_{L^*}^*:=\underline{\ptn}$. Then by construction, we will have $\ptn_{\ell+1}^*<\ptn_\ell^*$ because each part of $\ptn_{\ell+1}^*$ is contained in a part of $\ptn_\ell^*$ and at least one part of $\ptn_{\ell+1}^*$ will be properly contained in one of those in $\ptn_\ell^*$. Thus $\spl^*=(\ptn_0^*,\ptn_1^*,\dots,\ptn_{L^*}^*)$ is a splitting chain of order $N$ and length $L^*\leq L(\spl)$ with $\mathcal{B}(\spl^*)=\left(\bigcup_{\ell=0}^{L^*-1}\ptn_\ell^*\right)\setminus\underline{\ptn}=\mathcal{B}(\spl)$. Moreover, $\spl^*$ is reduced because each $\lambda\in\mathcal{B}(\spl^*)$ is contained in exactly one $\ptn_\ell^*$, and $\spl^*$ is unique because it has been completely determined by $\mathcal{B}(\spl)$.

\item[(c)] Suppose $\spl^*\in\mathcal{R}_N$. The first claim is obvious from \Cref{branchlevel} if $\mathcal{B}(\spl^*)\setminus\overline{\ptn}=\varnothing$, so suppose otherwise and choose an arbitrary branch $\lambda^\circ\in\mathcal{B}(\spl^*)\setminus\overline{\ptn}$. We will construct a splitting chain $\spl^\circ\in\mathcal{S}_N$ such that $\spl^\circ\simeq\spl^*$ and such that $\spl^\circ$ has a level containing $\lambda^\circ$ and no other branches. The set $\mathcal{B}':=\{\lambda\in\mathcal{B}(\spl^*):\lambda\not\subset\lambda^\circ\}$ is partially ordered by $\subset$ with unique largest element $[N]$, so we may apply the same algorithm in the proof of part (b) to obtain the unique reduced splitting chain $\spl'=(\ptn_0',\ptn_1',\dots,\ptn_L')$ satisfying $\mathcal{B}(\spl')=\mathcal{B}'$. There is a smallest branch in $\mathcal{B}(\spl')$ that contains $\lambda^\circ$, say $\lambda'$, and there are no subsets of $\lambda^\circ$ in $\mathcal{B}(\spl')$. Thus if $\lambda^\circ=\{i_1,i_2,\dots,i_n\}$, the singletons $\{i_1\},\{i_2\},\dots,\{i_n\}$ must appear in $\ptn_\ell'$ for all $\ell>\ell_{\spl'}(\lambda')$. Now let $\ptn_0,\ptn_1,\dots,\ptn_{L'}$ be the partitions satisfying $\ptn_\ell=\ptn_\ell'$ for $0\leq\ell\leq\ell_{\spl'}(\lambda')$, and for $\ell>\ell_{\spl'}(\lambda')$ take $\ptn_\ell$ to be equal to $\ptn_\ell'$ but with $\{i_1\}\{i_2\}\dots\{i_n\}$ replaced by $\lambda^\circ=\{i_1,i_2,\dots,i_n\}$. This yields partitions $$\overline{\ptn}=\ptn_0>\ptn_1>\dots>\ptn_{L'}$$ with $\mathcal{B}(\spl^*)\setminus(\ptn_0\cup\ptn_1\cup\dots\cup\ptn_{L'})=\{\lambda\in\mathcal{B}(\spl^*):\lambda\subsetneq\lambda^\circ\}$ where $\lambda^\circ$ is the only non-singleton part in $\ptn_{L'}$. We continue recursively for $\ell\geq L'$, defining $\ptn_{\ell+1}$ to be the partition comprised of the largest branches remaining in $\mathcal{B}(\spl^*)\setminus(\ptn_0\cup\ptn_1\cup\dots\cup\ptn_\ell)$ and any leftover singletons in $[N]$. We end the recursion at the first $\ell\geq L'$ such that $\mathcal{B}(\spl^*)\setminus(\ptn_0\cup\ptn_1\cup\dots\cup\ptn_\ell)=\varnothing$ and set $L:=\ell+1$ and $\ptn_L:=\underline{\ptn}$. The result is a splitting chain $\spl^\circ=(\ptn_0,\ptn_1,\dots,\ptn_L)$ with $\mathcal{B}(\spl^\circ)=\mathcal{B}(\spl^*)$ (i.e., $\spl^\circ\simeq\spl^*$) and a level $\ptn_{L'}$ whose only non-singleton part is $\lambda^\circ$, and hence $E_{\ptn_{L'}}(\bm{s})=e_{\lambda^\circ}(\bm{s})$. Thus for $\lambda^\circ\in\mathcal{R}_N$ we have a splitting chain $\spl^\circ\simeq\spl^*$ satisfying
\begin{align*}
\mathcal{LP}_{\spl^\circ}(0)&=\bigcap_{\ell=1}^{L-1}\left\{\bm{s}\in\C^{\binom{N}{2}}:\re(E_{\ptn_\ell}(\bm{s}))>0\right\}\\
&\subset\left\{\bm{s}\in\C^{\binom{N}{2}}:\re(E_{\ptn_{L'}}(\bm{s}))>0\right\}=\left\{\bm{s}\in\C^{\binom{N}{2}}:\re(e_{\lambda^\circ}(\bm{s}))>0\right\}~,
\end{align*}
and hence $$\bigcap_{\substack{\spl\in\mathcal{S}_N\\\spl\simeq\spl^*}}\mathcal{LP}_\spl(0)\subset\left\{\bm{s}\in\C^{\binom{N}{2}}:\re(e_{\lambda^\circ}(\bm{s}))>0\right\}~.$$ Since this argument works for any $\lambda^\circ\in\mathcal{B}(\spl^*)\setminus\overline{\ptn}$, it follows that $$\bigcap_{\substack{\spl\in\mathcal{S}_N\\\spl\simeq\spl^*}}\mathcal{LP}_\spl(0)\subset\bigcap_{\lambda\in\mathcal{B}(\spl)\setminus\overline{\ptn}}\left\{\bm{s}\in\C^{\binom{N}{2}}:\re(e_\lambda(\bm{s}))>0\right\}=\mathcal{BP}_{\spl^*}~.$$ To show the reverse containment, suppose $\bm{s}\in\mathcal{BP}_{\spl^*}$, so that $\re(e_\lambda(\bm{s}))>0$ for all $\lambda\in\mathcal{B}(\spl^*)\setminus\overline{\ptn}$. For any splitting chain $\spl\simeq\spl^*$ and any level $\ptn_\ell$ with $1\leq\ell\leq L(\spl)-1$, the level exponent $E_{\ptn_\ell}(\bm{s})=\sum_{\lambda\in\mathcal{\ptn_\ell}}e_\lambda(\bm{s})$ is a sum over at least one $\lambda\in\mathcal{B}(\spl)\cap\ptn_\ell\subset\mathcal{B}(\spl^*)\setminus\overline{\ptn}$ and hence $\re(E_{\ptn_\ell}(\bm{s}))>0$. It follows that $\bm{s}\in\mathcal{LP}_\spl(0)$ for all $\spl\simeq\spl^*$, and we conclude that $$\bigcap_{\substack{\spl\in\mathcal{S}_N\\\spl\simeq\spl^*}}\mathcal{LP}_\spl(0)=\mathcal{BP}_{\spl^*}~.$$ Finally, since this holds for all $\spl^*\in\mathcal{R}_N$, part (b) implies $\bigcap_{\spl\in\mathcal{S}_N}\mathcal{LP}_\spl(0)=\bigcap_{\spl^*\in\mathcal{R}_N}\mathcal{BP}_\spl$.
\end{enumerate}
\end{proof}

It is worth noting here that the recursive algorithm in the proof of part (b) of \Cref{reduction} can be used to find the reduction of any splitting chain. We now apply this algorithm to the splitting chain $\spl\in\mathcal{S}_9$ from \Cref{fig_2}.

\begin{example}\label{reduction_alg}
Recall $\spl=(\ptn_0,\ptn_1,\ptn_2,\ptn_3,\ptn_4)\in\mathcal{S}_9$ from \Cref{fig_2}, where
\begin{align*}
\ptn_0&=\{1,2,3,4,5,6,7,8,9\}~,\\
\ptn_1&=\{1,2,3,4,5\}\{6,7,8,9\}~,\\
\ptn_2&=\{1,2,3\}\{4,5\}\{6,7,8,9\}~,\\
\ptn_3&=\{1,2,3\}\{4\}\{5\}\{6\}\{7\}\{8\}\{9\}~,\\
\ptn_4&=\{1\}\{2\}\{3\}\{4\}\{5\}\{6\}\{7\}\{8\}\{9\}~.
\end{align*}
Before starting the algorithm, note that its branch set is $$\mathcal{B}(\spl)=\big\{\{1,2,3,4,5,6,7,8,9\},~\{1,2,3,4,5\},~\{6,7,8,9\},~\{1,2,3\},~\{4,5\}\big\}~.$$ We initialize the algorithm by letting $\ptn_0^*:=\{1,2,3,4,5,6,7,8,9\}$, and the recursive part runs as follows:
\begin{itemize}
\item $\ell=0$ : The maximal branches remaining in $\mathcal{B}(\spl)\setminus\ptn_0^*=\big\{\{1,2,3,4,5\},~\{6,7,8,9\}, \{1,2,3\}, \{4,5\}\big\}$ (partially ordered via $\subset$) are the incomparable sets $\{1,2,3,4,5\}$ and $\{6,7,8,9\}$, so we define the partition $$\ptn_1^*:=\{1,2,3,4,5\}\{6,7,8,9\}~.$$
\item $\ell=1$ : The maximal branches remaining in $\mathcal{B}(\spl)\setminus(\ptn_0^*\cup\ptn_1^*)=\{\{1,2,3\},~\{4,5\}\}$ are the incomparable sets $\{1,2,3\}$ and $\{4,5\}$, so by including leftover singletons $\{i\}\subset[9]$ we define the partition $$\ptn_2^*:=\{1,2,3\}\{4,5\}\{6\}\{7\}\{8\}\{9\}~.$$
\item $\ell=2$ : We now have $\mathcal{B}(\spl)\setminus(\ptn_0^*\cup\ptn_1^*\cup\ptn_2^*)=\varnothing$, so let $L^*:=\ell+1=3$ and end the recursion.
\end{itemize}
Finally, let $$\ptn_3^*=\ptn_{L^*}^*:=\underline{\ptn}=\{1\}\{2\}\{3\}\{4\}\{5\}\{6\}\{7\}\{8\}\{9\}~,$$ and note that the algorithm is done. It is straightforward to verify that the resulting tuple $\spl^*:=(\ptn_0^*,\ptn_1^*,\ptn_2^*,\ptn_3^*)$, where
\begin{align*}
\ptn_0^*&=\{1,2,3,4,5,6,7,8,9\}~,\\
\ptn_1^*&=\{1,2,3,4,5\}\{6,7,8,9\}~,\\
\ptn_2^*&=\{1,2,3\}\{4,5\}\{6\}\{7\}\{8\}\{9\}~,\\
\ptn_3^*&=\{1\}\{2\}\{3\}\{4\}\{5\}\{6\}\{7\}\{8\}\{9\}~,
\end{align*}
is a reduced splitting chain of order 9, with $\spl\simeq\spl^*$ and $L(\spl^*)\leq L(\spl)$.\\
\end{example}

We may now introduce branch pairs and establish their relationship with level pairs.

\begin{definition}\label{bpair_def}
If $\spl^*\in\mathcal{R}_N$ and $\bm{k}=(k_\lambda)$ is a tuple of positive integers indexed by $\lambda\in\mathcal{B}(\spl^*)$, we call $[\spl^*,\bm{k}]$ a \emph{branch pair}.
\end{definition}

\begin{theorem}\label{bijection_thm}
Suppose $\spl^*\in\mathcal{R}_N$. There is a bijection $$\left\{[\spl^*,\bm{k}]:\bm{k}=(k_\lambda)\in\N^{\mathcal{B}(\spl^*)}\right\}\longleftrightarrow\bigsqcup_{\substack{\spl\in\mathcal{S}_N\\\spl\simeq\spl^*}}\left\{(\spl,\bm{n}):\bm{n}=(\eta_0,\eta_1,\dots,\eta_{L(\spl)-1})\in\N^{L(\spl)}\right\}$$ such that if $[\spl^*,\bm{k}]$ and $(\spl,\bm{n})$ correspond, we have $k_{[N]}=\eta_0$ and for each $\lambda\in\mathcal{B}(\spl)\setminus\overline{\ptn}$ we have
\begin{equation}\label{n_to_k_eqn}
k_\lambda=\sum_{\ell=\ell_\spl(\lambda^*)+1}^{\ell_\spl(\lambda)}\eta_\ell
\end{equation}
where $\lambda^*\in\mathcal{B}(\spl)$ is the smallest branch properly containing $\lambda$.
\end{theorem}

\begin{proof}
Fix $\spl^*\in\mathcal{R}_N$ and let $\bm{k}=(k_\lambda)$ be an arbitrary tuple of positive integers indexed by $\lambda\in\mathcal{B}(\spl^*)$. We associate a unique level pair to $[\spl^*,\bm{k}]$ as follows. The set $$\mathcal{M}:=\left\{-1+\sum_{\substack{\lambda'\in\mathcal{B}(\spl^*)\\\lambda'\supset\lambda}}k_{\lambda'}:\lambda\in\mathcal{B}(\spl^*)\right\}$$ is comprised of finitely many, say $L$, nonnegative integers. Put $m_0:=-1$ and let $\{m_1,m_2,\dots,m_L\}$ be the enumeration of $\mathcal{M}$ satisfying $m_0<m_1<m_2<\dots<m_L$. For each $\lambda\in\mathcal{B}(\spl^*)$ define $$\ell_{[\spl^*,\bm{k}]}(\lambda):=\text{the unique $\ell\in\{0,1,\dots,L-1\}$ such that $\sum_{\substack{\lambda'\in\mathcal{B}(\spl^*)\\\lambda'\supset\lambda}}k_{\lambda'}=m_{\ell+1}+1~.$}$$ Then by the definition of $\mathcal{M}=\{m_1,m_2,\dots,m_L\}$, for each $\ell\in\{0,1,\dots,L-1\}$ there is at least one $\lambda\in\mathcal{B}(\spl^*)$ satisfying $\ell_{[\spl^*,\bm{k}]}(\lambda)=\ell$, and $\lambda=[N]$ is the unique branch satisfying $\ell_{[\spl^*,\bm{k}]}(\lambda)=0$. Moreover, we have $\ell_{[\spl^*,\bm{k}]}(\lambda')<\ell_{[\spl^*,\bm{k}]}(\lambda)$ whenever $\lambda,\lambda'\in\mathcal{B}(\spl^*)$ satisfy $\lambda\subsetneq\lambda'$. We now construct $L$ partitions $\ptn_0,\ptn_1,\dots,\ptn_{L-1}$ of $[N]$ as follows. Let $\ptn_0:=\{[N]\}$, and for each $\ell\in\{1,\dots,L-1\}$ let $\mathcal{B}_\ell(\spl^*)$ be the subset of $\mathcal{B}(\spl^*)$ defined by $$\lambda\in\mathcal{B}_\ell(\spl^*)\qquad\iff\qquad\parbox{7.2cm}{$\ell_{[\spl^*,\bm{k}]}(\lambda)\geq\ell$ and $\ell_{[\spl^*,\bm{k}]}(\lambda^*)<\ell$, where $\lambda^*$ is the smallest branch in $\mathcal{B}(\spl^*)$ satisfying $\lambda\subsetneq\lambda^*$,}$$ let $\ptn_\ell$ be the partition of $[N]$ comprised of all $\lambda\in\mathcal{B}_\ell(\spl^*)$ and all $\{i\}\subset[N]\setminus\bigcup_{\lambda\in\mathcal{B}_\ell(\spl^*)}\lambda$, and finally let $\ptn_L:=\underline{\ptn}$. Now if $\ell\in\{1,2,\dots,L\}$ and $\lambda\in\ptn_\ell$, then either $\lambda$ is a singleton or $\lambda\in\mathcal{B}_\ell(\spl^*)$. In the latter case we have $\ell_{[\spl^*,\bm{k}]}(\lambda^*)<\ell\leq\ell_{[\spl^*,\bm{k}]}(\lambda)$ where $\lambda^*$ is the smallest branch in $\mathcal{B}(\spl^*)$ satisfying $\lambda\subsetneq\lambda^*$. If $\ell_{[\spl^*,\bm{k}]}(\lambda^*)=\ell-1$, then $\lambda^*\in\ptn_{\ell-1}$. Otherwise $\ell_{[\spl^*,\bm{k}]}(\lambda^*)<\ell-1$, in which case $\lambda\in\ptn_{\ell-1}$, so in any case each $\lambda\in\ptn_\ell$ is contained in some part of $\ptn_{\ell-1}$ and hence $\ptn_\ell\leq\ptn_{\ell-1}$. Moreover, there is at least one part $\lambda'\in\ptn_{\ell-1}$ with $\ell_{[\spl^*,\bm{k}]}(\lambda')=\ell-1$, so $\lambda'\notin\mathcal{B}_\ell(\spl^*)$ implies $\lambda'\notin\ptn_\ell$ and hence $\ptn_\ell<\ptn_{\ell-1}$. Now $\spl:=(\ptn_0,\ptn_1,\dots,\ptn_L)$ is a tuple of partitions of $[N]$ satisfying $\ptn_0>\ptn_1>\dots>\ptn_L=\underline{\ptn}$, so $\spl$ is a splitting chain of order $N$ and length $L(\spl)=L$. It is clear from the construction of $\spl$ that $\mathcal{B}(\spl)=\bigcup_{\ell=0}^{L-1}\mathcal{B}_\ell(\spl^*)=\mathcal{B}(\spl^*)$, and that each branch $\lambda\in\mathcal{B}(\spl)=\mathcal{B}(\spl^*)$ has depth $\ell_\spl(\lambda)=\ell_{[\spl^*,\bm{k}]}(\lambda)$. Thus if we define $\bm{n}:=(\eta_0,\eta_1,\dots,\eta_{L-1})\in\N^L$ by $\eta_\ell:=m_{\ell+1}-m_\ell$, it follows that $(\spl,\bm{n})$ is a level pair such that $\spl\simeq\spl^*$ and every $\lambda\in\mathcal{B}(\spl)$ satisfies $$\sum_{\substack{\lambda'\in\mathcal{B}(\spl)\\\lambda'\supset\lambda}}k_{\lambda'}=m_{\ell_{[\spl^*,\bm{k}]}(\lambda)+1}+1=\sum_{\ell=0}^{\ell_{[\spl^*,\bm{k}]}(\lambda)}(m_{\ell+1}-m_\ell)=\sum_{\ell=0}^{\ell_\spl(\lambda)}\eta_\ell~.$$ Then $k_{[N]}=\eta_0$, and if $\lambda\in\mathcal{B}(\spl)\setminus\overline{\ptn}$ and $\lambda^*$ is the smallest branch in $\mathcal{B}(\spl)$ properly containing $\lambda$ we have $$k_\lambda=\sum_{\substack{\lambda'\in\mathcal{B}(\spl)\\\lambda'\supset\lambda}}k_{\lambda'}-\sum_{\substack{\lambda'\in\mathcal{B}(\spl)\\\lambda'\supset\lambda^*}}k_{\lambda'}=\sum_{\ell=0}^{\ell_\spl(\lambda)}\eta_\ell-\sum_{\ell=0}^{\ell_\spl(\lambda^*)}\eta_\ell=\sum_{\ell=\ell_\spl(\lambda^*)+1}^{\ell_\spl(\lambda)}\eta_\ell~.$$ Therefore by setting $F([\spl^*,\bm{k}]):=(\spl,\bm{n})$ we obtain a well-defined map $$F:\left\{[\spl^*,\bm{k}]:\bm{k}=(k_\lambda)\in\N^{\mathcal{B}(\spl^*)}\right\}\longrightarrow\bigsqcup_{\substack{\spl\in\mathcal{S}_N\\\spl\simeq\spl'}}\left\{(\spl,\bm{n}):\bm{n}=(\eta_0,\eta_1,\dots,\eta_{L(\spl)-1})\in\N^{L(\spl)}\right\}$$ satisfying \eqref{n_to_k_eqn}. We will now show that $F$ is a bijection by constructing an inverse. Let $\spl\in\mathcal{S}_N$ be any splitting chain with reduction $\spl^*$, let $\bm{n}=(\eta_0,\eta_1,\dots,\eta_{L(\spl)-1})$ be an arbitrary tuple of $L(\spl)$ positive integers, and define $G((\spl,\bm{n})):=[\spl^*,\bm{k}]$ by defining $k_\lambda\in\N$ for each $\lambda\in\mathcal{B}(\spl^*)=\mathcal{B}(\spl)$ via $$k_\lambda:=\begin{cases}\eta_0&\text{if }\lambda=[N],\\\displaystyle{\sum_{\ell=\ell_\spl(\lambda^*)+1}^{\ell_\spl(\lambda)}\eta_\ell}&\text{if $\lambda^*\in\mathcal{B}(\spl)$ is the smallest branch properly containing $\lambda$.}\end{cases}$$ Therefore we have a well-defined map $$G:\bigsqcup_{\substack{\spl\in\mathcal{S}_N\\\spl\simeq\spl^*}}\left\{(\spl,\bm{n}):\bm{n}=(\eta_0,\eta_1,\dots,\eta_{L(\spl)-1})\in\N^{L(\spl)}\right\}\longrightarrow\left\{[\spl^*,\bm{k}]:\bm{k}=(k_\lambda)\in\N^{\mathcal{B}(\spl^*)}\right\}~,$$ and it is immediate from \eqref{n_to_k_eqn} and the definition of $G$ that $G\circ F([\spl^*,\bm{k}])=[\spl^*,\bm{k}]$ for every $\bm{k}=(k_\lambda)$ indexed by $\lambda\in\mathcal{B}(\spl^*)$. It remains to show that $F\circ G((\spl,\bm{n}))=(\spl,\bm{n})$ for all level pairs in $$\bigsqcup_{\substack{\spl\in\mathcal{S}_N\\\spl\simeq\spl^*}}\left\{(\spl,\bm{n}):\bm{n}=(\eta_0,\eta_1,\dots,\eta_{L(\spl)-1})\in\N^{L(\spl)}\right\}~.$$ To this end, let $(\spl',\bm{n}')$ be such a level pair and suppose $[\spl^*,\bm{k}]=G((\spl',\bm{n}'))$, so that
\begin{equation}\label{gross}
k_\lambda=\begin{cases}\eta_0'&\text{if }\lambda=[N],\\\displaystyle{\sum_{\ell=\ell_{\spl'}(\lambda^*)+1}^{\ell_{\spl'}(\lambda)}\eta_\ell'}&\text{if $\lambda^*\in\mathcal{B}(\spl')$ is the smallest branch properly containing $\lambda$},\end{cases}
\end{equation}
for each $\lambda\in\mathcal{B}(\spl')$. Now suppose $(\spl,\bm{n})=F([\spl^*,\bm{k}])$ and recall the following details from our definition of $F$. The strictly increasing set of integers $\mathcal{M}=\{m_1,m_2,\dots,m_L\}$ is defined by $$\mathcal{M}=\left\{-1+\sum_{\substack{\lambda'\in\mathcal{B}(\spl^*)\\\lambda'\supset\lambda}}k_{\lambda'}:\lambda\in\mathcal{B}(\spl^*)\right\}$$ and satisfies $\eta_\ell=m_{\ell+1}-m_\ell$ for all $\ell\in\{0,1,\dots,L-1\}$, where $m_0=-1$. Moreover, recall that $\spl=(\ptn_0,\ptn_1,\dots,\ptn_L)$ is then completely determined using the integers defined for each $\lambda\in\mathcal{B}(\spl^*)$ by $$\ell_{[\spl^*,\bm{k}]}(\lambda)=\text{the unique $\ell\in\{0,1,\dots,L-1\}$ such that $\sum_{\substack{\lambda'\in\mathcal{B}(\spl^*)\\\lambda'\supset\lambda}}k_{\lambda'}=m_{\ell+1}+1~,$}$$ and we saw that $L(\spl)=L$, $\mathcal{B}(\spl)=\mathcal{B}(\spl^*)$, and $\ell_\spl(\lambda)=\ell_{[\spl^*,\bm{k}]}(\lambda)$ for all $\lambda\in\mathcal{B}(\spl)=\mathcal{B}(\spl^*)$. Now since $\mathcal{B}(\spl^*)=\mathcal{B}(\spl')$ and each integer $k_\lambda$ with $\lambda\in\mathcal{B}(\spl')$ is given by \eqref{gross}, we have $$\{m_1,m_2,\dots,m_L\}=\mathcal{M}=\left\{-1+\sum_{\substack{\lambda'\in\mathcal{B}(\spl')\\\lambda'\supset\lambda}}k_{\lambda'}:\lambda\in\mathcal{B}(\spl')\right\}=\left\{-1+\sum_{\ell=0}^{\ell_{\spl'}(\lambda)}\eta_\ell':\lambda\in\mathcal{B}(\spl')\right\}~.$$ In particular, for each $\lambda\in\mathcal{B}(\spl)=\mathcal{B}(\spl^*)=\mathcal{B}(\spl')$ we have
\begin{equation}\label{n2n'}
\sum_{\ell=0}^{\ell_\spl(\lambda)}\eta_\ell=m_{\ell_\spl(\lambda)+1}+1=\sum_{\substack{\lambda'\in\mathcal{B}(\spl^*)\\\lambda'\supset\lambda}}k_{\lambda'}=\sum_{\substack{\lambda'\in\mathcal{B}(\spl')\\\lambda'\supset\lambda}}k_{\lambda'}=\sum_{\ell=0}^{\ell_{\spl'}(\lambda)}\eta_\ell'~.
\end{equation}
Since $\spl'$ is a splitting chain, it must satisfy $\{[N]\}=\ptn_0'>\ptn_1'>\dots>\ptn_{L(\spl')}'=\underline{\ptn}$, and hence for each level index $\ell'\in\{0,1,2,\dots,L(\spl')-1\}$ we may select a branch $\lambda^{(\ell')}\in\mathcal{B}(\spl')\cap\ptn_{\ell'}'$ satisfying $\ell_{\spl'}(\lambda^{(\ell')})=\ell'$ and have $$L(\spl')-1=\ell_{\spl'}(\lambda^{(L(\spl')-1)})=\max\{\ell_{\spl'}(\lambda):\lambda\in\mathcal{B}(\spl')\}~.$$ Now since each $\eta_\ell'$ is positive, it follows that $$\{m_1,m_2,\dots,m_L\}=\left\{-1+\sum_{\ell=0}^{\ell_{\spl'}(\lambda)}\eta_\ell':\lambda\in\mathcal{B}(\spl')\right\}=\left\{-1+\sum_{\ell=0}^{\ell'}\eta_\ell':\ell'\in\{0,1,\dots,L(\spl')-1\}\right\}~.$$ But the values $m_1,m_2,\dots,m_L$ strictly increase and the sums $-1+\sum_{\ell=0}^{\ell'}\eta_\ell'$ also strictly increase with $\ell'$, so it must be the case that $L(\spl')=L=L(\spl)$ and moreover, $$m_{\ell'+1}=-1+\sum_{\ell=0}^{\ell'}\eta_\ell'\quad\text{for all $\ell'\in\{0,1,\dots,L(\spl')-1\}$}~.$$ Thus $\eta_0'=m_1+1=\eta_0$, and for every $\ell'\in\{1,\dots,L(\spl)-1\}$ we have $$\eta_{\ell'}=m_{\ell'+1}-m_{\ell'}=\left(-1+\sum_{\ell=0}^{\ell'}\eta_\ell'\right)-\left(-1+\sum_{\ell=0}^{\ell'-1}\eta_\ell'\right)=\eta_{\ell'}'~,$$ so we conclude that $\bm{n}=\bm{n}'$. Now \eqref{n2n'} and positivity of $\eta_\ell=\eta_\ell'$ imply $\ell_{\spl'}(\lambda)=\ell_\spl(\lambda)=\ell_{[\spl^*,\bm{k}]}(\lambda)$ for all $\lambda\in\mathcal{B}(\spl')=\mathcal{B}(\spl^*)=\mathcal{B}(\spl)$, so each partition $\ptn_\ell$ defined via the set $\mathcal{B}_\ell(\spl^*)$ above is precisely $\ptn_\ell'$. Therefore $\spl=\spl'$, so $$F\circ G((\spl',\bm{n}'))=F([\spl^*,\bm{k}])=(\spl,\bm{n})=(\spl',\bm{n}')$$ and we conclude that $G=F^{-1}$.\\
\end{proof}

\begin{figure}[h]
\begin{center}
\caption{Recall that the splitting pair $(\spl,\bm{n})$ associated to the tree in \Cref{Extree_1} had $\bm{n}=(2,1,3,2)$ in \Cref{fig_2}. By \Cref{bijection_thm}, $(\spl,\bm{n})$ corresponds to $[\spl^*,\bm{k}]$ where $\spl^*$ is the reduction computed in \Cref{reduction_alg} and $\bm{k}$ is displayed in the diagram below. Note that these $\bm{k}$ and $\bm{n}$ indeed satisfy \eqref{n_to_k_eqn}.}\label{fig_3}
\vspace{0.5cm}

\includegraphics[scale=1.2]{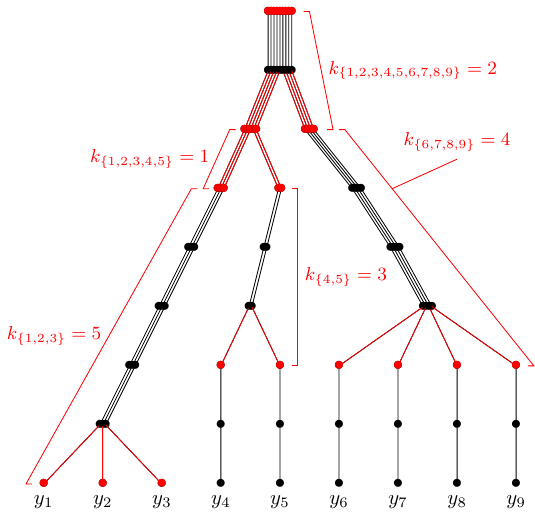}
\end{center}
\end{figure}

With \Cref{reduction} and \Cref{bijection_thm} in hand, we may now give a ``branch-centric" reinterpretation of \Cref{pairset_integral} in the $b=0$ case.

\begin{corollary}\label{bpair_integral}
If $a\in\C$, $[\spl^*,\bm{k}]$ is a branch pair, and $(\spl,\bm{n})$ is the level pair corresponding to $[\spl^*,\bm{k}]$, then for every $\bm{s}\in\C^{\binom{N}{2}}$ we have
\begin{align*}\int_{\mathcal{T}(\spl,\bm{n})}\big(\max_{i<j}|x_i-x_j|\big)^a&\prod_{i<j}|x_i-x_j|^{s_{ij}}\,|dx|\\
&=q^{-(N-1+a+\sum_{i<j}s_{ij})(k_{[N]}-1)}\cdot\frac{M_{\spl^*\!}(q)}{q^{N-1}}\cdot\prod_{\lambda\in\mathcal{B}(\spl^*)\setminus\overline{\ptn}}q^{-e_\lambda(\bm{s})k_\lambda}~.
\end{align*}
\end{corollary}

\begin{proof}
If $b=0$, \Cref{pairset_integral} gives
\begin{align*}\int_{\mathcal{T}(\spl,\bm{n})}\big(\max_{i<j}|x_i-x_j|\big)^a&\prod_{i<j}|x_i-x_j|^{s_{ij}}\,|dx|\\
&=q^{-(N-1+a+\sum_{i<j}s_{ij})(\eta_0-1)}\cdot\frac{M_\spl(q)}{q^{N-1}}\cdot\prod_{\ell=1}^{L(\spl)-1}q^{-E_{\ptn_\ell}(\bm{s})\eta_\ell}~.
\end{align*}
Since $\spl\simeq\spl^*$, part (a) of \Cref{reduction} implies $M_{\spl^*\!}(q)=M_\spl(q)$ and $\mathcal{B}(\spl^*)=\mathcal{B}(\spl)$. We also have $k_{[N]}=\eta_0$ by \Cref{bijection_thm}, so it suffices to show that
\begin{equation}\label{exp_k_to_n}
\sum_{\ell=1}^{L(\spl)-1}E_{\ptn_\ell}(\bm{s})\eta_\ell=\sum_{\lambda\in\mathcal{B}(\spl)\setminus\overline{\ptn}}e_\lambda(\bm{s})k_\lambda~.
\end{equation} To see why \eqref{exp_k_to_n} is true, note that $$E_{\ptn_\ell}(\bm{s})=\sum_{\lambda\in\mathcal{B}(\spl)\cap\ptn_\ell}e_\lambda(\bm{s})~,$$ and for $\ell\in\{1,2,\dots,L(\spl)-1\}$ we have $\lambda\in\mathcal{B}(\spl)\cap\ptn_\ell$ if and only if $\ell_\spl(\lambda^*)+1\leq\ell\leq\ell_\spl(\lambda)$, where $\lambda^*$ denotes the smallest branch in $\mathcal{B}(\spl)$ properly containing $\lambda$. Therefore if $\lambda\in\mathcal{B}(\spl)\setminus\overline{\ptn}$, then the branch exponent $e_\lambda(\bm{s})$ is a summand of $E_{\ptn_\ell}(\bm{s})$ if and only if $\ell_\spl(\lambda^*)+1\leq\ell\leq\ell_\spl(\lambda)$, so we have $$\sum_{\ell=1}^{L(\spl)-1}E_{\ptn_\ell}(\bm{s})\eta_\ell=\sum_{\lambda\in\mathcal{B}(\spl)\setminus\overline{\ptn}}\left(\sum_{\ell=\ell_\spl(\lambda^*)+1}^{\ell_\spl(\lambda)}e_\lambda(\bm{s})\eta_\ell\right)=\sum_{\lambda\in\mathcal{B}(\spl)\setminus\overline{\ptn}}e_\lambda(\bm{s})\left(\sum_{\ell=\ell_\spl(\lambda^*)+1}^{\ell_\spl(\lambda)}\eta_\ell\right)~.$$ But $k_\lambda=\sum_{\ell=\ell_\spl(\lambda^*)+1}^{\ell_\spl(\lambda)}\eta_\ell$ by \eqref{n_to_k_eqn} in \Cref{bijection_thm}, so \eqref{exp_k_to_n} is proved and the corollary follows.\\
\end{proof}

We continue our ``branch-centric" discussion with an analogue of \Cref{level_reason}.

\begin{remark}
Note that the integral formula in \Cref{bpair_integral} provides yet another method for computing $\mu^N(\mathcal{T}(\spl,\bm{n}))$, but now in terms of the branch pair $[\spl^*,\bm{k}]$ corresponding to $(\spl,\bm{n})$. Indeed, setting $s_{ij}=a=0$ for all $i<j$ gives $e_\lambda(\bm{s})=\#\lambda-1$ by part (b) of \Cref{branchlevel}, and then the formula in \Cref{bpair_integral} simplifies very nicely:
\begin{equation}\label{branch_measure}
\mu^N(\mathcal{T}(\spl,\bm{n}))=M_{\spl^*\!}(q)\cdot\prod_{\lambda\in\mathcal{B}(\spl^*)}q^{-(\#\lambda-1)k_\lambda}~.
\end{equation}
The exponential factors in the formula in \Cref{bpair_integral} are completely determined by the branch pair $[\spl^*,\bm{k}]$ corresponding to the level pair $(\spl,\bm{n})$. Since $k_{[N]}=\eta_0$ in this case, the leftmost factor $q^{-(N-1+a+\sum_{i<j}s_{ij})(k_{[N]}-1)}$ pertains to ``root data" and the root polytope (just as in \Cref{level_reason}), with $b=0$. The ``branch data" that determine the factor $q^{-e_\lambda(\bm{s})k_\lambda}$ is comprised of the branch $\lambda\in\mathcal{B}(\spl^*)\setminus\overline{\ptn}=\mathcal{B}(\spl)\setminus\overline{\ptn}$ and the integer $k_\lambda$, which have clear visual interpretations in the tree diagram for any $x\in\mathcal{T}(\spl,\bm{n})$ (recall \Cref{fig_3}). In analogy with \eqref{level_criterion} in \Cref{level_reason}, we have
\begin{equation}\label{branch_criterion}
|q^{-e_\lambda(\bm{s})}|_{\C}<1\quad\text{for all}\quad\lambda\in\mathcal{B}(\spl^*)\setminus\overline{\ptn}\qquad\iff\qquad\bm{s}\in\mathcal{BP}_{\spl^*}~,
\end{equation}
which is precisely why we call $\mathcal{BP}_{\spl^*}$ the branch polytope.
\end{remark}

We now give the ``branch-centric" analogue of \Cref{spl_integral}, which will have a similar proof and a similar purpose. Just as for level functions in \Cref{spl_integral}, this is where branch functions enter the picture.

\begin{proposition}\label{bspl_integral}
Suppose $\spl^*\in\mathcal{R}_N$ and $a\in\C$. If $M_{\spl^*\!}(q)>0$, then for every $\spl\simeq\spl^*$ the integral $$\int_{R_\spl^N}\big(\max_{i<j}|x_i-x_j|\big)^a\prod_{i<j}|x_i-x_j|^{s_{ij}}\,|dx|$$ converges absolutely for all $\bm{s}\in\mathcal{RP}_N(a,0)\cap\mathcal{BP}_{\spl^*}$, and for such $\bm{s}$ we have $$\sum_{\substack{\spl\in\mathcal{S}_N\\\spl\simeq\spl^*}}\int_{R_\spl^N}\big(\max_{i<j}|x_i-x_j|\big)^a\prod_{i<j}|x_i-x_j|^{s_{ij}}\,|dx|=\frac{1}{1-q^{-(N-1+a+\sum_{i<j}s_{ij})}}\cdot I_{\spl^*\!\!,q}(\bm{s})~.$$ Otherwise $M_{\spl^*\!}(q)=0$, in which case $R_\spl^N=\varnothing$ for all $\spl\simeq\spl^*$ and all integrals above are zero.
\end{proposition}

\begin{proof}
The $M_{\spl^*\!}(q)=0$ case is immediate from \eqref{branch_measure} and the definition of $R_\spl^N$, so suppose $M_{\spl^*\!}(q)>0$. The first claim follows from part (c) of \Cref{reduction} and \Cref{spl_integral}. To prove the second claim, suppose $\bm{s}\in\mathcal{RP}_N(a,0)\cap\mathcal{BP}_{\spl^*}$, note that the function $$x\mapsto\sum_{\substack{\spl\in\mathcal{S}_N\\\spl\simeq\spl^*}}\bm{1}_{R_\spl^N}(x)\left|\big(\max_{i<j}|x_i-x_j|\big)^a\prod_{i<j}|x_i-x_j|^{s_{ij}}\right|_{\C}$$ is in $L^1(K^N,\mu^N)$ by \Cref{spl_integral}, and that it dominates every partial sum of the function $$x\mapsto\sum_{\substack{\spl\in\mathcal{S}_N\\\spl\simeq\spl^*}}\sum_{\bm{n}\in\N^{L(\spl)}}\bm{1}_{\mathcal{T}(\spl,\bm{n})}(x)\big(\max_{i<j}|x_i-x_j|\big)^a\prod_{i<j}|x_i-x_j|^{s_{ij}}~.$$ Then the Dominated Convergence Theorem, \Cref{bijection_thm}, \Cref{bpair_integral}, Fubini's Theorem for absolutely convergent sums, \eqref{decomp}, \eqref{root_criterion}, and \eqref{branch_criterion} imply
\begin{align*}
\sum_{\substack{\spl\in\mathcal{S}_N\\\spl\simeq\spl^*}}&\int_{R_\spl^N}\big(\max_{i<j}|x_i-x_j|\big)^a\prod_{i<j}|x_i-x_j|^{s_{ij}}\,|dx|\\
&=\sum_{\substack{\spl\in\mathcal{S}_N\\\spl\simeq\spl^*}}\sum_{\bm{n}\in\N^{L(\spl)}}\int_{\mathcal{T}(\spl,\bm{n})}\big(\max_{i<j}|x_i-x_j|\big)^a\prod_{i<j}|x_i-x_j|^{s_{ij}}\,|dx|\\
&=\sum_{\bm{k}\in\N^{\mathcal{B}(\spl^*)}}q^{-(N-1+a+\sum_{i<j}s_{ij})(k_{[N]}-1)}\cdot\frac{M_{\spl^*\!}(q)}{q^{N-1}}\cdot\prod_{\lambda\in\mathcal{B}(\spl^*)\setminus\overline{\ptn}}q^{-e_\lambda(\bm{s})k_\lambda}\\
&=\sum_{k_{[N]}=1}^\infty q^{-(N-1+a+\sum_{i<j}s_{ij})(k_{[N]}-1)}\cdot\frac{M_{\spl^*\!}(q)}{q^{N-1}}\cdot\prod_{\lambda\in\mathcal{B}(\spl^*)\setminus\overline{\ptn}}\sum_{k_\lambda=1}^\infty q^{-e_\lambda(\bm{s})k_\lambda}\\
&=\frac{1}{1-q^{-(N-1+a+\sum_{i<j}s_{ij})}}\cdot I_{\spl^*\!\!,q}(\bm{s})~.
\end{align*}
\end{proof}

\Cref{bspl_integral} is the second of the three main components of the proof of \Cref{main}. In fact, we can easily prove the first statement in part (c) of \Cref{main} now: Given $\spl^*\in\mathcal{R}_N$ with $M_{\spl^*\!}(q)>0$ and $a=b=0$, the two formulas in \Cref{spl_integral} in \Cref{bspl_integral} imply $$\sum_{\substack{\spl\in\mathcal{S}_N\\\spl\simeq\spl^*}}J_{\spl,q}(0,\bm{s})=(1-q^{-(N-1+\sum_{i<j}s_{ij})})\cdot\sum_{\substack{\spl\in\mathcal{S}_N\\\spl\simeq\spl^*}}\int_{R_\spl^N}\prod_{i<j}|x_i-x_j|^{s_{ij}}\,|dx|=I_{\spl^*\!\!,q}(\bm{s})$$ for all $\bm{s}\in\mathcal{RP}_N(0,0)\cap\mathcal{BP}_{\spl^*}$. The left-hand and right-hand expressions above are both holomorphic in the open set $\mathcal{BP}_{\spl^*}$, which is also simply connected because it is convex. Therefore since the two expressions agree on $\mathcal{RP}_N(0,0)\cap\mathcal{BP}_{\spl^*}$, they must in fact agree on all of $\mathcal{BP}_{\spl^*}$. Otherwise $M_{\spl^*\!}(q)=0$ implies all three expressions above are identically zero on $\mathcal{BP}_{\spl^*}$, so the first statement in part (c) of \Cref{main} is proved in all cases. We conclude this subsection with the following analogue of \Cref{pair_conclusion}, which is immediate from \Cref{bspl_integral} and part (c) of \Cref{reduction}:

\begin{corollary}\label{bspl_conclusion}
Suppose the residue field of $K$ has cardinality $q$ and suppose $a\in\C$. The integral $$\int_{R^N}\big(\max_{i<j}|x_i-x_j|\big)^a\prod_{i<j}|x_i-x_j|^{s_{ij}}\,|dx|$$ converges absolutely for all $\bm{s}\in\mathcal{RP}_N(a,0)\cap\bigcap_{\spl^*\in\mathcal{R}_N}\mathcal{BP}_{\spl^*}=\Omega_N(a,0)$, and for such $\bm{s}$ it converges to $$\frac{1}{1-q^{-(N-1+a+\sum_{i<j}s_{ij})}}\cdot\sum_{\spl^*\in\mathcal{R}_N}I_{\spl^*\!\!,q}(\bm{s})~.$$
\end{corollary}

\subsection{The final step}\label{3_5}

We are now ready for the third and final part of the proof of \Cref{main}, which is the following:

\begin{lemma}\label{rho_thm}
Suppose $K$ is a $p$-field with residue field cardinality $q$, suppose $a,b\in\C$, suppose $\rho:\mathcal{N}\to\C$ satisfies \eqref{rhogro}, and define $$Z_N(K,a,b,\bm{s}):=\int_{R^N}\big(\max_{i<j}|x_i-x_j|\big)^a\big(\min_{i<j}|x_i-x_j|\big)^b\prod_{i<j}|x_i-x_j|^{s_{ij}}\,|dx|$$ for all $\bm{s}\in\Omega_N(a,b)$. Then for all such $\bm{s}$ we have $$Z_N^\rho(K,a,b,\bm{s})=\left(\sum_{m\in\Z}\frac{\rho(q^{-m})}{q^{m(N+a+b+\sum_{i<j}s_{ij})}}\right)\left(1-\frac{1}{q^{N+a+b+\sum_{i<j}s_{ij}}}\right)Z_N(K,a,b,\bm{s})~,$$ and the sum over $m\in\Z$ converges absolutely uniformly on each compact subset of $\Omega_N(a,b)$.
\end{lemma}

\begin{proof}
We first prove the following claim: For each $m\in\Z$ and every $\bm{s}\in\Omega_N(a,b)$ the integral $$\int_{(P^m)^N\setminus(P^{m+1})^N}\rho(\|x\|)\big(\max_{i<j}|x_i-x_j|\big)^a\big(\min_{i<j}|x_i-x_j|\big)^b\prod_{i<j}|x_i-x_j|^{s_{ij}}\,|dx|$$ converges absolutely to $$\frac{\rho(q^{-m})}{q^{m(N+a+b+\sum_{i<j}s_{ij})}}\left(1-\frac{1}{q^{N+a+b+\sum_{i<j}s_{ij}}}\right)Z_N(K,a,b,\bm{s})~.$$ To see why this claim holds, note that $Z_N(K,a,b,\bm{s})$ is defined for all $\bm{s}\in\Omega_N(a,b)$ by \Cref{pair_conclusion}. Then for any $m\in\Z$, the change of variables $R^N\to(P^m)^N$ defined by $x\mapsto\pi^my$ gives
\begin{align*}
\int_{(P^m)^N}&\big(\max_{i<j}|x_i-x_j|\big)^a\big(\min_{i<j}|x_i-x_j|\big)^b\prod_{i<j}|x_i-x_j|^{s_{ij}}\,|dx|\\
&=\frac{1}{q^{mN}}\int_{R^N}\big(\max_{i<j}|\pi^my_i-\pi^my_j|\big)^a\big(\min_{i<j}|\pi^my_i-\pi^my_j|\big)^b\prod_{i<j}|\pi^my_i-\pi^my_j|^{s_{ij}}\,|dy|\\
&=\frac{1}{q^{m(N+a+b+\sum_{i<j}s_{ij})}}\int_{R^N}\big(\max_{i<j}|y_i-y_j|\big)^a\big(\min_{i<j}|y_i-y_j|\big)^b\prod_{i<j}|y_i-y_j|^{s_{ij}}\,|dy|\\
&=\frac{1}{q^{m(N+a+b+\sum_{i<j}s_{ij})}}\cdot Z_N(K,a,b,\bm{s})
\end{align*}
for all $\bm{s}\in\Omega_N(a,b)$. But the norm $\|x\|=\max_{1\leq i\leq N}|x_i|$ takes the constant value $q^{-m}$ at every $x\in(P^m)^N\setminus(P^{m+1})^N$, so for every $m\in\Z$ and every $\bm{s}\in\Omega_N(a,b)$ we have
\begin{align*}
&\int_{(P^m)^N\setminus(P^{m+1})^N}\rho(\|x\|)\big(\max_{i<j}|x_i-x_j|\big)^a\big(\min_{i<j}|x_i-x_j|\big)^b\prod_{i<j}|x_i-x_j|^{s_{ij}}\,|dx|\\
&~~~~~~=\rho(q^{-m})\left(\frac{1}{q^{m(N+a+b+\sum_{i<j}s_{ij})}}\cdot Z_N(K,a,b,\bm{s})-\frac{1}{q^{(m+1)(N+a+b+\sum_{i<j}s_{ij})}}\cdot Z_N(K,a,b,\bm{s})\right)\\
&~~~~~~=\frac{\rho(q^{-m})}{q^{m(N+a+b+\sum_{i<j}s_{ij})}}\left(1-\frac{1}{q^{N+a+b+\sum_{i<j}s_{ij}}}\right)Z_N(K,a,b,\bm{s})
\end{align*}
and the desired claim is proved. In particular, since $(\re(s_{ij}))_{i<j}\in\Omega_N(\re(a),\re(b))$ whenever $\bm{s}\in\Omega_N(a,b)$, note that the claim also holds if $\rho(\cdot)$, $a$, $b$, and $s_{ij}$ are replaced by $|\rho(\cdot)|_{\C}$, $\re(a)$, $\re(b)$, and $\re(s_{ij})$. Now for the main claim, note that
\begin{align*}
\rho(\|x\|)&\big(\max_{i<j}|x_i-x_j|\big)^a\big(\min_{i<j}|x_i-x_j|\big)^b\prod_{i<j}|x_i-x_j|^{s_{ij}}\\
&=\sum_{m\in\Z}\left(\rho(q^{-m})\big(\max_{i<j}|x_i-x_j|\big)^a\big(\min_{i<j}|x_i-x_j|\big)^b\prod_{i<j}|x_i-x_j|^{s_{ij}}\right)\bm{1}_{(P^m)^N\setminus(P^{m+1})^N}(x)
\end{align*}
for all $x\in K^N\setminus\{0\}$, and therein each partial sum is dominated by the function
\begin{align*}
x\mapsto&\left|\rho(\|x\|)\big(\max_{i<j}|x_i-x_j|\big)^a\big(\min_{i<j}|x_i-x_j|\big)^b\prod_{i<j}|x_i-x_j|^{s_{ij}}\right|_{\C}\\
&~~~=\sum_{m\in\Z}|\rho(q^{-m})|_{\C}\big(\max_{i<j}|x_i-x_j|\big)^{\re(a)}\big(\min_{i<j}|x_i-x_j|\big)^{\re(b)}\prod_{i<j}|x_i-x_j|^{\re(s_{ij})}\bm{1}_{(P^m)^N\setminus(P^{m+1})^N}(x)~.
\end{align*}
Now Fubini's Theorem for sums of nonnegative terms and the claim we just proved give
\begin{align*}
\int_{K^N}&\left|\rho(\|x\|)\big(\max_{i<j}|x_i-x_j|\big)^a\big(\min_{i<j}|x_i-x_j|\big)^b\prod_{i<j}|x_i-x_j|^{s_{ij}}\right|_{\C}|dx|\\
&=\sum_{m\in\Z}\int_{(P^m)^N\setminus(P^{m+1})^N}|\rho(q^{-m})|_{\C}\big(\max_{i<j}|x_i-x_j|\big)^{\re(a)}\big(\min_{i<j}|x_i-x_j|\big)^{\re(b)}\prod_{i<j}|x_i-x_j|^{\re(s_{ij})}\,|dx|\\
&=\sum_{m\in\Z}\frac{|\rho(q^{-m})|_{\C}}{q^{m(\re(N+a+b+\sum_{i<j}s_{ij}))}}\left(1-\frac{1}{q^{\re(N+a+b+\sum_{i<j}s_{ij})}}\right)Z_N(K,\re(a),\re(b),(\re(s_{ij}))_{i<j})
\end{align*}
for every $\bm{s}\in\Omega_N(a,b)$. Now suppose $C$ is any compact subset of $\Omega_N(a,b)$. Since $C$ is therefore a compact subset of $\mathcal{RP}_N(a,b)=\{\bm{s}\in\C^{\binom{N}{2}}:\re(N-1+a+b+\sum_{i<j}s_{ij})>0\}$, there exist real numbers $\sigma_1$ and $\sigma_2$ satisfying $$\limsup_{n\to\infty}\frac{\log|\rho(\frac{1}{n})|_{\C}}{\log(n)}\leq 1<\sigma_1\leq\re\bigg(N+a+b+\sum_{i<j}s_{ij}\bigg)\leq\sigma_2<\infty=-\limsup_{n\to\infty}\frac{\log|\rho(n)|_{\C}}{\log(n)}$$ for all $\bm{s}\in C$. Then to see that the preceding sum over $m\in\Z$ converges uniformly on $C$, it suffices to verify the convergence of the two series $$\sum_{m=0}^\infty\frac{|\rho(q^{-m})|_{\C}}{q^{m\sigma_1}}\qquad\text{and}\qquad\sum_{m=1}^\infty|\rho(q^m)|_{\C} q^{m\sigma_2}~.$$ Indeed, if $\log:[0,\infty]\to[-\infty,\infty]$ is the extended logarithm we have
\begin{align*}
\log\left(\limsup_{m\to\infty}\sqrt[m]{\frac{|\rho(q^{-m})|_{\C}}{q^{m\sigma_1}}}\right)&=\log(q)\cdot\left(\limsup_{m\to\infty}\frac{\log|\rho(q^{-m})|_{\C}}{\log(q^m)}-\sigma_1\right)\\
&\leq\log(q)\cdot\left(\limsup_{n\to\infty}\frac{\log|\rho(\frac{1}{n})|_{\C}}{\log(n)}-\sigma_1\right)<0
\end{align*}
and
\begin{align*}
\log\left(\limsup_{m\to\infty}\sqrt[m]{|\rho(q^m)|_{\C} q^{m\sigma_2}}\right)&=\log(q)\cdot\left(\limsup_{m\to\infty}\frac{\log|\rho(q^m)|_{\C}}{\log(q^m)}+\sigma_2\right)\\
&\leq\log(q)\cdot\left(\limsup_{n\to\infty}\frac{\log|\rho(n)|_{\C}}{\log(n)}+\sigma_2\right)<0~,
\end{align*}
so the series both converge by the root test and we conclude that our series expansion for $$\int_{K^N}\left|\rho(\|x\|)\big(\max_{i<j}|x_i-x_j|\big)^a\big(\min_{i<j}|x_i-x_j|\big)^b\prod_{i<j}|x_i-x_j|^{s_{ij}}\right|_{\C}|dx|$$ converges uniformly on $C$. Thus by the dominated convergence theorem we have
\begin{align*}
Z_N^\rho(K,a,b,\bm{s})&=\int_{K^N}\rho(\|x\|)\big(\max_{i<j}|x_i-x_j|\big)^a\big(\min_{i<j}|x_i-x_j|\big)^b\prod_{i<j}|x_i-x_j|^{s_{ij}}|dx|\\
&=\left(\sum_{m\in\Z}\frac{\rho(q^{-m})}{q^{m(N+a+b+\sum_{i<j}s_{ij})}}\right)\left(1-\frac{1}{q^{N+a+b+\sum_{i<j}s_{ij}}}\right)Z_N(K,a,b,\bm{s})~,
\end{align*}
and hence the sum over $m\in\Z$ converges absolutely uniformly on $C$.\\
\end{proof}

Finally, we combine \Cref{rho_thm} with \Cref{pair_conclusion} and \Cref{bspl_integral} to finish the proof of \Cref{main}:

\begin{proof}[Proof of \Cref{main}]\
\begin{itemize}
\item[(a)] Since $\rho$ is not identically zero, there exists $m\in\Z$ such that $\rho(q^{-m})\neq 0$. Moreover, the quantity $1-\frac{1}{q^{N+a+b+\sum_{i<j}s_{ij}}}$ attains nonzero values on every open subset $U\subset\C^{\binom{N}{2}}$, so term $$\frac{\rho(q^{-m})}{q^{m(N+a+b+\sum_{i<j}s_{ij})}}\left(1-\frac{1}{q^{N+a+b+\sum_{i<j}s_{ij}}}\right)Z_N(K,a,b,\bm{s})$$ appearing in the proof above may converge absolutely at every point of an open set $U\subset\C^{\binom{N}{2}}$ only if the integral $Z_N(K,a,b,\bm{s})$ does. But \Cref{pair_conclusion} says that the integral defining $Z_N(K,a,b,\bm{s})$ converges absolutely if and only if $\bm{s}\in\Omega_N(a,b)$, and we know that the parenthetical sum over $m\in\Z$ in \Cref{rho_thm} converges absolutely uniformly on $\Omega_N(a,b)$. Thus $Z_N^\rho(K,a,b,\bm{s})$ converges absolutely for every $\bm{s}\in\Omega_N(a,b)$, and $\Omega_N(a,b)$ is the largest open set with this property.
\item[(b)] If $C$ is a compact subset of $\Omega_N(a,b)$, then $Z_N(K,a,b,\bm{s})$ restricts to a continuous and hence bounded function on $C$, and note that the same is true for the function $\bm{s}\mapsto1-\frac{1}{q^{N+a+b+\sum_{i<j}s_{ij}}}$. We already showed that the parenthetical sum in \Cref{rho_thm} converges uniformly on $C$, so by \Cref{rho_thm}, \Cref{pair_conclusion}, and \Cref{rootdef} we have
\begin{align*}
Z_N^\rho(K,a,b,\bm{s})&=\left(\sum_{m\in\Z}\rho(q^m)q^{m(N+a+b+\sum_{i<j}s_{ij})}\right)\cdot\frac{1-q^{-(N+a+b+\sum_{i<j}s_{ij})}}{1-q^{-(N-1+a+b+\sum_{i<j}s_{ij})}}\cdot\sum_{\spl\in\mathcal{S}_N}J_{\spl,q}(b,\bm{s})\\
&=H_q^\rho\left(N+a+b+\sum_{i<j}s_{ij}\right)\cdot\sum_{\spl\in\mathcal{S}_N}J_{\spl,q}(b,\bm{s})~,
\end{align*}
and the sum converges uniformly on $C$.\\
\item[(c)] We already proved the first claim relating level and branch functions immediately after the proof of \Cref{bspl_integral}. If $C$ is a compact subset of $\mathcal{RP}_N(a,0)\cap\bigcap_{\spl^*\in\mathcal{R}_N}\mathcal{BP}_{\spl^*}$, then $Z_N(K,a,0,\bm{s})$ (i.e., the value of the integral from \Cref{bspl_conclusion}) restricts to a continuous function on $C$. But $$\mathcal{RP}_N(a,0)\cap\bigcap_{\spl^*\in\mathcal{R}_N}\mathcal{BP}_{\spl^*}=\Omega_N(a,0)~,$$ so \Cref{rho_thm}, \Cref{bspl_conclusion}, and \Cref{rootdef} similarly imply
\begin{align*}
Z_N^\rho(K,a,0,\bm{s})&=\left(\sum_{m\in\Z}\rho(q^m)q^{m(N+a+\sum_{i<j}s_{ij})}\right)\cdot\frac{1-q^{-(N+a+\sum_{i<j}s_{ij})}}{1-q^{-(N-1+a+\sum_{i<j}s_{ij})}}\cdot\sum_{\spl^*\in\mathcal{R}_N}I_{\spl^*\!\!,q}(\bm{s})\\
&=H_q^\rho\left(N+a+\sum_{i<j}s_{ij}\right)\cdot\sum_{\spl^*\in\mathcal{R}_N}I_{\spl^*\!\!,q}(\bm{s})~,
\end{align*}
and the sum converges uniformly on $C$.
\end{itemize}
\end{proof}

\appendix

\section{Explicit computation for $N=4$}
Let $N=4$ and fix $a$, $b$, and $\rho$ as in \Cref{Ex_N2n3}. Then $\bm{s}=(s_{12},s_{13},s_{14},s_{23},s_{24},s_{34})\in\C^6$ and we have the root polytope $$\mathcal{RP}_4(a,b)=\{\bm{s}\in\C^6:\re(3+a+b+\sum_{1\leq i<j\leq 4}s_{ij})>0\}~,$$ on which the root function is holomorphic and defined by $$\bm{s}\mapsto\frac{1-q^{-(4+a+b+\sum_{1\leq i<j\leq 4}s_{ij})}}{1-q^{-(3+a+b+\sum_{1\leq i<j\leq 4}s_{ij})}}\cdot\sum_{m\in\Z}\rho(q^m)q^{m(4+a+b+\sum_{1\leq i<j\leq 4}s_{ij})}~.$$ There are 32 splitting chains $\spl\in\mathcal{S}_4$, so we will save table space below by suppressing the partition labels ``$\ptn_\ell=$" and expressing each level polytope $\bigcap_{\ell=1}^{L(\spl)-1}\{\bm{s}\in\C^6:\re(b+E_{\ptn_\ell}(\bm{s}))>0\}$ by simply listing the conditions $\re(b+E_{\ptn_\ell}(\bm{s}))>0$. Given $\spl\in\mathcal{S}_4$, the level $\ptn_1$ can either contain one part of size 3 (and a singleton), one part of size 2 (and two singletons), two parts of size 2, or four singletons. Thus it will be practical to sort $\spl\in\mathcal{S}_4$ according to the form of $\ptn_1$: 

\begin{itemize}
\item[(1)] There are four $\spl\in\mathcal{S}_4$ with $\ptn_1=\{1,2,3\}\{4\}$. Unsurprisingly, they form a table very similar to the one for $\mathcal{S}_3$ in \Cref{Ex_N2n3}:
\begin{center}
\def\arraystretch{1.2}
\begin{tabular}{c||c|c}
$\spl$ & $J_{\spl,q}(b,\bm{s})$ & $\mathcal{LP}_\spl(b)$\\
\hline\hline
\begin{tabular}{c}
$\{1,2,3,4\}$\\
$\{1,2,3\}\{4\}$\\
$\{1\}\{2\}\{3\}\{4\}$
\end{tabular}
& $\dfrac{(q-1)^2(q-2)}{q^3}\cdot\dfrac{1}{q^{2+b+s_{12}+s_{13}+s_{23}}-1}$ & $\re(2+b+s_{12}+s_{13}+s_{23})>0$\\

\hline
\begin{tabular}{c}
$\{1,2,3,4\}$\\
$\{1,2,3\}\{4\}$\\
$\{1,2\}\{3\}\{4\}$\\
$\{1\}\{2\}\{3\}\{4\}$
\end{tabular}
&
\begin{tabular}{r}
$\dfrac{(q-1)^3}{q^3}\cdot\dfrac{1}{q^{2+b+s_{12}+s_{13}+s_{23}}-1}$\\
$\cdot\dfrac{1}{q^{1+b+s_{12}}-1}$
\end{tabular}
&
\begin{tabular}{c}
$\re(2+b+s_{12}+s_{13}+s_{23})>0$\\
and\\
$\re(1+b+s_{12})>0$
\end{tabular}\\

\hline
\begin{tabular}{c}
$\{1,2,3,4\}$\\
$\{1,2,3\}\{4\}$\\
$\{1,3\}\{2\}\{4\}$\\
$\{1\}\{2\}\{3\}\{4\}$
\end{tabular}
&
\begin{tabular}{r}
$\dfrac{(q-1)^3}{q^3}\cdot\dfrac{1}{q^{2+b+s_{12}+s_{13}+s_{23}}-1}$\\
$\cdot\dfrac{1}{q^{1+b+s_{13}}-1}$
\end{tabular}
&
\begin{tabular}{c}
$\re(2+b+s_{12}+s_{13}+s_{23})>0$\\
and\\
$\re(1+b+s_{13})>0$
\end{tabular}\\

\hline
\begin{tabular}{c}
$\{1,2,3,4\}$\\
$\{1,2,3\}\{4\}$\\
$\{2,3\}\{1\}\{4\}$\\
$\{1\}\{2\}\{3\}\{4\}$
\end{tabular}
&
\begin{tabular}{r}
$\dfrac{(q-1)^3}{q^3}\cdot\dfrac{1}{q^{2+b+s_{12}+s_{13}+s_{23}}-1}$\\
$\cdot\dfrac{1}{q^{1+b+s_{23}}-1}$
\end{tabular}
&
\begin{tabular}{c}
$\re(2+b+s_{12}+s_{13}+s_{23})>0$\\
and\\
$\re(1+b+s_{23})>0$
\end{tabular}
\end{tabular}
\end{center}
All four of the splitting chains in the table are reduced and each satisfies $J_{\spl,q}(0,\bm{s})=I_{\spl,q}(\bm{s})$. There are also four $\spl\in\mathcal{S}_4$ satisfying $\ptn_1=\{1,2,4\}\{3\}$, and their table is obtained by simply transposing the indices $3$ and $4$ in the table above. Similarly, there are another four $\spl\in\mathcal{S}_4$ satisfying $\ptn_1=\{1,3,4\}\{2\}$ and another four satisfying $\ptn_1=\{2,3,4\}\{1\}$. Thus there are 16 distinct $\spl\in\mathcal{S}_4$ such that $\ptn_1$ has a part of size 3, and all of them are reduced.

\item[(2)] There are six $\spl\in\mathcal{S}_4$ such that $\ptn_1$ contains a single part of size 2. All six are reduced and satisfy $J_{\spl,q}(0,\bm{s})=I_{\spl,q}(\bm{s})$:

\begin{center}
\def\arraystretch{1.2}
\begin{tabular}{c||c|c}
$\spl$ & $J_{\spl,q}(b,\bm{s})$ & $\mathcal{LP}_\spl(b)$\\
\hline\hline
\begin{tabular}{c}
$\{1,2,3,4\}$\\
$\{1,2\}\{3\}\{4\}$\\
$\{1\}\{2\}\{3\}\{4\}$
\end{tabular}
& $\dfrac{(q-1)^2(q-2)}{q^3}\cdot\dfrac{1}{q^{1+b+s_{12}}-1}$ & $\re(1+b+s_{12})>0$\\

\hline
\begin{tabular}{c}
$\{1,2,3,4\}$\\
$\{1,3\}\{2\}\{4\}$\\
$\{1\}\{2\}\{3\}\{4\}$
\end{tabular}
& $\dfrac{(q-1)^2(q-2)}{q^3}\cdot\dfrac{1}{q^{1+b+s_{13}}-1}$ & $\re(1+b+s_{13})>0$\\

\hline
\begin{tabular}{c}
$\{1,2,3,4\}$\\
$\{1,4\}\{2\}\{3\}$\\
$\{1\}\{2\}\{3\}\{4\}$
\end{tabular}
& $\dfrac{(q-1)^2(q-2)}{q^3}\cdot\dfrac{1}{q^{1+b+s_{14}}-1}$ & $\re(1+b+s_{14})>0$\\

\hline
\begin{tabular}{c}
$\{1,2,3,4\}$\\
$\{2,3\}\{1\}\{4\}$\\
$\{1\}\{2\}\{3\}\{4\}$
\end{tabular}
& $\dfrac{(q-1)^2(q-2)}{q^3}\cdot\dfrac{1}{q^{1+b+s_{23}}-1}$ & $\re(1+b+s_{23})>0$\\

\hline
\begin{tabular}{c}
$\{1,2,3,4\}$\\
$\{2,4\}\{1\}\{3\}$\\
$\{1\}\{2\}\{3\}\{4\}$
\end{tabular}
& $\dfrac{(q-1)^2(q-2)}{q^3}\cdot\dfrac{1}{q^{1+b+s_{24}}-1}$ & $\re(1+b+s_{24})>0$\\

\hline
\begin{tabular}{c}
$\{1,2,3,4\}$\\
$\{3,4\}\{1\}\{2\}$\\
$\{1\}\{2\}\{3\}\{4\}$
\end{tabular}
& $\dfrac{(q-1)^2(q-2)}{q^3}\cdot\dfrac{1}{q^{1+b+s_{34}}-1}$ & $\re(1+b+s_{34})>0$
\end{tabular}
\end{center}

\item[(3)] There are three splitting chains $\spl\in\mathcal{S}_4$ with $\ptn_1=\{1,2\}\{3,4\}$:

\begin{center}
\def\arraystretch{1.2}
\begin{tabular}{c||c|c}
$\spl$ & $J_{\spl,q}(b,\bm{s})$ & $\mathcal{LP}_\spl(b)$\\
\hline\hline
\begin{tabular}{c}
$\{1,2,3,4\}$\\
$\{1,2\}\{3,4\}$\\
$\{1\}\{2\}\{3\}\{4\}$
\end{tabular}
& $\dfrac{(q-1)^3}{q^3}\cdot\dfrac{1}{q^{2+b+s_{12}+s_{34}}-1}$ & $\re(2+b+s_{12}+s_{34})>0$\\

\hline
\begin{tabular}{c}
$\{1,2,3,4\}$\\
$\{1,2\}\{3,4\}$\\
$\{1,2\}\{3\}\{4\}$\\
$\{1\}\{2\}\{3\}\{4\}$
\end{tabular}
&
\begin{tabular}{r}
$\dfrac{(q-1)^3}{q^3}\cdot\dfrac{1}{q^{2+b+s_{12}+s_{34}}-1}$\\
$\cdot\dfrac{1}{q^{1+b+s_{12}}-1}$
\end{tabular}
&
\begin{tabular}{c}
$\re(2+b+s_{12}+s_{34})>0$\\
and\\
$\re(1+b+s_{12})>0$
\end{tabular}\\

\hline
\begin{tabular}{c}
$\{1,2,3,4\}$\\
$\{1,2\}\{3,4\}$\\
$\{1\}\{2\}\{3,4\}$\\
$\{1\}\{2\}\{3\}\{4\}$
\end{tabular}
&
\begin{tabular}{r}
$\dfrac{(q-1)^3}{q^3}\cdot\dfrac{1}{q^{2+b+s_{12}+s_{34}}-1}$\\
$\cdot\dfrac{1}{q^{1+b+s_{34}}-1}$
\end{tabular}
&
\begin{tabular}{c}
$\re(2+b+s_{12}+s_{34})>0$\\
and\\
$\re(1+b+s_{34})>0$
\end{tabular}
\end{tabular}
\end{center}
There are also three $\spl\in\mathcal{S}_4$ satisfying $\ptn_1=\{1,3\}\{2,4\}$, and their table is obtained by simply transposing the indices $2$ and $3$ in the table above. Similarly, there are another three $\spl\in\mathcal{S}_4$ satisfying $\ptn_1=\{1,4\}\{2,3\}$. Thus there are nine distinct $\spl\in\mathcal{S}_4$ such that $\ptn_1$ has a pair of parts of size 2, but only three of them are reduced (the three that have length 2).

\item[(4)] Finally, we have only one $\spl\in\mathcal{S}_4$ with $\ptn_1=\{1\}\{2\}\{3\}\{4\}$. It is $\spl=(\{1,2,3,4\},~\{1\}\{2\}\{3\}\{4\})$, which is clearly reduced with $$J_{\spl,q}(b,\bm{s})=\dfrac{(q-1)(q-2)(q-3)}{q^3}\qquad\text{and}\qquad\mathcal{LP}_\spl(b)=\C^6~.$$
\end{itemize}

Combining all the level polytope conditions ``$\re(\cdot)>0$" from (1)-(4) with the root polytope condition $\re(3+a+b+\sum_{1\leq i<j\leq 4})>0$, we conclude that $\bm{s}\in\Omega_4(a,b)$ if and only if $\bm{s}$ simultaneously satisfies the conditions
\begin{align*}
\re(1+b+s_{ij})>0&\qquad\text{for}\qquad 1\leq i<j\leq 4~,\\
\re(2+b+s_{12}+s_{34})>0&~,\\
\re(2+b+s_{13}+s_{24})>0&~,\\
\re(2+b+s_{14}+s_{23})>0&~,\\
\re(2+b+s_{ij}+s_{ik}+s_{jk})>0&\qquad\text{for}\qquad 1\leq i<j<k\leq 4~,\qquad\text{and}\\
\re(3+a+b+s_{12}+s_{13}+s_{14}+s_{23}+s_{24}+s_{34})>0&~.
\end{align*} 
For every $\bm{s}$ in this region, the integral $Z_4^\rho(K,a,b,\bm{s})$ converges absolutely to
\begin{align*}
Z_4^\rho(K,a,b,\bm{s})&=\frac{1-q^{-(4+a+b+\sum_{1\leq i<j\leq 4}s_{ij})}}{1-q^{-(3+a+b+\sum_{1\leq i<j\leq 4}s_{ij})}}\cdot\sum_{m\in\Z}\rho(q^m)q^{m(4+a+b+\sum_{1\leq i<j\leq 4}s_{ij})}\\
\cdot\frac{1}{q^3}\Bigg\{&(q-1)(q-2)(q-3)+(q-1)^2(q-2)\Bigg[\frac{1}{q^{2+b+s_{12}+s_{13}+s_{23}}-1}+\frac{1}{q^{2+b+s_{12}+s_{14}+s_{24}}-1}\\
&~~~~~~~~~~~~~~~~~~~~~~~~~~~~~~~~~~~~~~~~~~~~~~~~+\frac{1}{q^{2+b+s_{13}+s_{14}+s_{34}}-1}+\frac{1}{q^{2+b+s_{23}+s_{24}+s_{34}}-1}\\
&~~~~~~~~~~~~~~~~~~~~~~~~~~~~~~~~~~~~~~~~~~~~~~~~+\dfrac{1}{q^{1+b+s_{12}}-1}+\dfrac{1}{q^{1+b+s_{13}}-1}+\dfrac{1}{q^{1+b+s_{14}}-1}\\
&~~~~~~~~~~~~~~~~~~~~~~~~~~~~~~~~~~~~~~~~~~~~~~~~+\dfrac{1}{q^{1+b+s_{12}}-1}+\dfrac{1}{q^{1+b+s_{13}}-1}+\dfrac{1}{q^{1+b+s_{14}}-1}\Bigg]\\
&~~~~~~~~~~~~~~~~~~~~~~~+(q-1)^3\Bigg[\dfrac{1}{q^{2+b+s_{12}+s_{34}}-1}\left(1+\dfrac{1}{q^{1+b+s_{12}}-1}+\dfrac{1}{q^{1+b+s_{34}}-1}\right)\\
&~~~~~~~~~~~~~~~~~~~~~~~~~~~~~~~~~~~+\dfrac{1}{q^{2+b+s_{13}+s_{24}}-1}\left(1+\dfrac{1}{q^{1+b+s_{13}}-1}+\dfrac{1}{q^{1+b+s_{24}}-1}\right)\\
&~~~~~~~~~~~~~~~~~~~~~~~~~~~~~~~~~~~+\dfrac{1}{q^{2+b+s_{14}+s_{23}}-1}\left(1+\dfrac{1}{q^{1+b+s_{14}}-1}+\dfrac{1}{q^{1+b+s_{23}}-1}\right)\Bigg]\\
&~~~+(q-1)^3\Bigg[\dfrac{1}{q^{2+b+s_{12}+s_{13}+s_{23}}-1}\left(\dfrac{1}{q^{1+b+s_{12}}-1}+\dfrac{1}{q^{1+b+s_{13}}-1}+\dfrac{1}{q^{1+b+s_{23}}-1}\right)\\
&~~~~~~~~~~~~~~~+\dfrac{1}{q^{2+b+s_{12}+s_{14}+s_{24}}-1}\left(\dfrac{1}{q^{1+b+s_{12}}-1}+\dfrac{1}{q^{1+b+s_{14}}-1}+\dfrac{1}{q^{1+b+s_{24}}-1}\right)\\
&~~~~~~~~~~~~~~~+\dfrac{1}{q^{2+b+s_{13}+s_{14}+s_{34}}-1}\left(\dfrac{1}{q^{1+b+s_{13}}-1}+\dfrac{1}{q^{1+b+s_{14}}-1}+\dfrac{1}{q^{1+b+s_{34}}-1}\right)\\
&~~~~~~~~~~~~~~~+\dfrac{1}{q^{2+b+s_{23}+s_{24}+s_{34}}-1}\left(\dfrac{1}{q^{1+b+s_{23}}-1}+\dfrac{1}{q^{1+b+s_{24}}-1}+\dfrac{1}{q^{1+b+s_{34}}-1}\right)\Bigg]\Bigg\}~.
\end{align*}

The terms inside $\{\dots\}$ are grouped to emphasize several facts. The first group $(q-1)(q-2)(q-3)$ vanishes unless $q\geq 4$. The second group $(q-1)^2(q-2)\big[\dots\big]$ vanishes unless $q\geq 3$. The third group $(q-1)^3\big[\dots\big]$ is nonzero for all $q\geq 2$ (and hence all $K$), but collapses from nine terms down to three by part (c) of \Cref{main} when $b=0$ (recall \Cref{Ex_N4}). The last group $(q-1)^3\big[\dots\big]$ is also nonzero for all $q\geq 2$ and corresponds to the 12 splitting chains of length 3 from case (1). 

\newpage
\noindent\textbf{Acknowledgements}: I would like to sincerely thank my advisor Chris Sinclair for all of the support, advice, and stimulating conversations that first inspired this work and led to many improvements.


\bibliography{references}
\bibliographystyle{amsalpha}

\begin{center}
\noindent\rule{4cm}{.5pt}
\vspace{.25cm}

\noindent {\sc \small Joe Webster}\\
{\small Department of Mathematics, University of Oregon, Eugene OR 97403} \\
email: {\tt jwebster@uoregon.edu}
\end{center}

\end{document}